\documentclass[letterpaper,11pt]{article}

\usepackage{amssymb,amsmath,amsthm}
\usepackage[margin=1in]{geometry}
\usepackage{thm-restate}
\usepackage{array}
\usepackage[rgb]{xcolor}
\usepackage{comment}
\usepackage{caption}
\usepackage{subcaption}
\usepackage{pifont}
\usepackage{anyfontsize}
\usepackage[shortlabels]{enumitem}
\usepackage{tikz}
\usepackage[bookmarksnumbered=true,hypertexnames=false,final]{hyperref}
\usepackage{algorithm, algpseudocode}
\usepackage[capitalize,sort]{cleveref}
\usetikzlibrary{decorations.pathreplacing,arrows.meta,calc,3d,perspective}

\hypersetup{
colorlinks,
linkcolor={red!50!black},
citecolor={red!50!black},
urlcolor={blue!50!black}
}

\newtheorem{theorem}{Theorem}

\newtheorem{lemma}[theorem]{Lemma}
\newtheorem{claim}[theorem]{Claim}
\newtheorem{definition}{Definition}

\makeatletter
\g@addto@macro{\UrlBreaks}{%
\do\/%
\do\a\do\b\do\c\do\d\do\e\do\f\do\g\do\h\do\i\do\j\do\k\do\l\do\m%
\do\n\do\o\do\p\do\q\do\r\do\s\do\t\do\u\do\v\do\w\do\x\do\y\do\z%
\do\A\do\B\do\C\do\D\do\E\do\F\do\G\do\H\do\I\do\J\do\K\do\L\do\M%
\do\N\do\O\do\P\do\Q\do\R\do\S\do\T\do\U\do\V\do\W\do\X\do\Y\do\Z%
\do\0\do\1\do\2\do\3\do\4\do\5\do\6\do\7\do\8\do\9%
}
\makeatother

\newcommand*{\wLoG}{w.l.o.g.}
\newcommand*{\WLoG}{W.l.o.g.}
\newcommand*{\Th}{^{\textrm{th}}}
\let\eps\varepsilon
\newcommand*{\defeq}{:=}

\newcommand*{\analogue}{analog}
\newcommand*{\Analogue}{Analog}

\newcommand*{\ceil}[1]{\left\lceil #1 \right\rceil}

\newcommand*{\smallceil}[1]{\lceil #1 \rceil}

\DeclareMathOperator*{\E}{E}

\DeclareMathOperator*{\argmin}{argmin}

\DeclareMathOperator{\opt}{opt}
\newcommand*{\optprog}[3]{
\begin{array}{*3{>{\displaystyle}l}}
#1 & \multicolumn{2}{>{\displaystyle}l}{#2}
#3 \end{array}}

\algnewcommand{\LineComment}[1]{\State \textcolor{gray}{\texttt{//} \textit{#1}}}

\algblockdefx{RepeatN}{EndRepeatN}[1]{\textbf{repeat} #1 \textbf{times}}{\textbf{end repeat}}
\makeatletter
\newcommand{\algmargin}{\the\ALG@thistlm}
\makeatother
\algnewcommand{\parState}[1]{\State\mbox{\hspace*{2em}%
\parbox[t]{\dimexpr\linewidth-\algmargin-2em}{\hspace*{-2em}\strut #1\strut}}}

\allowdisplaybreaks[1]
\newcommand*{\thmdep}[2]{}

\excludecomment{optional}

\newcommand*{\Null}{\texttt{null}}

\newcommand*{\nonrot}{^{\mathrm{nr}}}
\newcommand*{\rot}{^{\mathrm{r}}}
\newcommand*{\thin}{skewed}
\newcommand*{\Thin}{Skewed}
\newcommand*{\rankLemmaNote}{Rank Lemma:
the number of non-zero variables in an extreme-point solution
to a linear program is at most the number of
non-trivial constraints~\cite[Lemma 2.1.4]{iterative-methods}.}

\newcommand*{\Acal}{\mathcal{A}}

\newcommand*{\Ccal}{\mathcal{C}}

\newcommand*{\Hhat}{\widehat{H}}
\newcommand*{\Htild}{\widetilde{H}}

\newcommand*{\Itild}{\widetilde{I}}
\newcommand*{\Ihat}{\widehat{I}}
\newcommand*{\Jtild}{\widetilde{J}}
\newcommand*{\Jhat}{\widehat{J}}

\newcommand*{\Scal}{\mathcal{S}}

\newcommand*{\Tcal}{\mathcal{T}}
\newcommand*{\What}{\widehat{W}}
\newcommand*{\Wtild}{\widetilde{W}}
\newcommand*{\Xtild}{\widetilde{X}}

\newcommand*{\APoG}{\mathrm{APoG}}

\newcommand*{\WLarge}{W^{(L)}}
\newcommand*{\WSmall}{W^{(S)}}
\newcommand*{\WhatLarge}{\What^{(L)}}
\newcommand*{\HLarge}{H^{(L)}}
\newcommand*{\HSmall}{H^{(S)}}
\newcommand*{\HhatLarge}{\Hhat^{(L)}}

\newcommand*{\epsLarge}{\eps_1}
\newcommand*{\epsSmall}{\eps_2}
\newcommand*{\epsCont}{\eps_{\mathrm{cont}}}
\newcommand*{\nW}{n_W}
\newcommand*{\nH}{n_H}
\newcommand*{\nC}{n_C}
\newcommand*{\Qbest}{Q_{\mathrm{best}}}
\newcommand*{\Imed}{I_{\mathrm{med}}}

\DeclareMathOperator{\size}{size}

\newcommand*{\optSP}{\opt_{\mathrm{SP}}}
\DeclareMathOperator{\hsum}{hsum}
\DeclareMathOperator{\hmax}{hmax}
\DeclareMathOperator{\wsum}{wsum}
\DeclareMathOperator{\wmax}{wmax}

\DeclareMathOperator{\fopt}{fopt}
\DeclareMathOperator{\fcopt}{fcopt}
\DeclareMathOperator{\level}{level}

\DeclareMathOperator{\lingroupWide}{\mathtt{lingroupWide}}
\DeclareMathOperator{\lingroupTall}{\mathtt{lingroupTall}}

\DeclareMathOperator{\thinGPack}{\mathtt{skewed4Pack}}
\DeclareMathOperator{\greedyPack}{\mathtt{greedyPack}}

\DeclareMathOperator{\remMed}{\mathtt{removeMedium}}

\DeclareMathOperator{\thinCPack}{\mathtt{skewedCPack}}
\DeclareMathOperator{\iterPackings}{\mathtt{iterPackings}}
\DeclareMathOperator{\round}{\mathtt{round}}
\DeclareMathOperator{\greedyCPack}{\mathtt{greedyCPack}}
\DeclareMathOperator{\FP}{FP}

\newcommand*{\ropenInterval}[1]{[#1)}

\definecolor{myBlue}{HTML}{00b0f0}
\definecolor{myGreen}{HTML}{92d050}
\definecolor{myLightBlue}{HTML}{ccf1ff}
\definecolor{myLightGreen}{HTML}{e6f5d6}
\definecolor{shadedLightBlue}{HTML}{b4d9e4}
\definecolor{shadedLightGreen}{HTML}{ceddc0}

\title{Tight Approximation Algorithms for\\Geometric Bin Packing with Skewed Items}
\date{}

\author{
Arindam Khan\thanks{Indian Institute of Science, Bengaluru, India. {\tt arindamkhan@iisc.ac.in}.}
\and
Eklavya Sharma\thanks{Indian Institute of Science, Bengaluru, India. {\tt eklavyas@iisc.ac.in}.}
}

\begin{document}

\maketitle

\begin{abstract}
\noindent%
In the \emph{Two-dimensional Bin Packing (2BP)} problem,
we are given a set of rectangles of height and width at most one
and our goal is to find an axis-aligned nonoverlapping packing
of these rectangles into the minimum number of unit square bins.
The problem admits no APTAS and the current best approximation ratio
is $1.406$ by Bansal and Khan [SODA'14].
A well-studied variant of the problem is \emph{Guillotine Two-dimensional Bin Packing (G2BP)},
where all rectangles must be packed in such a way that
every rectangle in the packing can be obtained by
recursively applying a sequence of end-to-end axis-parallel cuts,
also called \emph{guillotine cuts}.
Bansal, Lodi, and Sviridenko [FOCS'05] obtained an APTAS for this problem.
Let $\lambda$ be the smallest constant such that for every set $I$ of items,
the number of bins in the optimal solution to G2BP for $I$
is upper bounded by $\lambda\opt(I) + c$,
where $\opt(I)$ is the number of bins in the optimal solution to 2BP for $I$ and $c$ is a constant.
It is known that $4/3 \le \lambda \le 1.692$.
Bansal and Khan [SODA'14] conjectured that $\lambda = 4/3$.
The conjecture, if true, will imply a $(4/3+\varepsilon)$-approximation algorithm for 2BP.
According to convention, for a given constant $\delta>0$,
a rectangle is \emph{large} if both its height and width are at least $\delta$,
and otherwise it is called \emph{skewed}.
We make progress towards the conjecture by showing $\lambda = 4/3$ for \emph{skewed instance},
i.e., when all input rectangles are skewed.
Even for this case, the previous best upper bound on $\lambda$ was roughly 1.692.
We also give an APTAS for 2BP for skewed instance,
though general 2BP does not admit an APTAS.

\end{abstract}

\section{Introduction}

Two-dimensional Bin Packing (2BP)  is a well-studied problem in combinatorial optimization.
It finds numerous applications in logistics, databases, and cutting stock. 
In 2BP, we are given a set of $n$ rectangular items and square bins of side length 1.
The $i\Th$ item is characterized by its width $w(i) \in (0,1]$ and height $h(i) \in (0,1]$.
Our goal is to find an axis-aligned nonoverlapping packing of
these items into the minimum number of square bins of side length 1.
There are  two well-studied variants: (i) where the items cannot be rotated, and
(ii) they can be rotated by 90 degrees.

As is conventional in bin packing, we focus on asymptotic approximation algorithm.
For any optimization problem, the asymptotic approximation ratio (AAR)
of algorithm $\Acal$ is defined as $\lim_{m \to \infty} \sup_{I: \opt(I) = m} ({\Acal(I)}/{\opt(I)})$, 
where $\opt(I)$ is the optimal objective value 
and $\Acal(I)$ is the objective value of the solution output by algorithm $\Acal$, respectively, on input $I$.
Intuitively,  AAR captures the algorithm's behavior
when $\opt(I)$ is large.
We call a bin packing algorithm $\alpha$-asymptotic-approximate iff its AAR is at most $\alpha$.
An Asymptotic Polynomial-Time Approximation Scheme (APTAS) is an algorithm
that accepts a parameter $\eps$ and has AAR of $(1+\eps)$.

2BP is a generalization of classical 1-D bin packing problem \cite{HobergR17, bp-aptas}. 
However, unlike 1-D bin packing, 2BP does not admit an APTAS unless P=NP \cite{bansal2006bin}.
In 1982, Chung, Garey, and Johnson~\cite{chung1982packing} gave an approximation algorithm
with AAR 2.125 for 2BP. Caprara~\cite{caprara2008} obtained
a $T_{\infty}(\approx 1.691)$-asymptotic-approximation algorithm.
Bansal, Caprara, and Sviridenko~\cite{rna} introduced the Round and Approx framework
to obtain an AAR of $1+\ln(T_{\infty})$ ($\approx 1.525$).
Then Jansen and Praedel~\cite{JansenP2013} obtained an AAR of 1.5.
The present best AAR is $1+\ln(1.5)$ ($\approx 1.405$),
due to Bansal and Khan~\cite{bansal2014binpacking},
and works for both the cases with and without rotations.
The best lower bounds on the AAR for 2BP are
1 + 1/3792 and 1 + 1/2196 \cite{chlebik2009hardness},
for the versions with and without rotations, respectively.

In the context of geometric packing, guillotine cuts are well-studied  and heavily used in practice \cite{sweeney1992cutting}. 
The notions of {\em guillotine cuts} and {\em $k$-stage packing} were introduced 
by Gilmore and Gomory in their seminal paper  \cite{gilmore1965multistage} on cutting stock problem. 
In $k$-stage packing each stage consists of either vertical or horizontal (but not both) axis-parallel end-to-end cuts, also called guillotine cuts. 
In each stage, each of the rectangular regions obtained on the previous stage is considered separately
and can be cut again by using guillotine cuts. In $k$-stage packing,
the minimum number of cuts to obtain each rectangle from the initial packing is at most $k$, plus an
additional cut to trim (i.e., separate the rectangles themselves from waste area).
Note that in the cutting process we change the orientation (vertical or horizontal) of the cuts $k-1$ times. 
2-stage packing, also called {\em shelf packing}, has been studied extensively. 
In {\em guillotine packing}, the packing of items in each bin should be \emph{guillotinable},
i.e., items have to be packed in alternate horizontal and vertical stages
but there is no limit on the number of stages that can be used.
See \cref{sec:guill-examples} for examples.
Caprara et al.~\cite{caprara2005fast} gave an APTAS for 2-stage 2BP.
Bansal et al.~\cite{bansal2005tale} showed an APTAS for guillotine 2BP.

The presence of an APTAS for guillotine 2BP raises an important question:
can the optimal solution to guillotine 2BP be used as a good approximate solution to 2BP?
Formally, let $\opt(I)$ and $\opt_g(I)$ be the minimum number of bins and the
minimum number of guillotinable bins, respectively, needed to pack items $I$.
Let $\lambda$ be the smallest constant such that for some constant $c$
and for every set $I$ of items, we get $\opt_g(I) \le \lambda\opt(I) + c$.
Then $\lambda$ is called the Asymptotic Price of Guillotinability (APoG).
It is easy to show that $\APoG \ge 4/3$\footnote{Consider
a set $I$ of items containing $2m$ rectangles of width 0.6 and height 0.4 and
$2m$ rectangles of width 0.4 and height 0.6.
Then $\opt(I) = m$ and $\opt_g(I) = \ceil{4m/3}$.}.
Bansal and Khan~\cite{bansal2014binpacking} conjectured that $\APoG = 4/3$.
If true, this would imply a $(4/3+\eps)$-asymptotic-approximation algorithm for 2BP.
However, the present upper bound on APoG is only $T_\infty$ ($\approx1.691$),
due to Caprara's HDH algorithm~\cite{caprara2008} for 2BP, which produces a 2-stage packing.

APTASes are known for some special cases for 2BP,
such as when all items are squares~\cite{bansal2006bin} or
when all rectangles are small in both dimensions
\cite{coffman1980performance} (see \cref{thm:nfdh-small-2} in \cref{sec:nfdh}).
Another important class is {\em skewed} rectangles.
We say that a rectangle is $\delta$-large if,
for some constant $\delta>0$, its width and height are more than $\delta$;
otherwise, the rectangle is $\delta$-skewed.
We just say that a rectangle is large or skewed when $\delta$ is clear from the context.
An instance of 2BP is skewed if all the rectangles in the input are skewed.
Skewed instances are important in geometric packing (see Section \ref{subs:prior}).
This special case is practically relevant~\cite{galvez2020tight}:
e.g., in scheduling, it captures scenarios where
no job can consume a significant amount of a shared resource
(energy, memory space, etc.) for a significant amount of time.
Even for skewed instance for 2BP, the best known AAR is 1.406 \cite{bansal2014binpacking}.
Also, for skewed instance, the best known upper bound on APoG is $T_\infty \approx 1.691$.

\subsection{Related Works}
\label{subs:prior}
Multidimensional packing and covering problems are fundamental in combinatorial optimization \cite{CKPT17}. 
Vector packing (VP) is another variant of bin packing, where the input is a set of vectors in $[0, 1]^d$ and the
goal is to partition the vectors into the minimum number of parts (bins) such that in each part, the sum of 
vectors is at most 1 in every coordinate. 
The present best approximation algorithm attains an AAR of $(0.807+\ln(d+1))$ \cite{bansal2016improved}
and there is a matching $\Omega(\ln d)$-hardness \cite{sandeep2021optimal}.  Generalized multidimensional 
packing \cite{aco-gvbp, aco-gvks} generalizes both geometric and vector packing. 

In two-dimensional strip packing (2SP) \cite{coffman1980performance, steinberg1997strip}, we are given a 
set of rectangles and a  bounded width strip. The goal is to obtain an axis-aligned nonoverlapping packing 
of all rectangles such that the height of the packing is minimized.
The best-known approximation ratio for SP is $5/3+\eps$ \cite{harren20115} and it is NP-hard to obtain 
better than 3/2-approximation. 
However, there exist APTASes for the problem, for both the cases
with and without rotations \cite{kenyon1996strip, jansen2005strip}.
In two-dimensional knapsack (2GK) \cite{jansen2004rectangle}, the rectangles have associated profits and 
our goal is to pack the maximum profit subset into a unit square knapsack.
The present best polynomial-time (resp.~pseudopolynomial-time) approximation ratio
for 2GK is 1.809~\cite{galvez2017approximating} (resp.\ 4/3~\cite{GalSocg21}).
These geometric packing problem are studied for $d$-dimensions ($d\ge 2$) \cite{eku-hdhk} as well.

Both 2SP and 2GK are also well-studied under guillotine packing. 
Seiden and Woeginger~\cite{seiden2005two} gave an APTAS for guillotine 2SP.
Khan et al.~\cite{KhanSocg21} have recently given a pseudopolynomial-time
approximation scheme for guillotine 2GK.
Recently, guillotine cuts \cite{pach2000cutting} have received attention due to
their connection with the maximum independent set of rectangles (MISR) problem~\cite{AdamaszekHW19}.
In MISR, we are given a set of possibly overlapping rectangles and the goal is to find the maximum cardinality set of rectangles so that there is
no pairwise overlap. It was noted in  \cite{khan2020guillotine, abed2015guillotine}  that for any set of $n$ non-overlapping axis-parallel rectangles,  if there is a guillotine cutting sequence 
separating $\alpha n$ of them, then it implies a $1/\alpha$-approximation for MISR. 

Skewed instance is an important special case in these problems.
In some problems, such as MISR and 2GK, if all items are $\delta$-large then we can solve them exactly in polynomial time.
So, the inherent difficulty of these problems lies in skewed instances.
For VP, hard instances are again skewed, e.g.,
Bansal, Eli\'a\v{s} and Khan~\cite{bansal2016improved} showed that
hard instances for 2-D VP (for a class of algorithms called {\em rounding based algorithms})
are skewed instances, where one dimension is $1-\eps$ and the other dimension is $\eps$.  
Galvez el al.~\cite{galvez2020tight} recently studied strip packing when all items are skewed.
For skewed instances, they showed $(3/2-\eps)$ hardness of approximation and a matching $(3/2+\eps)$-approximation algorithm. 
For 2GK, when the height of each item is at most $\eps^3$,
a $(1-72\eps)$-approximation algorithm is known~\cite{fishkin2005efficient}.

\subsection{Our Contributions}

We study 2BP for the special case of $\delta$-\thin{} rectangles,
where $\delta \in (0, 1/2]$ is a constant.

First, we make progress towards the conjecture \cite{bansal2014binpacking} that $\APoG = 4/3$.
Even for \thin{} rectangles, we only knew $4/3 \le \APoG \le T_{\infty}(\approx 1.691)$.
We resolve the conjecture for \thin{} rectangles,  by giving lower and upper bounds of roughly $4/3$ 
when $\delta$ is a small constant.

Specifically, we give an algorithm for 2BP, called $\thinGPack_{\eps}$,
that takes a parameter $\eps \in (0, 1)$ as input.
For a set $I$ of $\delta$-\thin{} rectangles, we show that when
$\delta$ and $\eps$ are close to 0, $\thinGPack_{\eps}(I)$ outputs a
4-stage packing of $I$ into roughly $4\opt(I)/3 + O(1)$ bins.

\begin{restatable}{theorem}{rthmThinGPack}
\label{thm:thin-gpack}
Let $I$ be a set of $\delta$-skewed items, where $\delta \in (0, 1/2]$.
Then $\thinGPack_{\eps}(I)$ outputs a 4-stage packing of $I$
in time $O((1/\eps)^{O(1/\eps)} + n\log n)$.
Furthermore, the number of bins used is at most 
$({4}/{3})(1+8\delta)(1+7\eps)\opt(I) + ({8}/{\eps^2}) + 30$.
\end{restatable}

A tighter analysis shows that
when $\delta \le 1/16$ and $\eps \le 10^{-4}$,
then $\thinGPack$ has AAR $(76/45)(1+7\eps) < T_{\infty}$,
which improves upon the best-known bound on APoG for the general case.

The lower bound of $4/3$ on APoG can be extended to \thin{} items.
We formally prove this in \cref{sec:apog-lb}.
Hence, our bounds on APoG are tight for \thin{} items.
Our result indicates that to improve the bounds for APoG in the general case,
we should focus on $\delta$-large items.

Our bounds on APoG also hold when items can be rotated.
See \cref{sec:guill-rot} for details.

Our other main result is an APTAS for 2BP for \thin{} items. 
Formally, we give an algorithm for 2BP, called $\thinCPack$,
and we show that for every constant $\eps \in (0, 1)$,
there exists a constant $\delta \in (0, \eps)$ such that the algorithm has an AAR of $1+\eps$
when all items in the input are $\delta$-\thin{} rectangles.
$\thinCPack$ can also be extended to the case where items can be rotated. %

The best-known AAR for 2BP is $1 + \ln(1.5) + \eps$.
Our result indicates that to improve upon algorithms for 2BP,
one should focus on $\delta$-large items.

In \cref{sec:guill-thin}, we describe the $\thinGPack$ algorithm
and prove \cref{thm:thin-gpack}.
In \cref{sec:thin-bp}, we describe the $\thinCPack$ algorithm
and prove that it has an AAR of $1+\eps$.

\section{Preliminaries}
\label{sec:preliminaries}

Let $[n] \defeq \{1, 2, \ldots, n\}$, for $n \in \mathbb{N}$.
For a rectangle $i$, its area $a(i) \defeq w(i)h(i)$.
For a set $I$ of rectangles, let $a(I) \defeq \sum_{i \in I} a(i)$.
An \emph{axis-aligned packing} of an item $i$ in a bin
is specified by a pair $(x(i), y(i))$, where $x(i), y(i) \in [0,1]$,
so that $i$ is placed in the region
$[x(i), x(i)+w(i)] \times [y(i), y(i)+h(i)]$.
A packing of rectangles in a bin is called \emph{nonoverlapping} iff for any two
distinct items $i$ and $j$, the rectangles
$(x(i), x(i)+w(i)) \times (y(i), y(i)+h(i))$ and
$(x(j), x(j)+w(j)) \times (y(j), y(j)+h(j))$ are disjoint.
Equivalently, items may only intersect at their boundaries.\\

\noindent \textbf{Next-Fit Decreasing Height (NFDH):}
The NFDH algorithm~\cite{coffman1980performance}
is a simple algorithm for 2SP and 2BP. We will use the following results on NFDH.
We defer the proofs to \cref{sec:nfdh}.

\begin{restatable}{lemma}{rthmNfdhSmall}
\label{thm:nfdh-small}
Let $I$ be a set of items where each item $i$ has $w(i) \le \delta_W$
and $h(i) \le \delta_H$. Let there be a bin of width $W$ and height $H$.
If $a(I) \le (W - \delta_W)(H - \delta_H)$, then NFDH can pack $I$ into the bin.
\end{restatable}

\begin{lemma}
\label{thm:nfdh-wide-tall}
\label{thm:nfdh-tall}
\label{thm:nfdh-wide}
Let $I$ be a set of rectangular items. Then NFDH uses less than
$(2a(I)+1)/(1-\delta)$ bins to pack $I$ when $h(i) \le \delta$ for each item $i$
and $2a(I)/(1-\delta) + 3$ bins when $w(i) \le \delta$ for each item $i$.
\end{lemma}

If we swap the coordinate axes in NFDH, we get the
Next-Fit Decreasing Width (NFDW) algorithm.
\Analogue{}s of the above results hold for NFDW.\\

\noindent \textbf{Slicing Items:}
We will consider variants of 2BP where some items can be \emph{sliced}.
Formally, slicing a rectangular item $i$ using a horizontal cut is the operation of
replacing $i$ by two items $i_1$ and $i_2$ such that
$w(i) = w(i_1) = w(i_2)$ and $h(i) = h(i_1) + h(i_2)$.
Slicing using vertical cut is defined analogously.
Allowing some items to be sliced may reduce the number of bins required to pack.
See \cref{fig:frac-pack} for an example.

\begin{figure}[htb]
\centering
\begin{tikzpicture}[
myarrow/.style = {->,>={Stealth},semithick},
mybrace/.style = {decoration={amplitude=3pt,brace,mirror,raise=1pt},semithick,decorate},
every node/.style = {scale=0.8},
scale=0.8
]
\draw (0,0) rectangle +(3,3);
\draw[fill={black!30}] (0,0) rectangle +(1.2,3);
\draw[fill={black!10}] (3,0) rectangle +(-1.5,1.2);
\draw[fill={black!10}] (3,1.2) rectangle +(-1.5,1.2);
\draw[mybrace] (0,0) -- node[below=1pt] {0.4} +(1.2,0);
\draw[mybrace] (1.5,0) -- node[below=1pt] {0.5} +(1.5,0);
\draw[mybrace] (3,0) -- node[right=2pt] {0.4} +(0,1.2);
\draw[mybrace] (3,1.2) -- node[right=2pt] {0.4} +(0,1.2);

\draw[fill={black!30}] (-8,0) rectangle +(1.2,3);
\path (-8,0) -- node[below=0pt] {0.4} +(1.2,0);
\path (-8,0) -- node[left=0pt] {1} +(0,3);
\node at (-6.2,1.5) {+};
\draw[fill={black!10}] (-5,1) rectangle +(3,1.2);
\path (-5,1) -- node[left=0pt] {0.4} +(0,1.2);
\draw[dashed] (-3.5,2.5) -- (-3.5,0.5);
\node[rotate=90,transform shape] at (-3.5,0.3) {\large\ding{34}};
\draw[mybrace] (-5,1) -- node[below=1pt] {0.5} +(1.5,0);
\draw[mybrace] (-3.5,1) -- node[below=1pt] {0.5} +(1.5,0);
\draw[myarrow] (-1.5,1.5) -- (-0.5,1.5);
\end{tikzpicture}

\caption{Packing two items into a bin, where one item is sliced using a vertical cut.
If slicing were forbidden, two bins would be required.}
\label{fig:frac-pack}
\end{figure}
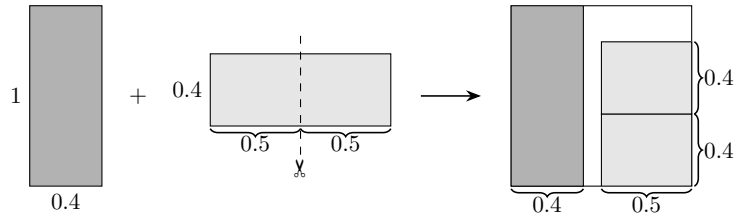

Alamdari et al.~\cite{alamdari2013smart} explored algorithms for
a variant of 2SP where items can be sliced using vertical cuts,
which has applications in smart-grid electricity allocation.
Many packing algorithms \cite{kenyon1996strip,JansenP2013,bansal2005tale}
solve the sliceable version of the problem as a subroutine.

\section{Guillotinable Packing of \Thin{} Rectangles}
\label{sec:guill-thin}

An item is called $(\delta_W, \delta_H)$-\thin{} iff its width is at most $\delta_W$
or its height is at most $\delta_H$.
In this section, we consider the problem of obtaining tight upper and lower bounds
on APoG for $(\delta_W, \delta_H)$-\thin{} items.
We will describe the $\thinGPack$ algorithm and prove \cref{thm:thin-gpack}.

\subsection{Packing With Slicing}

Before describing $\thinGPack$,
let us first look at a closely-related variant of this problem,
called the \emph{sliceable 2D bin packing problem}, denoted as S2BP.
In this problem, we are given two sets of rectangular items, $\Wtild$ and $\Htild$, where
items in $\Wtild$ have width more than $1/2$, and items in $\Htild$ have height more than $1/2$.
$\Wtild$ is called the set of wide items and $\Htild$ is called the set of tall items.
We are allowed to \emph{slice} items in $\Wtild$ using horizontal cuts
and slice items in $\Htild$ using vertical cuts, and our task is to pack
$\Wtild \cup \Htild$ into the minimum number of bins without rotating the items.
See \cref{fig:bp-vs-sbp} for an example that illustrates the difference
between 2BP and S2BP.

\begin{figure}[htb]
\begin{subfigure}{0.45\textwidth}
\centering
\begin{tikzpicture}[
witem/.style={draw,fill={black!30}},
hitem/.style={draw,fill={black!10}},
bin/.style={draw,thick},
myarrow/.style={->,>={Stealth},thick},
scale=0.7,
]
\begin{scope}
\node at (0.3, 2.0) {$\Wtild$:};
\node at (0.3, 0.6) {$\Htild$:};
\path[hitem] (1.0, 0.0) rectangle +(1.2, 1.2);
\path[hitem] (2.4, 0.0) rectangle +(1.2, 1.2);
\path[witem] (1.0, 1.4) rectangle +(1.2, 1.2);
\path[witem] (2.4, 1.4) rectangle +(1.2, 1.2);
\end{scope}
\draw[myarrow] (3.9, 1.3) -- (5.2, 1.3);
\begin{scope}[xshift={5.5cm},yshift={-0.8cm}]
\path[hitem] (0.0, 0.0) rectangle +(1.2, 1.2);
\path[hitem] (2.2, 0.0) rectangle +(1.2, 1.2);
\path[witem] (0.0, 2.2) rectangle +(1.2, 1.2);
\path[witem] (2.2, 2.2) rectangle +(1.2, 1.2);
\path[bin] (0.0, 0.0) rectangle +(2, 2);
\path[bin] (2.2, 0.0) rectangle +(2, 2);
\path[bin] (0.0, 2.2) rectangle +(2, 2);
\path[bin] (2.2, 2.2) rectangle +(2, 2);
\end{scope}
\end{tikzpicture}

\caption{Packing items into 4 bins without slicing.}
\end{subfigure}
\hfil
\begin{subfigure}{0.5\textwidth}
\centering
\begin{tikzpicture}[
witem/.style={draw,fill={black!30}},
hitem/.style={draw,fill={black!10}},
bin/.style={draw,thick},
myarrow/.style={->,>={Stealth},thick},
cutline/.style={draw={black!50!red},dashed,semithick},
scale=0.7,
]
\begin{scope}
\node at (0.3, 2.0) {$\Wtild$:};
\node at (0.3, 0.6) {$\Htild$:};
\path[hitem] (1.0, 0.0) rectangle +(1.2, 1.2);
\path[hitem] (2.4, 0.0) rectangle +(1.2, 1.2);
\path[witem] (1.0, 1.4) rectangle +(1.2, 1.2);
\path[witem] (2.4, 1.4) rectangle +(1.2, 1.2);
\path[cutline] (2.3, 2.0) -- (3.7, 2.0);
\node[xscale=-1,transform shape] at (3.9, 1.98) {\large\ding{34}};
\path[cutline] (3.0, -0.1) -- (3.0, 1.3);
\node[rotate=90,transform shape] at (3.02, -0.3) {\large\ding{34}};
\end{scope}
\draw[myarrow] (3.9, 1.3) -- (5.2, 1.3);
\begin{scope}[xshift={5.5cm},yshift={0.3cm}]
\path[witem] (0.0, 0.0) rectangle +(1.2, 1.2);
\path[hitem] (2.2, 0.0) rectangle +(1.2, 1.2);
\path[witem] (0.0, 1.2) rectangle +(1.2, 0.6);
\path[witem] (2.2, 1.2) rectangle +(1.2, 0.6);
\path[hitem] (1.2, 0.0) rectangle +(0.6, 1.2);
\path[hitem] (3.4, 0.0) rectangle +(0.6, 1.2);
\path[bin] (0.0, 0.0) rectangle +(2, 2);
\path[bin] (2.2, 0.0) rectangle +(2, 2);
\end{scope}
\end{tikzpicture}

\caption{Packing items into 2 bins by horizontally slicing an item in $\Wtild$
and vertically slicing an item in $\Htild$.}
\end{subfigure}
\caption[2BP vs.~S2BP]{Example illustrating 2BP vs.~S2BP.
There are 2 wide items ($\Wtild$) and 2 tall items ($\Htild$).
The items are squares of side length 0.6 and the bins are squares of side length 1.}
\label{fig:bp-vs-sbp}
\end{figure}
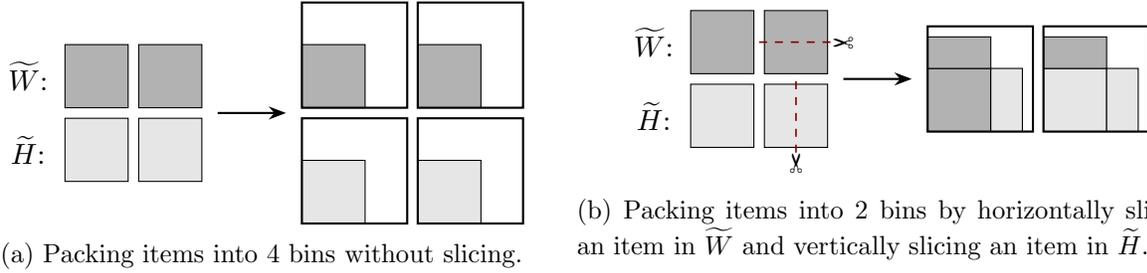

We first describe a simple $4/3$-asymptotic-approximation algorithm
for S2BP, called $\greedyPack$, that outputs a 2-stage packing.
Later, we will show how to use $\greedyPack$ to design $\thinGPack$.

We assume that the bin is a square of side length 1. Since we can slice items,
we allow items in $\Wtild$ to have height more than 1
and items in $\Htild$ to have width more than 1.

For $X \subseteq \Wtild$, $Y \subseteq \Htild$, define $\hsum(X) \defeq \sum_{i \in X} h(i)$; 
$\wsum(Y) \defeq \sum_{i \in Y} w(i)$;
 $\wmax(X) \defeq \max_{i \in X} w(i) \textrm{ if } X \neq \emptyset$, and  $0  \textrm{ if } X = \emptyset$;
$\hmax(Y) \defeq \max_{i \in Y} h(i) \textrm{ if } Y \neq \emptyset$, and  $0 \textrm{ if } Y = \emptyset$.

In the algorithm $\greedyPack(\Wtild, \Htild)$, we first sort items $\Wtild$ in decreasing order
of width and sort items $\Htild$ in decreasing order of height.
Suppose $\hsum(\Wtild) \ge \wsum(\Htild)$. Let $X$ be the largest prefix of $\Wtild$
of total height at most 1, i.e., if $\hsum(\Wtild) > 1$,
then $X$ is a prefix of $\Wtild$ such that $\hsum(X) = 1$ (slice items if needed),
and $X = \Wtild$ otherwise.
Pack $X$ into a bin such that the items touch the right edge of the bin.
Then we pack the largest possible prefix of $\Htild$
into the empty rectangular region of width $1 - \wmax(X)$ in the left side of the bin.
We call this a type-1 bin. See \cref{fig:greedy-pack:1} for an example.
If $\hsum(\Wtild) < \wsum(\Htild)$, we proceed analogously in a coordinate-swapped way,
i.e., we first pack tall items in the bin and then pack wide items in the remaining space.
Call this bin a type-2 bin.
We pack the rest of the items into bins in the same way.

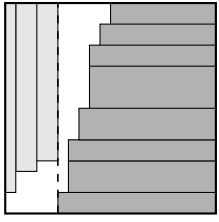
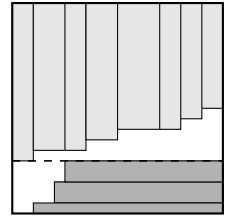
\begin{figure}[htb]
\begin{subfigure}{0.45\textwidth}
    \centering
    \tikzset{mytransform/.style={scale=0.7}}
    \tikzset{wItem/.style={draw,fill={black!30}}}
    \tikzset{hItem/.style={draw,fill={black!10}}}
    \ifcsname pGameL\endcsname\else\newlength{\pGameL}\fi
\setlength{\pGameL}{0.2cm}
\tikzset{bin/.style={draw,thick}}
\begin{tikzpicture}[mytransform]
\path[wItem] (5\pGameL, 0\pGameL) rectangle +(15\pGameL, 2\pGameL);
\path[wItem] (6\pGameL, 2\pGameL) rectangle +(14\pGameL, 3\pGameL);
\path[wItem] (6\pGameL, 5\pGameL) rectangle +(14\pGameL, 2\pGameL);
\path[wItem] (7\pGameL, 7\pGameL) rectangle +(13\pGameL, 3\pGameL);
\path[wItem] (8\pGameL, 10\pGameL) rectangle +(12\pGameL, 4\pGameL);
\path[wItem] (8\pGameL, 14\pGameL) rectangle +(12\pGameL, 2\pGameL);
\path[wItem] (9\pGameL, 16\pGameL) rectangle +(11\pGameL, 2\pGameL);
\path[wItem] (10\pGameL, 18\pGameL) rectangle +(10\pGameL, 2\pGameL);
\path[hItem] (1\pGameL, 4\pGameL) rectangle +(2\pGameL, 16\pGameL);
\path[hItem] (0\pGameL, 2\pGameL) rectangle +(1\pGameL, 18\pGameL);
\path[hItem] (3\pGameL, 5\pGameL) rectangle +(2\pGameL, 15\pGameL);
\draw[semithick,dashed] (5\pGameL, 0\pGameL) -- +(0\pGameL, 20\pGameL);
\path[bin] (0\pGameL, 0\pGameL) rectangle (20\pGameL, 20\pGameL);
\end{tikzpicture}

    \caption{A type-1 bin produced by $\greedyPack$.
Wide items are packed on the right. Tall items are packed on the left.}%
\label{fig:greedy-pack:1}
\end{subfigure}
\hfill
\begin{subfigure}{0.45\textwidth}
    \centering
    \tikzset{mytransform/.style={xscale=-0.7,yscale=0.7,rotate=90}}
    \tikzset{wItem/.style={draw,fill={black!10}}}
    \tikzset{hItem/.style={draw,fill={black!30}}}
    \ifcsname pGameL\endcsname\else\newlength{\pGameL}\fi
\setlength{\pGameL}{0.2cm}
\tikzset{bin/.style={draw,thick}}
\begin{tikzpicture}[mytransform]
\path[wItem] (5\pGameL, 0\pGameL) rectangle +(15\pGameL, 2\pGameL);
\path[wItem] (6\pGameL, 2\pGameL) rectangle +(14\pGameL, 3\pGameL);
\path[wItem] (6\pGameL, 5\pGameL) rectangle +(14\pGameL, 2\pGameL);
\path[wItem] (7\pGameL, 7\pGameL) rectangle +(13\pGameL, 3\pGameL);
\path[wItem] (8\pGameL, 10\pGameL) rectangle +(12\pGameL, 4\pGameL);
\path[wItem] (8\pGameL, 14\pGameL) rectangle +(12\pGameL, 2\pGameL);
\path[wItem] (9\pGameL, 16\pGameL) rectangle +(11\pGameL, 2\pGameL);
\path[wItem] (10\pGameL, 18\pGameL) rectangle +(10\pGameL, 2\pGameL);
\path[hItem] (1\pGameL, 4\pGameL) rectangle +(2\pGameL, 16\pGameL);
\path[hItem] (0\pGameL, 2\pGameL) rectangle +(1\pGameL, 18\pGameL);
\path[hItem] (3\pGameL, 5\pGameL) rectangle +(2\pGameL, 15\pGameL);
\draw[semithick,dashed] (5\pGameL, 0\pGameL) -- +(0\pGameL, 20\pGameL);
\path[bin] (0\pGameL, 0\pGameL) rectangle (20\pGameL, 20\pGameL);
\end{tikzpicture}

    \caption{A type-2 bin produced by $\greedyPack$.
Tall items are packed above. Wide items are packed below.}%
\label{fig:greedy-pack:2}
\end{subfigure}
\caption{Examples of type-1 and type-2 bins produced by $\greedyPack$.}
\label{fig:greedy-pack}
\end{figure}

\begin{claim}
\label{thm:greedy-pack}
$\greedyPack(\Wtild, \Htild)$ outputs a 2-stage packing of $\Wtild \cup \Htild$.
It runs in $O(m + |\Wtild|\log|\Wtild| + |\Htild|\log|\Htild|)$ time,
where $m$ is the number of bins used.
Furthermore, it slices items in $\Wtild$ by making at most $m-1$ horizontal cuts
and slices items in $\Htild$ by making at most $m-1$ vertical cuts.
\end{claim}

Since items in $\Wtild$ have width more than $1/2$,
no two items can be placed side-by-side.
Hence, $\smallceil{\hsum(\Wtild)} = \opt(\Wtild) \le \opt(\Wtild \cup \Htild)$.
Similarly, $\smallceil{\wsum(\Htild)} \le \opt(\Wtild \cup \Htild)$.
So, if all bins have the same type, $\greedyPack$ uses
$\max(\smallceil{\hsum(\Wtild)}, \smallceil{\wsum(\Htild)}) = \opt(\Wtild \cup \Htild)$ bins.
We will now focus on the case where
some bins have type 1 and some have type 2.

\begin{definition}
In a type-1 bin, let $X$ and $Y$ be the wide and tall items, respectively.
The bin is called \emph{full} iff $\hsum(X) = 1$ and $\wsum(Y) = 1 - \wmax(X)$.
Define fullness for type-2 bins analogously.
\end{definition}

We first show that full bins pack items of a large total area,
and then we show that if some bins have type 1 and some bins have type 2,
then there can be at most 2 non-full bins.
This will help us get an upper-bound on the number of bins used by $\greedyPack(\Wtild, \Htild)$
in terms of $a(\Wtild \cup \Htild)$.

\begin{lemma}
\label{thm:area-bound}
Let there be $m_1$ type-1 full bins.
Let $J_1$ be the items in them.
Then $m_1 \le 4a(J_1)/3 + 1/3$.
\end{lemma}
\begin{proof}
In the $j\Th$ full bin of type 1, let $X_j$ be the items from $\Wtild$
and $Y_j$ be the items from $\Htild$. Let
$\ell_j \defeq \wmax(X_j)  \textrm{ if } j \le m_1$
 and $\ell_{m_1+1} \defeq 1/2$. 
Since all items have their larger dimension more than $1/2$,
$\ell_j \ge 1/2$ and $\hmax(Y_j) > 1/2$, for any $j \in [m_1]$.

$a(X_j) \ge \ell_{j+1}$, since $X_j$ has height 1 and width at least $\ell_{j+1}$.
$a(Y_j) \ge (1-\ell_j)/2$, since $Y_j$ has width $1 - \ell_j$ and height more than $1/2$.
Therefore,
$a(J_1) = \sum_{j=1}^{m_1} (a(X_j) + a(Y_j))
\ge \sum_{j=1}^{m_1} (\ell_{j+1} + (1-\ell_j)/2)
\ge \sum_{j=1}^{m_1} \left(({\ell_{j+1}}/{2}) + ({1}/{4})
    + ({1}/{2}) - ({\ell_j}/{2})\right)  
= ({3m_1}/{4}) + ({1}/{4}) - ({\ell_1}/{2})
\ge ({3m_1-1}/{4})$.\\
In the above inequalities, we used that $\ell_{j+1} \ge 1/2$ 
 and $\ell_1 \le 1$.

Therefore, $m_1 \le 4a(J_1)/3 + 1/3$.
\end{proof}

An \analogue{} of \cref{thm:area-bound} can be proven for type-2 bins.
Note that \cref{thm:area-bound} implies that the average area of full bins is close to $3/4$.
It is possible for an individual full bin to have area close to 1/2,
but the number of such bins is small, due to the telescopic sum in \cref{thm:area-bound}.

Let $m$ be the number of bins used by $\greedyPack(\Wtild, \Htild)$.
After $j$ bins have been packed, let $A_j$ be the height of the remaining items in $\Wtild$
and $B_j$ be the width of the remaining items in $\Htild$.
Let $t_j$ be the type of the $j\Th$ bin (1 for type-1 bin and 2 for type-2 bin).
So $t_j = 1 \iff A_{j-1} \ge B_{j-1}$.

We first show that $|A_{j-1} - B_{j-1}| \le 1 \implies |A_j - B_j| \le 1$.
This means that once the difference between $\hsum(\Wtild)$ and $\wsum(\Htild)$ becomes at most 1,
it continues to stay at most 1.
Next, we show that $t_j \neq t_{j+1} \implies |A_{j-1} - B_{j-1}| \le 1$.
This means that if some bins have type 1 and some have type 2,
then the difference between $\hsum(\Wtild)$ and $\wsum(\Htild)$ will eventually become at most 1.
In the first non-full bin, we will use up all the wide items or the tall items.
We will show that the remaining items have total height or total width at most 1,
so we can have at most 1 more non-full bin.
Hence, there can be at most 2 non-full bins when we have both type-1 and type-2 bins.

In the $j\Th$ bin, let $a_j$ be the height of items from $\Wtild$
and $b_j$ be the width of items from $\Htild$.
Hence, for all $j \in [m]$,
$A_{j-1} = A_j + a_j$ and $B_{j-1} = B_j + b_j$.

\begin{lemma}
\label{thm:diff-capture}
$|A_{j-1} - B_{j-1}| \le 1 \implies |A_j - B_j| \le 1$.
\end{lemma}
\begin{proof}
\WLoG{}, assume $A_{j-1} \ge B_{j-1}$. So, $t_j = 1$. Suppose $a_j < b_j$.
Then $a_j < 1$, so we used up $\Wtild$ in the $j\Th$ bin. Therefore,
$A_j = 0 \implies A_{j-1} = a_j < b_j \le b_j + B_j = B_{j-1}$,
which contradicts. Hence, $a_j \ge b_j$.
As $0 \le (A_{j-1} - B_{j-1}), (a_j - b_j) \le 1$, we get
$A_j - B_j = (A_{j-1} - B_{j-1}) - (a_j - b_j) \in [-1, 1]$.
\end{proof}

\begin{lemma}
\label{thm:tdiff-implies-adiff}
$t_j \neq t_{j+1} \implies |A_{j-1} - B_{j-1}| \le 1$.
\end{lemma}
\begin{proof}
\WLoG{}, assume $t_j = 1$ and $t_{j+1} = 2$. Then
\[ A_{j-1} \ge B_{j-1} \textrm{ and } A_j < B_j
\implies B_{j-1} \le A_{j-1} < B_{j-1} + a_j - b_j
\implies A_{j-1} - B_{j-1} \in \ropenInterval{0, 1}. \qedhere \]
\end{proof}

\begin{lemma}
\label{thm:non-full-ub}
If all bins don't have the same type, then there can be at most 2 non-full bins.
\end{lemma}
\begin{proof}
Let there be $p$ full bins. %
Assume \wLoG{} that in the $(p+1)\Th$ bin, we used up all items from $\Wtild$ but not $\Htild$.
Hence, $A_{p+1} = 0$ and $\forall i \ge p+2$, $t_i = 2$.
Since all bins don't have the same type, $\exists k \le p+1$ such that
$t_k = 1$ and $t_{k+1} = 2$.
By \cref{thm:tdiff-implies-adiff,thm:diff-capture}, we get $|A_{p+1} - B_{p+1}| \le 1$,
implying  $B_{p+1} \le 1$.
Hence, the $(p+1)\Th$ bin will use up all tall items,
implying at most 2 non-full bins.
\end{proof}

\begin{theorem}
\label{thm:greedy-pack-bins}
The number of bins $m$ used by $\greedyPack$ is at most
\\ $\max\left(\smallceil{\hsum(\Wtild)}, \smallceil{\wsum(\Htild)},
\frac{4}{3}a(\Wtild \cup \Htild) + \frac{8}{3}\right)$.
\end{theorem}
\begin{proof}
If all bins have the same type, then $m \le \max(\smallceil{\hsum(\Wtild)}, \smallceil{\wsum(\Htild)})$.

Let there be $m_1$ (resp.~$m_2$) full bins of type 1 (resp.~type 2)
and let $J_1$ (resp.~$J_2$) be the items inside those bins.
Then by \cref{thm:area-bound}, we get $m_1 \le 4a(J_1)/3 + 1/3$ and $m_2 \le 4a(J_2)/3 + 1/3$.
Hence, $m_1 + m_2 \le 4a(\Wtild \cup \Htild)/3 + 2/3$.
If all bins don't have the same type, then by \cref{thm:non-full-ub},
there can be at most 2 non-full bins, so $\greedyPack(\Wtild, \Htild)$
uses at most $4a(\Wtild \cup \Htild)/3 + 8/3$ bins.
\end{proof}

\subsection{The \texorpdfstring{$\thinGPack$}{skewed4Pack} Algorithm}
\label{sec:thin-gpack}

We now return to the 2BP problem.
$\thinGPack$ is an algorithm for 2BP takes as input a set $I$ of rectangular items
and a parameter $\eps \in (0, 1)$ where $\eps^{-1} \in \mathbb{Z}$.
It outputs a 4-stage bin packing of $I$.
$\thinGPack$ has the following outline:
\begin{enumerate}[A.]
\item Use linear grouping \cite{bp-aptas,kenyon1996strip} to round up the width or height of each item in $I$.
    This gives us a new instance $\Ihat$.
\item Pack $\Ihat$ into $1/\eps^2 + 1$ shelves,
    after possibly \emph{slicing} some items.
    Each shelf is a rectangular region with width or height more than $1/2$ and is fully packed, i.e.,
    the total area of items in a shelf equals the area of the shelf.
    If we treat each shelf as an item, we get a new instance $\Itild$.
\item Compute a packing of $\Itild$ into bins, after possibly slicing some items,
    using $\greedyPack$.
\item Pack most of the items of $I$ into the shelves in the bins. We will prove that
    the remaining items have very small area, so they can be packed separately using NFDH.
\end{enumerate}

\noindent \textbf{A. Item Classification and Rounding.}
Define $W \defeq \{i \in I: h(i) \le \delta_H \}$ and $H \defeq I - W$.
Items in $W$ are called \emph{wide} and items in $H$ are called \emph{tall}.
Let $\WLarge \defeq \{i \in W: w(i) > \eps \}$ and $\WSmall \defeq W - \WLarge$.
Similarly, let $\HLarge \defeq \{i \in H: h(i) > \eps \}$ and $\HSmall \defeq H - \HLarge$.

We will now use \emph{linear grouping}~\cite{bp-aptas,kenyon1996strip}
to round up the widths of items $\WLarge$ and the heights of items $\HLarge$
to get items $\WhatLarge$ and $\HhatLarge$, respectively.
By \cref{thm:lingroup-n} in \cref{sec:lingroup},
items in $\WhatLarge$ have at most $1/\eps^2$ distinct widths
and items in $\HhatLarge$ have at most $1/\eps^2$ distinct heights.

Let $\What \defeq \WhatLarge \cup \WSmall$,
$\Hhat \defeq \HhatLarge \cup \HSmall$,
and $\Ihat \defeq \What \cup \Hhat$.
Let $\fopt(\Ihat)$ be the minimum number of bins needed to pack $\Ihat$,
where items in $\WhatLarge$ can be sliced using horizontal cuts,
items in $\HhatLarge$ can be sliced using vertical cuts,
and items in $\WSmall \cup \HSmall$ can be sliced using both vertical and horizontal cuts.
Then the following lemma follows from \cref{thm:lingroup-repack} in \cref{sec:lingroup}.

\begin{lemma}
\label{thm:lingroup-opt-compare}
$\fopt(\Ihat) < (1+\eps)\opt(I) + 2$.
\end{lemma}

\noindent \textbf{B. Creating Shelves.}
We will use ideas from Kenyon and R\'emila's 2SP algorithm~\cite{kenyon1996strip}
to pack $\Ihat$ into \emph{shelves}.
Roughly, we solve a linear program to compute an optimal strip packing of $\What$,
where the packing is 3-stage. The first stage of cuts gives us shelves
and the second stage gives us containers.
From each shelf, we trim off space that doesn't belong to any container.
See \cref{sec:guill-thin-extra:shelves} for details.
Let $\Wtild$ be the shelves thus obtained.
Analogously, we can pack items $\Hhat$ into shelves $\Htild$.
Shelves in $\Wtild$ are called \emph{wide shelves}
and shelves in $\Htild$ are called \emph{tall shelves}.
Let $\Itild \defeq \Wtild \cup \Htild$.
We can interpret each shelf in $\Itild$ as a rectangular item.
We allow slicing $\Wtild$ and $\Htild$ using horizontal cuts and vertical cuts, respectively.
In \cref{sec:guill-thin-extra:shelves}, we prove the following facts.%

\begin{restatable}{lemma}{rthmCreateShelves}
\label{thm:shelves}
$\Itild$ has the following properties:
(a) $|\Wtild| \le 1+1/\eps^2$ and $|\Htild| \le 1+1/\eps^2$; 
(b) Items in $\Wtild$ have width more than $1/2$
    and items in $\Htild$ have height more than $1/2$;
(c) $a(\Itild) = a(\Ihat)$;
(d) $\max(\smallceil{\hsum(\Wtild)}, \smallceil{\wsum(\Htild)}) \le \fopt(\Ihat)$.
\end{restatable}

\noindent \textbf{C. Packing Shelves Into Bins.}
So far, we have packed $\Ihat$ into shelves $\Wtild$ and $\Htild$.
We will now use $\greedyPack(\Wtild, \Htild)$ to pack the shelves into bins.
By \cref{thm:greedy-pack}, we get a 2-stage packing of $\Wtild \cup \Htild$
into $m$ bins, where we make at most $m-1$ horizontal cuts in $\Wtild$
and at most $m-1$ vertical cuts in $\Htild$.
The horizontal cuts (resp.~vertical cuts) increase the number of wide shelves
(resp.~tall shelves) from at most $1 + 1/\eps^2$
to at most $m + 1/\eps^2$.
By \cref{thm:greedy-pack-bins}, \cref{thm:shelves}(d)
and \cref{thm:shelves}(c), we get
$m \le \max\left(\smallceil{\hsum(\Wtild)}, \smallceil{\wsum(\Htild)},
    \frac{4}{3}a(\Itild) + \frac{8}{3}\right)
\le \frac{4}{3}\fopt(\Ihat) + \frac{8}{3}$.

\noindent \textbf{D. Packing Items Into Containers.}
So far, we have a packing of shelves into $m$ bins,
where the shelves contain slices of items $\Ihat$.
We will now repack a large subset of the items $\Ihat$ into the shelves
without slicing them. See \cref{fig:thin-gpack-output} for an example output.
We will do this using a standard greedy algorithm.
See \cref{sec:guill-thin-extra:pack-into-containers} for details of the 
algorithm and proof of the following lemma.

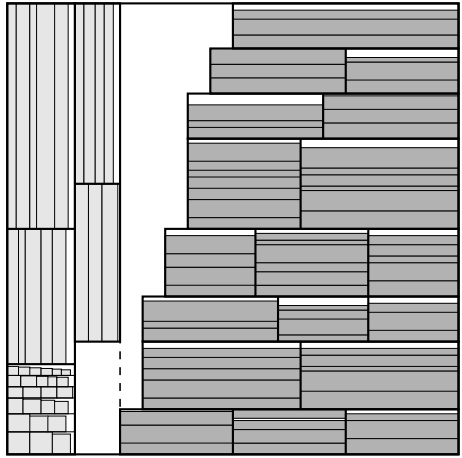
\begin{figure}[htb]
\centering
\ifcsname myu\endcsname\else\newlength{\myu}\fi
\setlength{\myu}{0.3cm}
\tikzset{mytransform/.style={}}
\tikzset{sepline/.style={draw,thick}}
\tikzset{halfsepline/.style={draw}}
\tikzset{wShelf/.style={draw,thick}}
\tikzset{hShelf/.style={draw,thick}}
\tikzset{wItem/.style={draw,fill={black!30},very thin}}
\tikzset{hItem/.style={draw,fill={black!10},very thin}}
\tikzset{bin/.style={draw,thick}}
\tikzset{binGrid/.style={draw,step=1\myu,{black!20}}}
\begin{tikzpicture}[mytransform]

\path[hItem] (0.0\myu, 0.0\myu) rectangle +(1.0\myu, 1.0\myu);
\path[hItem] (1.0\myu, 0.0\myu) rectangle +(1.0\myu, 1.0\myu);
\path[hItem] (2.0\myu, 0.0\myu) rectangle +(0.8\myu, 0.9\myu);
\path[hItem] (0.0\myu, 1.0\myu) rectangle +(1.0\myu, 0.8\myu);
\path[hItem] (1.0\myu, 1.0\myu) rectangle +(0.8\myu, 0.7\myu);
\path[hItem] (1.8\myu, 1.0\myu) rectangle +(0.8\myu, 0.7\myu);
\path[hItem] (0.0\myu, 1.8\myu) rectangle +(0.7\myu, 0.7\myu);
\path[hItem] (0.7\myu, 1.8\myu) rectangle +(0.8\myu, 0.65\myu);
\path[hItem] (1.5\myu, 1.8\myu) rectangle +(0.6\myu, 0.6\myu);
\path[hItem] (2.1\myu, 1.8\myu) rectangle +(0.6\myu, 0.55\myu);
\path[hItem] (0.0\myu, 2.5\myu) rectangle +(0.7\myu, 0.5\myu);
\path[hItem] (0.7\myu, 2.5\myu) rectangle +(0.8\myu, 0.5\myu);
\path[hItem] (1.5\myu, 2.5\myu) rectangle +(0.7\myu, 0.5\myu);
\path[hItem] (2.2\myu, 2.5\myu) rectangle +(0.7\myu, 0.5\myu);
\path[hItem] (0.0\myu, 3.0\myu) rectangle +(0.6\myu, 0.5\myu);
\path[hItem] (0.6\myu, 3.0\myu) rectangle +(0.7\myu, 0.48\myu);
\path[hItem] (1.3\myu, 3.0\myu) rectangle +(0.5\myu, 0.46\myu);
\path[hItem] (1.8\myu, 3.0\myu) rectangle +(0.4\myu, 0.44\myu);
\path[hItem] (2.2\myu, 3.0\myu) rectangle +(0.5\myu, 0.42\myu);
\path[hItem] (0.0\myu, 3.5\myu) rectangle +(0.5\myu, 0.4\myu);
\path[hItem] (0.5\myu, 3.5\myu) rectangle +(0.5\myu, 0.37\myu);
\path[hItem] (1.0\myu, 3.5\myu) rectangle +(0.5\myu, 0.34\myu);
\path[hItem] (1.5\myu, 3.5\myu) rectangle +(0.5\myu, 0.31\myu);
\path[hItem] (2.0\myu, 3.5\myu) rectangle +(0.4\myu, 0.28\myu);
\path[hItem] (2.4\myu, 3.5\myu) rectangle +(0.4\myu, 0.25\myu);

\path[hItem] (0.0\myu, 4\myu) rectangle +(0.5\myu, 6\myu);
\path[hItem] (0.5\myu, 4\myu) rectangle +(0.3\myu, 6\myu);
\path[hItem] (0.8\myu, 4\myu) rectangle +(0.7\myu, 6\myu);
\path[hItem] (1.5\myu, 4\myu) rectangle +(0.5\myu, 6\myu);
\path[hItem] (2.0\myu, 4\myu) rectangle +(0.6\myu, 6\myu);
\path[hItem] (0.0\myu, 10\myu) rectangle +(0.4\myu, 10\myu);
\path[hItem] (0.4\myu, 10\myu) rectangle +(0.6\myu, 10\myu);
\path[hItem] (1.0\myu, 10\myu) rectangle +(0.3\myu, 10\myu);
\path[hItem] (1.3\myu, 10\myu) rectangle +(0.8\myu, 10\myu);
\path[hItem] (2.1\myu, 10\myu) rectangle +(0.6\myu, 10\myu);

\path[hItem] (3.0\myu, 5\myu) rectangle +(0.6\myu, 7\myu);
\path[hItem] (3.6\myu, 5\myu) rectangle +(0.6\myu, 7\myu);
\path[hItem] (4.2\myu, 5\myu) rectangle +(0.7\myu, 7\myu);
\path[hItem] (3.0\myu, 12\myu) rectangle +(0.4\myu, 8\myu);
\path[hItem] (3.4\myu, 12\myu) rectangle +(0.5\myu, 8\myu);
\path[hItem] (3.9\myu, 12\myu) rectangle +(0.4\myu, 8\myu);
\path[hItem] (4.3\myu, 12\myu) rectangle +(0.4\myu, 8\myu);

\path[wItem] (5\myu, 0.0\myu) rectangle +(5\myu, 0.5\myu);
\path[wItem] (5\myu, 0.5\myu) rectangle +(5\myu, 0.8\myu);
\path[wItem] (5\myu, 1.3\myu) rectangle +(5\myu, 0.6\myu);
\path[wItem] (10\myu, 0.0\myu) rectangle +(5\myu, 0.5\myu);
\path[wItem] (10\myu, 0.5\myu) rectangle +(5\myu, 0.6\myu);
\path[wItem] (10\myu, 1.1\myu) rectangle +(5\myu, 0.4\myu);
\path[wItem] (10\myu, 1.6\myu) rectangle +(5\myu, 0.4\myu);
\path[wItem] (15\myu, 0.0\myu) rectangle +(5\myu, 0.7\myu);
\path[wItem] (15\myu, 0.7\myu) rectangle +(5\myu, 0.8\myu);
\path[wItem] (15\myu, 1.5\myu) rectangle +(5\myu, 0.3\myu);

\path[wItem] (6\myu, 2.0\myu) rectangle +(7\myu, 0.5\myu);
\path[wItem] (6\myu, 2.5\myu) rectangle +(7\myu, 0.8\myu);
\path[wItem] (6\myu, 3.3\myu) rectangle +(7\myu, 0.5\myu);
\path[wItem] (6\myu, 3.8\myu) rectangle +(7\myu, 0.5\myu);
\path[wItem] (6\myu, 4.3\myu) rectangle +(7\myu, 0.4\myu);
\path[wItem] (13\myu, 2.0\myu) rectangle +(7\myu, 0.8\myu);
\path[wItem] (13\myu, 2.8\myu) rectangle +(7\myu, 0.9\myu);
\path[wItem] (13\myu, 3.7\myu) rectangle +(7\myu, 0.2\myu);
\path[wItem] (13\myu, 3.9\myu) rectangle +(7\myu, 0.5\myu);
\path[wItem] (13\myu, 4.4\myu) rectangle +(7\myu, 0.3\myu);

\path[wItem] (6\myu, 5.0\myu) rectangle +(6\myu, 0.6\myu);
\path[wItem] (6\myu, 5.6\myu) rectangle +(6\myu, 0.3\myu);
\path[wItem] (6\myu, 5.9\myu) rectangle +(6\myu, 0.9\myu);
\path[wItem] (12\myu, 5.0\myu) rectangle +(4\myu, 0.3\myu);
\path[wItem] (12\myu, 5.3\myu) rectangle +(4\myu, 0.7\myu);
\path[wItem] (12\myu, 6.0\myu) rectangle +(4\myu, 0.4\myu);
\path[wItem] (12\myu, 6.4\myu) rectangle +(4\myu, 0.2\myu);
\path[wItem] (16\myu, 5.0\myu) rectangle +(4\myu, 0.5\myu);
\path[wItem] (16\myu, 5.5\myu) rectangle +(4\myu, 0.8\myu);
\path[wItem] (16\myu, 6.3\myu) rectangle +(4\myu, 0.4\myu);

\path[wItem] (7\myu, 7.0\myu) rectangle +(4\myu, 0.5\myu);
\path[wItem] (7\myu, 7.5\myu) rectangle +(4\myu, 0.8\myu);
\path[wItem] (7\myu, 8.3\myu) rectangle +(4\myu, 0.6\myu);
\path[wItem] (7\myu, 8.9\myu) rectangle +(4\myu, 0.8\myu);
\path[wItem] (11\myu, 7.0\myu) rectangle +(5\myu, 0.5\myu);
\path[wItem] (11\myu, 7.5\myu) rectangle +(5\myu, 0.6\myu);
\path[wItem] (11\myu, 8.1\myu) rectangle +(5\myu, 0.4\myu);
\path[wItem] (11\myu, 8.5\myu) rectangle +(5\myu, 0.8\myu);
\path[wItem] (11\myu, 9.3\myu) rectangle +(5\myu, 0.2\myu);
\path[wItem] (11\myu, 9.5\myu) rectangle +(5\myu, 0.3\myu);
\path[wItem] (16\myu, 7.0\myu) rectangle +(4\myu, 0.7\myu);
\path[wItem] (16\myu, 7.7\myu) rectangle +(4\myu, 0.8\myu);
\path[wItem] (16\myu, 8.5\myu) rectangle +(4\myu, 0.3\myu);
\path[wItem] (16\myu, 8.8\myu) rectangle +(4\myu, 0.5\myu);
\path[wItem] (16\myu, 9.3\myu) rectangle +(4\myu, 0.4\myu);

\path[wItem] (8\myu, 10.0\myu) rectangle +(5\myu, 0.5\myu);
\path[wItem] (8\myu, 10.5\myu) rectangle +(5\myu, 0.8\myu);
\path[wItem] (8\myu, 11.3\myu) rectangle +(5\myu, 0.5\myu);
\path[wItem] (8\myu, 11.8\myu) rectangle +(5\myu, 0.5\myu);
\path[wItem] (8\myu, 12.3\myu) rectangle +(5\myu, 0.4\myu);
\path[wItem] (8\myu, 12.6\myu) rectangle +(5\myu, 0.4\myu);
\path[wItem] (8\myu, 13.0\myu) rectangle +(5\myu, 0.8\myu);
\path[wItem] (13\myu, 10.0\myu) rectangle +(7\myu, 0.8\myu);
\path[wItem] (13\myu, 10.8\myu) rectangle +(7\myu, 0.9\myu);
\path[wItem] (13\myu, 11.7\myu) rectangle +(7\myu, 0.2\myu);
\path[wItem] (13\myu, 11.9\myu) rectangle +(7\myu, 0.5\myu);
\path[wItem] (13\myu, 12.4\myu) rectangle +(7\myu, 0.3\myu);
\path[wItem] (13\myu, 12.7\myu) rectangle +(7\myu, 0.9\myu);

\path[wItem] (8\myu, 14.0\myu) rectangle +(6\myu, 0.5\myu);
\path[wItem] (8\myu, 14.5\myu) rectangle +(6\myu, 0.3\myu);
\path[wItem] (8\myu, 14.8\myu) rectangle +(6\myu, 0.7\myu);
\path[wItem] (14\myu, 14.0\myu) rectangle +(6\myu, 0.7\myu);
\path[wItem] (14\myu, 14.7\myu) rectangle +(6\myu, 0.6\myu);
\path[wItem] (14\myu, 15.3\myu) rectangle +(6\myu, 0.6\myu);

\path[wItem] (9\myu, 16.0\myu) rectangle +(6\myu, 0.7\myu);
\path[wItem] (9\myu, 16.7\myu) rectangle +(6\myu, 0.6\myu);
\path[wItem] (9\myu, 17.3\myu) rectangle +(6\myu, 0.7\myu);
\path[wItem] (15\myu, 16.0\myu) rectangle +(5\myu, 0.6\myu);
\path[wItem] (15\myu, 16.6\myu) rectangle +(5\myu, 0.8\myu);
\path[wItem] (15\myu, 17.4\myu) rectangle +(5\myu, 0.2\myu);

\path[wItem] (10\myu, 18.0\myu) rectangle +(10\myu, 0.6\myu);
\path[wItem] (10\myu, 18.6\myu) rectangle +(10\myu, 0.7\myu);
\path[wItem] (10\myu, 19.3\myu) rectangle +(10\myu, 0.4\myu);

\path[sepline] (0\myu, 10\myu) -- +(3\myu, 0);
\path[halfsepline] (0\myu, 1.0\myu) -- +(3\myu, 0);
\path[halfsepline] (0\myu, 1.8\myu) -- +(3\myu, 0);
\path[halfsepline] (0\myu, 2.5\myu) -- +(3\myu, 0);
\path[halfsepline] (0\myu, 3.0\myu) -- +(3\myu, 0);
\path[halfsepline] (0\myu, 3.5\myu) -- +(3\myu, 0);

\path[sepline] (0\myu, 4\myu) -- +(3\myu, 0);
\path[sepline] (3\myu, 12\myu) -- +(2\myu, 0);
\path[sepline] (10\myu, 0\myu) -- +(0, 2\myu);
\path[sepline] (15\myu, 0\myu) -- +(0, 2\myu);
\path[sepline] (13\myu, 2\myu) -- +(0, 3\myu);
\path[sepline] (12\myu, 5\myu) -- +(0, 2\myu);
\path[sepline] (16\myu, 5\myu) -- +(0, 2\myu);
\path[sepline] (11\myu, 7\myu) -- +(0, 3\myu);
\path[sepline] (16\myu, 7\myu) -- +(0, 3\myu);
\path[sepline] (13\myu, 10\myu) -- +(0, 4\myu);
\path[sepline] (14\myu, 14\myu) -- +(0, 2\myu);
\path[sepline] (15\myu, 16\myu) -- +(0, 2\myu);

\path[wShelf] (5\myu, 0\myu) rectangle +(15\myu, 2\myu);
\path[wShelf] (6\myu, 2\myu) rectangle +(14\myu, 3\myu);
\path[wShelf] (6\myu, 5\myu) rectangle +(14\myu, 2\myu);
\path[wShelf] (7\myu, 7\myu) rectangle +(13\myu, 3\myu);
\path[wShelf] (8\myu, 10\myu) rectangle +(12\myu, 4\myu);
\path[wShelf] (8\myu, 14\myu) rectangle +(12\myu, 2\myu);
\path[wShelf] (9\myu, 16\myu) rectangle +(11\myu, 2\myu);
\path[wShelf] (10\myu, 18\myu) rectangle +(10\myu, 2\myu);
\path[hShelf] (0\myu, 0\myu) rectangle +(3\myu, 20\myu);
\path[hShelf] (3\myu, 5\myu) rectangle +(2\myu, 15\myu);
\draw[semithick,dashed] (5\myu, 0\myu) -- +(0\myu, 20\myu);
\path[bin] (0\myu, 0\myu) rectangle (20\myu, 20\myu);
\end{tikzpicture}

\caption[A type-1 bin in the packing of $\Ihat$ computed by $\thinGPack$.]%
{A type-1 bin in the packing of $\Ihat$ computed by $\thinGPack$.
The packing contains 5 tall containers in 2 tall shelves
and 18 wide containers in 8 wide shelves.}%
\label{fig:thin-gpack-output}
\end{figure}

\begin{restatable}{lemma}{rthmDiscardAreaUb}
\label{thm:discard-area-ub}
Let $P$ be a packing of $\Itild$ into $m$ bins, where we sliced wide shelves by making
at most $m-1$ horizontal cuts and sliced tall shelves by making at most $m-1$ vertical cuts.
Then we can (non-fractionally) pack a large subset of items $\Ihat$
into the shelves in $P$ such that
the unpacked items (also called \emph{discarded items}) from $\What$ have area less than
$\eps \hsum(\Wtild) + \delta_H(1 + \eps)(m + 1/\eps^2)$,
and the unpacked items from $\Hhat$ have area less than
$\eps \wsum(\Htild) + \delta_W(1 + \eps)(m + 1/\eps^2)$.
\end{restatable}

We will pack the wide discarded items into new bins using NFDH
and pack the tall discarded items into new bins using NFDW.

Finally, we prove the performance guarantee of $\thinGPack_{\eps}(I)$.
\begin{lemma}
\label{thm:thin-gpack-strong}
Let $I$ be a set of $(\delta_W, \delta_H)$-skewed items.
Then $\thinGPack_{\eps}(I)$ outputs a 4-stage packing of $I$
in time $O((1/\eps)^{O(1/\eps)} + n\log n)$
and uses less than $\alpha(1+\eps)\opt(I) + 2\beta$ bins, where
$\Delta \defeq \frac{1}{2}\left(\frac{\delta_H}{1-\delta_H}
    + \frac{\delta_W}{1-\delta_W}\right),\;
     \alpha \defeq \frac{4}{3}(1+4\Delta)(1+3\eps),\;
     \beta \defeq \frac{2\Delta(1+\eps)}{\eps^2} + \frac{10}{3}
    + \frac{19\Delta}{3} + \frac{16\Delta\eps}{3}$.
\end{lemma}
\begin{proof}
The discarded items are packed using NFDH or NFDW, which output a 2-stage packing.
Since $\greedyPack$ outputs a 2-stage packing of the shelves
and the packing of items into the shelves is a 2-stage packing,
the bin packing of non-discarded items is a 4-stage packing.
The time taken by $\thinGPack$ is at most $O((1/\eps)^{O(1/\eps)} + n\log n)$.

Suppose $\greedyPack$ uses at most $m$ bins. Then by \cref{thm:greedy-pack-bins},
$m \le 4\fopt(\Ihat)/3 + 8/3$.
Let $W^d$ and $H^d$ be the items discarded from $W$ and $H$, respectively.
By \cref{thm:discard-area-ub} and \cref{thm:shelves}(d),
$a(W^d) < \eps\fopt(\Ihat) + \delta_H(1 + \eps)(m + 1/\eps^2)$
and $a(H^d) < \eps\fopt(\Ihat) + \delta_W(1 + \eps)(m + 1/\eps^2)$.

By \cref{thm:nfdh-wide}, the number of bins used by $\thinGPack_{\eps}(I)$ is less than
$m + \frac{2a(W^d)+1}{1-\delta_H} + \frac{2a(H^d)+1}{1-\delta_W}
 \le (1 + 4\Delta(1+\eps))m + 4\eps(1+\Delta)\fopt(\Ihat)
    + 2(1+\Delta) + 4\Delta(1+\eps)/\eps^2
 \le \alpha\fopt(\Ihat) + 2(\beta - 1)
< \alpha(1+\eps)\opt(I) + 2\beta.$
The last inequality follows from \cref{thm:lingroup-opt-compare}.
\end{proof}

Now we conclude with the proof of \cref{thm:thin-gpack}.

\begin{proof}[Proof of \cref{thm:thin-gpack}]
This is a simple corollary of \cref{thm:thin-gpack-strong}, where
$\delta \le 1/2$ gives us $\Delta \le 2\delta$,
$\alpha(1+\eps) \le (4/3)(1+8\delta)(1+7\eps)$,
and $\beta \le 4/\eps^2 + 15$.
\end{proof}

\section{Almost-Optimal Bin Packing of \Thin{} Rectangles}
\label{sec:thin-bp}

In this section, we will describe the algorithm $\thinCPack$. $\thinCPack$ takes as input
a set $I$ of items and a parameter $\eps \in (0, 1/2]$, where $\eps^{-1} \in \mathbb{Z}$.
We will prove that $\thinCPack$ has AAR $1+20\eps$ when $\delta$ is sufficiently small.
$\thinCPack$ works roughly as follows:
\begin{enumerate}
\item Invoke the subroutine $\round(I)$ (described in \cref{sec:thin-bp:round}).
    $\round(I)$ returns a pair $(\Itild, \Imed)$.
    Here $\Imed$, called the set of \emph{medium items}, has low total area,
    so we can pack it in a small number of bins.
    $\Itild$, called the set of \emph{rounded items}, is obtained by
    rounding up the width or height of each item in $I - \Imed$,
    so that $\Itild$ has special properties that help us pack it easily.
\item Compute the optimal \emph{fractional compartmental} bin packing of $\Itild$
    (we will define \emph{compartmental} and \emph{fractional} later).
\item Use this packing of $\Itild$ to obtain a packing of $I$
    that uses slightly more number of bins.
\end{enumerate}

To bound the AAR of $\thinCPack$, we will prove a structural theorem
in \cref{sec:thin-bp:struct}, i.e., we will prove that
the optimal fractional compartmental packing of $\Itild$ uses close to $\opt(I)$ bins.

\subsection{Classifying and Rounding Items}
\label{sec:thin-bp:round}
\label{sec:thin-bp:remmed}

We now describe the algorithm $\round$ and
show that its output satisfies important properties.

First, we will find a set $\Imed \subseteq I$
and positive constants $\epsLarge$ and $\epsSmall$
such that $a(\Imed) \le \eps a(I)$, $\epsSmall \ll \epsLarge$,
and $I - \Imed$ is $(\epsSmall, \epsLarge]$-free, i.e.,
no item in $I - \Imed$ has its width or height in the interval $(\epsSmall, \epsLarge]$.
Then we can remove $\Imed$ from $I$ and pack it separately
into a small number of bins using NFDH. We will see that
the $(\epsSmall, \epsLarge]$-freeness of $I - \Imed$
will help us pack $I - \Imed$ efficiently.
Specifically, we require $\epsLarge \le \eps$, $\epsLarge^{-1} \in \mathbb{Z}$,
and $\epsSmall = f(\epsLarge)$, where $f(x) \defeq \frac{\eps x}{104(1+1/(\eps x))^{2/x-2}}$.
We explain this choice of $f$ in \cref{sec:thinCPack}.
Intuitively, such an $f$ ensures that $\epsSmall \ll \epsLarge$
and $\epsSmall^{-1} \in \mathbb{Z}$.
For $\thinCPack$ to work, we require $\delta \le \epsSmall$.
Finding such an $\Imed$ and $\epsLarge$ is a standard technique \cite{JansenP2013, bansal2014binpacking},
so we defer the details to \cref{sec:thin-bp-extra:remmed}.

Next, we classify the items in $I - \Imed$ into three disjoint classes:
\begin{itemize}[noitemsep]
\item Wide items: $W \defeq \{i \in I: w(i) > \epsLarge \textrm{ and } h(i) \le \epsSmall \}$.
\item Tall items: $H \defeq \{i \in I: w(i) \le \epsSmall \textrm{ and } h(i) > \epsLarge \}$.
\item Small items: $S \defeq \{i \in I: w(i) \le \epsSmall \textrm{ and } h(i) \le \epsSmall \}$.
\end{itemize}

We will now use \emph{linear grouping}~\cite{bp-aptas,kenyon1996strip}
to round up the widths of items $W$ and the heights of items $H$
to get items $\Wtild$ and $\Htild$, respectively.
By \cref{thm:lingroup-n} in \cref{sec:lingroup},
items in $\Wtild$ have at most $1/(\eps\epsLarge)$ distinct widths
and items in $\Htild$ have at most $1/(\eps\epsLarge)$ distinct heights.
Let $\Itild \defeq \Wtild \cup \Htild \cup S$.

\begin{definition}[Fractional packing]
Suppose we are allowed to slice wide items in $\Itild$ using horizontal cuts,
slice tall items in $\Itild$ using vertical cuts and slice
small items in $\Itild$ using both horizontal and vertical cuts.
For any $\Xtild \subseteq \Itild$, a bin packing of the slices of $\Xtild$
is called a \emph{fractional packing} of $\Xtild$.
The optimal fractional packing of $\Xtild$ is denoted by $\fopt(\Xtild)$.
\end{definition}

\begin{lemma}
\label{thm:thin-bp:lingroup-opt-compare}
$\fopt(\Itild) < (1+\eps)\opt(I) + 2$.
\end{lemma}
\begin{proof}
Directly follows from \cref{thm:lingroup-repack} in \cref{sec:lingroup}.
\end{proof}

\subsection{Structural Theorem}
\label{sec:thin-bp:struct}

We will now define compartmental packing
and prove the structural theorem, which says that
the number of bins in the optimal fractional compartmental packing of $\Itild$
is roughly equal to $\fopt(\Itild)$.

We first show how to \emph{discretize} a packing, i.e.,
we show that given a fractional packing of items in a bin,
we can remove a small fraction of tall and small items
and shift the remaining items leftwards so that the left and right edges
of each wide item belong to a constant-sized set $\Tcal$,
where $|\Tcal| \le (1+1/\eps\epsLarge)^{2/\epsLarge - 2}$.
Next, we define \emph{compartmental} packing and show how to convert
a discretized packing to a compartmental packing.

For any rectangle $i$ packed in a bin, let $x_1(i)$ and $x_2(i)$ denote the $x$-coordinates
of its left and right edges, respectively, and let $y_1(i)$ and $y_2(i)$
denote the $y$-coordinates of its bottom and top edges, respectively.
Let $R$ be the set of distinct widths of items in $\Wtild$.
Given the way we rounded items, $|R| \le 1/\eps\epsLarge$.
Recall that $\epsLarge \le \eps \le 1/2$ and $\epsLarge^{-1}, \eps^{-1} \in \mathbb{Z}$.

\begin{theorem}
\label{thm:disc-hor-pos}
Given a fractional packing of items $\Jtild \subseteq \Itild$ into a bin,
we can remove tall and small items of total area less than $\eps$
and shift some of the remaining items to the left such that for every wide item $i$,
we get $x_1(i), x_2(i) \in \Tcal$.
\end{theorem}
\begin{proof}
For wide items $u$ and $v$ in the bin, we say that $u \prec v$ iff
the right edge of $u$ is to the left of the left edge of $v$.
Formally $u \prec v \iff x_2(u) \le x_1(v)$.
We call $u$ a \emph{predecessor} of $v$.
A sequence $[i_1, i_2, \ldots, i_k]$ such that $i_1 \prec i_2 \prec \ldots \prec i_k$
is called a \emph{chain} ending at $i_k$.
For a wide item $i$, define $\level(i)$ as the number of items in the longest chain
ending at $i$. Formally, $\level(i) \defeq 1$ if $i$ has no predecessors,
and $\left(1 + \max_{j \prec i} \level(j)\right)$ otherwise.
Let $W_j$ be the items at level $j$, i.e., $W_j \defeq \{i: \level(i) = j\}$.
Note that the level of an item can be at most $1/\epsLarge-1$,
since each wide item has width more than $\epsLarge$.

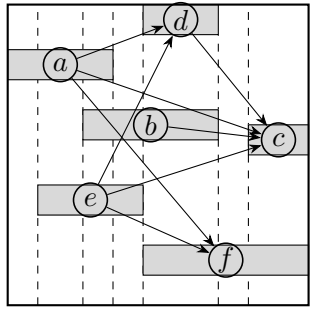
\begin{figure}[htb]
\centering
\ifcsname pGameL\endcsname\else\newlength{\pGameL}\fi
\setlength{\pGameL}{0.2cm}
\tikzset{bin/.style={draw,thick}}
\tikzset{item/.style={draw,fill={black!15}}}
\tikzset{myarrow/.style={draw,->,>={Stealth}}}
\tikzset{mynode/.style={pos=0.5,inner sep=0,minimum size=0.45cm,shape=circle,semithick,draw}}
\tikzset{cutline/.style={draw,dashed}}
\begin{tikzpicture}

\path[cutline] (2\pGameL, 0\pGameL) -- +(0, 20\pGameL);
\path[cutline] (5\pGameL, 0\pGameL) -- +(0, 20\pGameL);
\path[cutline] (7\pGameL, 0\pGameL) -- +(0, 20\pGameL);
\path[cutline] (9\pGameL, 0\pGameL) -- +(0, 20\pGameL);
\path[cutline] (14\pGameL, 0\pGameL) -- +(0, 20\pGameL);
\path[cutline] (16\pGameL, 0\pGameL) -- +(0, 20\pGameL);

\path[item] (0\pGameL, 15\pGameL) rectangle +(7\pGameL, 2\pGameL) node[mynode] (wa) {$a$};
\path[item] (5\pGameL, 11\pGameL) rectangle +(9\pGameL, 2\pGameL) node[mynode] (wb) {$b$};
\path[item] (16\pGameL, 10\pGameL) rectangle +(4\pGameL, 2\pGameL) node[mynode] (wc) {$c$};
\path[item] (9\pGameL, 18\pGameL) rectangle +(5\pGameL, 2\pGameL) node[mynode] (wd) {$d$};
\path[item] (2\pGameL, 6\pGameL) rectangle +(7\pGameL, 2\pGameL) node[mynode] (we) {$e$};
\path[item] (9\pGameL, 2\pGameL) rectangle +(11\pGameL, 2\pGameL) node[mynode] (wf) {$f$};

\path[bin] (0\pGameL, 0\pGameL) rectangle (20\pGameL, 20\pGameL);

\path[myarrow] (wa) -- (wc);
\path[myarrow] (wa) -- (wd);
\path[myarrow] (wa) -- (wf);
\path[myarrow] (wb) -- (wc);
\path[myarrow] (wd) -- (wc);
\path[myarrow] (we) -- (wc);
\path[myarrow] (we) -- (wd);
\path[myarrow] (we) -- (wf);
\end{tikzpicture}

\caption[Relation $\prec$ among items in a bin]%
{Example illustrating the $\prec$ relationship between wide items in a bin.
An edge is drawn from $u$ to $v$ iff $u \prec v$.
Here $W_1 = \{a, e, b\}$, $W_2 = \{d, f\}$ and $W_3 = \{c\}$.}
\label{fig:precedence-graph}
\end{figure}

We will describe an algorithm for discretization.
But first, we need to introduce two recursively-defined set families
$(S_1, S_2, \ldots)$ and $(T_0, T_1, \ldots)$.
Let $T_0 \defeq \{0\}$ and $t_0 \defeq 1$. For any $j > 0$, define
$t_j \defeq (1+1/\eps\epsLarge)^{2j},\, \delta_j \defeq \eps\epsLarge/t_{j-1},\,
S_j \defeq T_{j-1} \cup \{k\delta_j: k \in \mathbb{Z}, 0 \le k < 1/\delta_j\},\,
T_j \defeq \{x + y: x \in S_j, y \in R \cup \{0\}\}$.
Note that $\forall j > 0$, we have $T_{j-1} \subseteq S_j \subseteq T_j$
and $\delta_j^{-1} \in \mathbb{Z}$.
Define $\Tcal \defeq T_{1/\epsLarge-1}$.

Our discretization algorithm proceeds in stages, where in the $j\Th$ stage,
we apply two transformations to the items in the bin,
called \emph{strip-removal} and \emph{compaction}.
\\ \textbf{Strip-removal}: For each $x \in T_{j-1}$,
consider a strip of width $\delta_j$ and height 1 in the bin
whose left edge has coordinate $x$.
Discard the slices of tall and small items inside the strips.
\\ \textbf{Compaction}: Move all tall and small items as much towards the left as possible
(imagine a gravitational force acting leftwards on the tall and small items)
while keeping the wide items fixed.
Then move each wide item $i \in W_j$ leftwards till $x_1(i) \in S_j$.

Observe that the algorithm maintains the following invariant:
\textsl{after $k$ stages, for each $j \in [k]$,
each item $i \in W_j$ has $x_1(i) \in S_j$ (and hence $x_2(i) \in T_j$).}
This ensures that after the algorithm ends, $x_1(i), x_2(i) \in \Tcal$.
All that remains to prove is that the total area of items discarded
during strip-removal is at most $\eps$ and that compaction is always possible.

\begin{lemma}
For all $j \ge 0$, $|T_j| \le t_j$.
\end{lemma}
\begin{proof}
We will prove this by induction. The base case holds because $|T_0| = t_0 = 1$.
Now  assume $|T_{j-1}| \le t_{j-1}$. Then
$|T_j| \le (|R|+1)|S_j|
    \le \left(\frac{1}{\eps\epsLarge}+1\right)\left(|T_{j-1}| + \frac{1}{\delta_j}\right)
    \le \left(\frac{1}{\eps\epsLarge}+1\right)^2 t_{j-1}
    = t_j.$
\end{proof}

Therefore, $|\Tcal| \le t_{1/\epsLarge-1} = (1+1/\eps\epsLarge)^{2/\epsLarge - 2}$.

\begin{lemma}
Items discarded by strip-removal (across all stages)
have total area less than $\eps$.
\end{lemma}
\begin{proof}
In the $j\Th$ stage, we create $|T_{j-1}|$ strips,
and each strip has total area at most $\delta_j$.
Therefore, the area discarded in the $j\Th$ stage is at most
$|T_{j-1}|\delta_j \le t_{j-1}\delta_j = \eps\epsLarge$.
Since there can be at most $1/\epsLarge-1$ stages,
we discard an area of less than $\eps$ across all stages.
\end{proof}

\begin{lemma}
Compaction always succeeds, i.e., in the $j\Th$ stage, while moving
item $i \in W_j$ leftwards, no other item will block its movement.
\end{lemma}
\begin{proof}
Let $i \in W_j$. Let $z$ be the $x$-coordinate of the left edge of the strip
immediately to the left of item $i$, i.e.,
$z \defeq \max(\{x \in T_{j-1}: x \le x_1(i)\})$.
For any wide item $i'$, we have $x_2(i') \le x_1(i) \iff i' \prec i \iff \level(i') \le j-1$.
By our invariant, we get
$\level(i') \le j-1 \implies x_2(i') \in T_{j-1} \implies x_2(i') \le z$.
Therefore, for every wide item $i'$, $x_2(i') \not\in (z, x_1(i)]$.

In the $j\Th$ strip-removal, we cleared the strip $[z, z+\delta_j] \times [0, 1]$.
If $x_1(i) \in [z, z+\delta_j]$, then $i$ can freely move to $z$,
and $z \in T_{j-1} \subseteq S_j$.
Since no wide item has its right edge in $(z, x_1(i)]$, if $x_1(i) > z + \delta_j$,
all the tall and small items whose left edge lies in $[z+\delta_j, x_1(i)]$
will move leftwards by at least $\delta_j$ during compaction.
Hence, there would be an empty space of width at least $\delta_j$
to the left of item $i$
(see \cref{fig:compaction-zoom}).
Therefore, we can move $i$ leftwards to make
$x_1(i)$ a multiple of $\delta_j$, and then $x_1(i)$ would belong to $S_j$.
\end{proof}

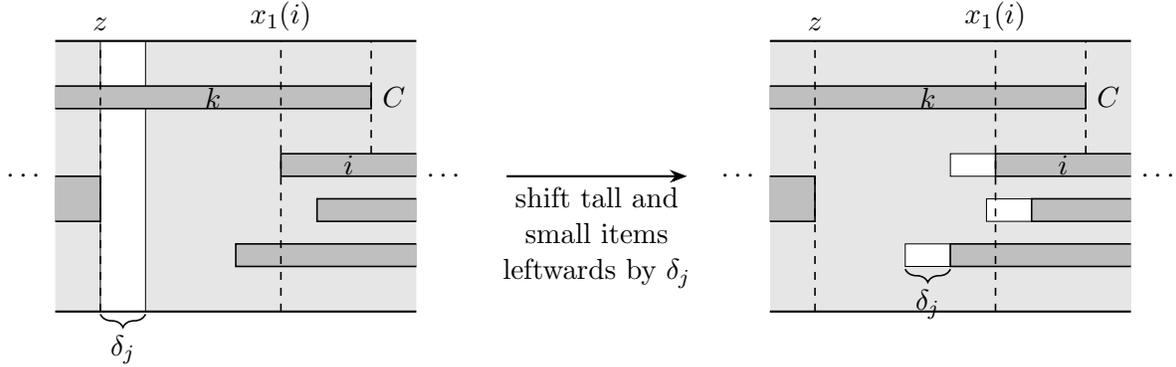
\begin{figure}[!htb]
\centering
\ifcsname myu\endcsname\else\newlength{\myu}\fi
\setlength{\myu}{0.6cm}
\tikzset{mypic/.pic={
\path[item] (3\myu, 1\myu) rectangle (7\myu, 1.5\myu);
\path[item-boundary] (7\myu, 1\myu) -- (3\myu, 1\myu) -- (3\myu, 1.5\myu) -- (7\myu, 1.5\myu);
\path[item] (4.8\myu, 2\myu) rectangle (7\myu, 2.5\myu);
\path[item-boundary] (7\myu, 2\myu) -- (4.8\myu, 2\myu) -- (4.8\myu, 2.5\myu) -- (7\myu, 2.5\myu);
\path[item] (4\myu, 3\myu) rectangle (7\myu, 3.5\myu) node[pos=0.5] {$i$};
\path[item-boundary] (7\myu, 3\myu) -- (4\myu, 3\myu) -- (4\myu, 3.5\myu) -- (7\myu, 3.5\myu);
\path[item] (-1\myu, 2\myu) rectangle (0\myu, 3\myu);
\path[item-boundary] (-1\myu, 2\myu) -- (0\myu, 2\myu) -- (0\myu, 3\myu) -- (-1\myu, 3\myu);
\path[item] (-1\myu, 4.5\myu) rectangle (6\myu, 5\myu) node[pos=0.5] {$k$};
\path[item-boundary] (-1\myu, 4.5\myu) -- (6\myu, 4.5\myu) -- (6\myu, 5\myu) -- (-1\myu, 5\myu);

\draw[thick]
    (-1\myu, 0\myu) -- (7\myu, 0\myu)
    (-1\myu, 6\myu) -- (7\myu, 6\myu);
\node[anchor=east] at (-1\myu, 3\myu) {$\cdots$};
\node[anchor=west] at (7\myu, 3\myu) {$\cdots$};
\draw[dashed,semithick]
    (0\myu, 0\myu) -- (0\myu, 6\myu)
    (4\myu, 0\myu) -- (4\myu, 6\myu)
    (6\myu, 3.5\myu) -- (6\myu, 6\myu);
\node[anchor=south] at (0\myu, 6\myu) {$z$};
\node[anchor=south] at (4\myu, 6\myu) {$x_1(i)$};
\node at (6.5\myu, 4.75\myu) {$C$};
}}
\begin{tikzpicture}[
item-boundary/.style={draw,semithick},
item/.style={fill={black!25}},
myarrow/.style={->,>={Stealth},thick},
mybrace/.style = {decoration={amplitude=5pt,brace,mirror,raise=1pt},semithick,decorate},
]
\begin{scope}
\path[fill={black!10}]
    (-1\myu, 0\myu) rectangle (0\myu, 6\myu)
    (1\myu, 0\myu) rectangle (7\myu, 6\myu);
\draw[very thin] (0\myu, 0\myu) -- (0\myu, 6\myu) (1\myu, 0\myu) -- (1\myu, 6\myu);
\draw[mybrace] (0\myu, 0\myu) -- node[below=5pt] {$\delta_j$} (1\myu, 0\myu);
\pic at (0\myu, 0\myu) {mypic};
\end{scope}
\draw[myarrow] (9\myu, 3\myu) -- (13\myu, 3\myu)
    node[pos=0.5,anchor=north,align=center,text width=3cm]
    {shift tall and small items leftwards by $\delta_j$};
\begin{scope}[xshift={9.5cm}]
\path[fill={black!10}] (-1\myu, 0\myu) rectangle (7\myu, 6\myu);
\draw[very thin,fill=white]
    (2\myu, 1\myu) rectangle (3\myu, 1.5\myu)
    (3.8\myu, 2\myu) rectangle (4.8\myu, 2.5\myu)
    (3\myu, 3\myu) rectangle (4\myu, 3.5\myu);
\pic at (0\myu, 0\myu) {mypic};
\draw[mybrace] (2\myu, 1\myu) -- node[below=5pt] {$\delta_j$} (3\myu, 1\myu);
\end{scope}
\end{tikzpicture}

\caption[Creation of empty space during compaction.]%
{This figure shows a region in the bin in the vicinity of item $i \in W_j$.
It illustrates how shifting tall and small items during compaction in the $j\Th$ stage
creates a free space of width $\delta$ to the left of some wide items, including $i$.
Wide items are shaded dark and the lightly shaded region
potentially contains tall and small items.
Note that some tall and small items in the region $C$
may be unable to shift left because item $k$ is blocking them.
All other tall and small items in this figure to the right of $z$
can shift left by $\delta_j$.}
\label{fig:compaction-zoom}
\end{figure}

Hence, compaction always succeeds and we get $x_1(i), x_2(i) \in \Tcal$
for each wide item $i$.
\end{proof}

\begin{definition}[Compartmental packing]
\label{defn:thin-bp:compartmental}
Consider a bin with some items packed into it. A \emph{compartment} $C$
is defined as a rectangular region in the bin satisfying the following properties:
\begin{itemize}[noitemsep]
\item $x_1(C), x_2(C) \in \Tcal$.
\item $y_1(C), y_2(C)$ are multiples of $\epsCont \defeq \eps\epsLarge/6|\Tcal|$.
\item $C$ does not contain both wide items and tall items.
\item If $C$ contains tall items, then $x_1(C)$ and $x_2(C)$
    are consecutive values in $\Tcal$.
\end{itemize}
If a compartment contains a wide item, it is called a \emph{wide compartment}.
Otherwise it is called a \emph{tall compartment}.
A packing of items into a bin is called \emph{compartmental}
iff there is a set of non-overlapping \emph{compartments} in the bin
such that each wide or tall item lies completely inside some compartment,
and there are at most $\nW \defeq 3(1/\epsLarge-1)|\Tcal| + 1$ wide compartments
and at most $\nH \defeq (1/\epsLarge-1)|\Tcal|$ tall compartments in the bin.
A packing of items into multiple bins is called compartmental iff
each bin is compartmental.
\end{definition}
Note that small items can be packed both inside and outside compartments.

The following two results are proved in \cref{sec:thin-bp-extra:compartmentalize}
using standard techniques.

\begin{restatable}{lemma}{rthmCompartmentalize}
\label{thm:thin-bp:compartmentalize}
Suppose $x_1(i), x_2(i) \in \Tcal$ for each wide item $i$ in a bin.
Then by removing wide and small items of area less than $\eps$,
we can get a compartmental packing of the remaining items.
\end{restatable}

\begin{restatable}{theorem}{rthmStruct}
\label{thm:struct}
For a set $\Itild$ of $\delta$-\thin{} rounded items,
define $\fcopt(\Itild)$ as the number of bins in
the optimal fractional compartmental packing%
\footnote{A \emph{fractional compartmental} packing of $\Itild$ is a fractional packing
of $\Itild$ that is also compartmental.}\!
of $\Itild$.
Then $\fcopt(\Itild) < (1+4\eps)\fopt(\Itild) + 2$.
\end{restatable}

\subsection{Packing Algorithm}
\label{sec:thin-bp:algo}

We now describe the $\thinCPack$ algorithm for packing a set $I$
of $\delta$-\thin{} items.
Roughly, $\thinCPack$ first computes $(\Itild, \Imed) \defeq \round(I)$.
It then computes the optimal fractional compartmental packing of $\Itild$
by first guessing a packing of empty compartments into bins
and then fractionally packing the wide and tall items into the compartments
using a linear program.
It then converts the fractional packing of $\Itild$ to a non-fractional packing of $I$
with only a tiny increase in the number of bins.
See \cref{fig:thincpack} for a visual overview of $\thinCPack$.
We defer the details to \cref{sec:enum-configs,sec:feas-lp,sec:greedy-cont,sec:thinCPack}
and simply state the final result.

\begin{restatable}{theorem}{rthmThinCPack}
The number of bins used by $\thinCPack_{\eps}(\Itild)$ is less than
\[ (1+20\eps)\opt(I) +
\frac{1}{13}\left(1 + \frac{1}{\eps\epsLarge}\right)^{2/\epsLarge - 2} + 23. \]
\end{restatable}

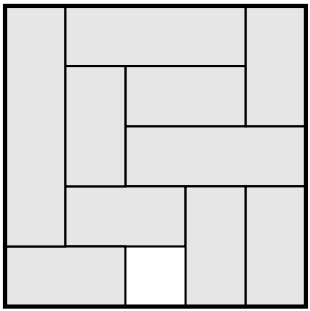
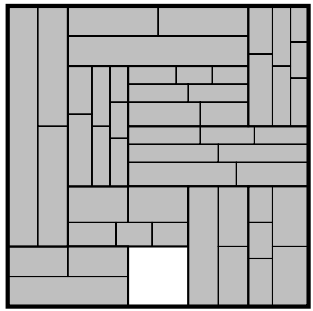
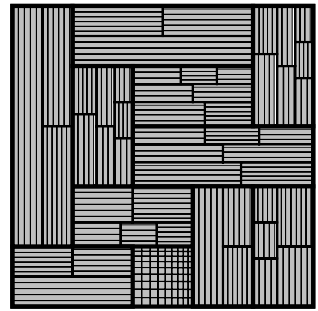
\begin{figure}[htb]
\hbadness=10000
\begin{subfigure}[t]{0.3\textwidth}
\centering
\tikzset{compartment/.style={draw,thick,fill={black!10}},
container/.style={},item/.style={}}
\ifcsname myu\endcsname\else\newlength{\myu}\fi
\setlength{\myu}{0.8cm}
\tikzset{pic1/.pic={
    \path[item]
        (0.00\myu, 0\myu) rectangle +(0.09\myu, 0.6\myu)
        (0.09\myu, 0\myu) rectangle +(0.10\myu, 0.6\myu)
        (0.19\myu, 0\myu) rectangle +(0.08\myu, 0.6\myu)
        (0.27\myu, 0\myu) rectangle +(0.08\myu, 0.6\myu);
    \path[item]
        (0\myu, 0.6\myu) rectangle +(0.10\myu, 0.4\myu)
        (0.10\myu, 0.6\myu) rectangle +(0.07\myu, 0.4\myu)
        (0.17\myu, 0.6\myu) rectangle +(0.07\myu, 0.4\myu)
        (0.24\myu, 0.6\myu) rectangle +(0.08\myu, 0.4\myu)
        (0.32\myu, 0.6\myu) rectangle +(0.07\myu, 0.4\myu);
    \path[item]
        (0.40\myu, 0\myu) rectangle +(0.11\myu, 0.5\myu)
        (0.50\myu, 0\myu) rectangle +(0.09\myu, 0.5\myu)
        (0.59\myu, 0\myu) rectangle +(0.09\myu, 0.5\myu);
    \path[item]
        (0.40\myu, 0.5\myu) rectangle +(0.08\myu, 0.5\myu)
        (0.48\myu, 0.5\myu) rectangle +(0.09\myu, 0.5\myu)
        (0.57\myu, 0.5\myu) rectangle +(0.08\myu, 0.5\myu);
    \path[item]
        (0.70\myu, 0\myu) rectangle +(0.1\myu, 0.4\myu)
        (0.80\myu, 0\myu) rectangle +(0.1\myu, 0.4\myu)
        (0.90\myu, 0\myu) rectangle +(0.1\myu, 0.4\myu);
    \path[item]
        (0.70\myu, 0.4\myu) rectangle +(0.07\myu, 0.3\myu)
        (0.77\myu, 0.4\myu) rectangle +(0.07\myu, 0.3\myu)
        (0.84\myu, 0.4\myu) rectangle +(0.07\myu, 0.3\myu)
        (0.91\myu, 0.4\myu) rectangle +(0.07\myu, 0.3\myu);
    \path[item]
        (0.70\myu, 0.7\myu) rectangle +(0.08\myu, 0.3\myu)
        (0.78\myu, 0.7\myu) rectangle +(0.08\myu, 0.3\myu)
        (0.86\myu, 0.7\myu) rectangle +(0.09\myu, 0.3\myu);
    \path[container]
        (0\myu, 0\myu) rectangle +(0.4\myu, 0.6\myu)
        (0\myu, 0.6\myu) rectangle +(0.4\myu, 0.4\myu)
        (0.4\myu, 0\myu) rectangle +(0.3\myu, 0.5\myu)
        (0.4\myu, 0.5\myu) rectangle +(0.3\myu, 0.5\myu)
        (0.7\myu, 0\myu) rectangle +(0.3\myu, 0.4\myu)
        (0.7\myu, 0.4\myu) rectangle +(0.3\myu, 0.3\myu)
        (0.7\myu, 0.7\myu) rectangle +(0.3\myu, 0.3\myu);
    \path[compartment] (0\myu, 0\myu) rectangle (1\myu, 1\myu);
}}
\tikzset{pic2/.pic={
    \path[item]
        (0.0\myu, 0\myu) rectangle +(0.1\myu, 1\myu)
        (0.1\myu, 0\myu) rectangle +(0.1\myu, 1\myu)
        (0.2\myu, 0\myu) rectangle +(0.1\myu, 1\myu)
        (0.3\myu, 0\myu) rectangle +(0.1\myu, 1\myu)
        (0.4\myu, 0\myu) rectangle +(0.1\myu, 1\myu);
    \path[item]
        (0.50\myu, 0\myu) rectangle +(0.08\myu, 0.5\myu)
        (0.58\myu, 0\myu) rectangle +(0.08\myu, 0.5\myu)
        (0.66\myu, 0\myu) rectangle +(0.08\myu, 0.5\myu)
        (0.74\myu, 0\myu) rectangle +(0.08\myu, 0.5\myu)
        (0.82\myu, 0\myu) rectangle +(0.08\myu, 0.5\myu)
        (0.90\myu, 0\myu) rectangle +(0.08\myu, 0.5\myu);
    \path[item]
        (0.50\myu, 0.5\myu) rectangle +(0.09\myu, 0.5\myu)
        (0.59\myu, 0.5\myu) rectangle +(0.09\myu, 0.5\myu)
        (0.68\myu, 0.5\myu) rectangle +(0.09\myu, 0.5\myu)
        (0.77\myu, 0.5\myu) rectangle +(0.09\myu, 0.5\myu)
        (0.86\myu, 0.5\myu) rectangle +(0.09\myu, 0.5\myu);
    \path[container]
        (0\myu, 0\myu) rectangle +(0.5\myu, 1\myu)
        (0.5\myu, 0\myu) rectangle +(0.5\myu, 0.5\myu)
        (0.5\myu, 0.5\myu) rectangle +(0.5\myu, 0.5\myu);
    \path[compartment] (0\myu, 0\myu) rectangle (1\myu, 1\myu);
}}
\tikzset{pic3/.pic={
    \path[item]
        (0.0\myu, 0\myu) rectangle +(0.1\myu, 0.4\myu)
        (0.1\myu, 0\myu) rectangle +(0.1\myu, 0.4\myu)
        (0.2\myu, 0\myu) rectangle +(0.1\myu, 0.4\myu)
        (0.3\myu, 0\myu) rectangle +(0.1\myu, 0.4\myu);
    \path[item]
        (0.00\myu, 0.4\myu) rectangle +(0.09\myu, 0.3\myu)
        (0.09\myu, 0.4\myu) rectangle +(0.09\myu, 0.3\myu)
        (0.18\myu, 0.4\myu) rectangle +(0.09\myu, 0.3\myu)
        (0.27\myu, 0.4\myu) rectangle +(0.09\myu, 0.3\myu);
    \path[item]
        (0.00\myu, 0.7\myu) rectangle +(0.08\myu, 0.3\myu)
        (0.08\myu, 0.7\myu) rectangle +(0.08\myu, 0.3\myu)
        (0.16\myu, 0.7\myu) rectangle +(0.08\myu, 0.3\myu)
        (0.24\myu, 0.7\myu) rectangle +(0.08\myu, 0.3\myu)
        (0.32\myu, 0.7\myu) rectangle +(0.08\myu, 0.3\myu);
    \path[item]
        (0.4\myu, 0\myu) rectangle +(0.1\myu, 0.5\myu)
        (0.5\myu, 0\myu) rectangle +(0.1\myu, 0.5\myu)
        (0.6\myu, 0\myu) rectangle +(0.1\myu, 0.5\myu)
        (0.7\myu, 0\myu) rectangle +(0.1\myu, 0.5\myu)
        (0.8\myu, 0\myu) rectangle +(0.1\myu, 0.5\myu)
        (0.9\myu, 0\myu) rectangle +(0.1\myu, 0.5\myu);
    \path[item]
        (0.40\myu, 0.5\myu) rectangle +(0.07\myu, 0.5\myu)
        (0.47\myu, 0.5\myu) rectangle +(0.07\myu, 0.5\myu)
        (0.54\myu, 0.5\myu) rectangle +(0.09\myu, 0.5\myu)
        (0.63\myu, 0.5\myu) rectangle +(0.09\myu, 0.5\myu)
        (0.72\myu, 0.5\myu) rectangle +(0.08\myu, 0.5\myu)
        (0.80\myu, 0.5\myu) rectangle +(0.08\myu, 0.5\myu)
        (0.88\myu, 0.5\myu) rectangle +(0.08\myu, 0.5\myu);
    \path[container]
        (0\myu, 0\myu) rectangle +(0.4\myu, 0.4\myu)
        (0\myu, 0.4\myu) rectangle +(0.4\myu, 0.3\myu)
        (0\myu, 0.7\myu) rectangle +(0.4\myu, 0.3\myu)
        (0.4\myu, 0\myu) rectangle +(0.6\myu, 0.5\myu)
        (0.4\myu, 0.5\myu) rectangle +(0.6\myu, 0.5\myu);
    \path[compartment] (0\myu, 0\myu) rectangle (1\myu, 1\myu);
}}
\tikzset{pic4/.pic={
    \path[item] foreach \h/\y in {0.15/0.00,0.14/0.15,0.13/0.29,0.12/0.42,0.11/0.54,
            0.10/0.65,0.09/0.75,0.08/0.84,0.07/0.92}{
         foreach \hdiff/\w/\x in {0.000/0.15/0.00,0.001/0.14/0.15,0.002/0.13/0.29,0.003/0.12/0.42,
                0.004/0.11/0.54,0.005/0.10/0.65,0.006/0.09/0.75,0.007/0.08/0.84,0.008/0.07/0.92}{
            (\x\myu, \y\myu) rectangle +(\w\myu, \h\myu-\hdiff\myu)
    }};
}}
\begin{tikzpicture}
\pic at (2\myu, 0\myu) {pic4};

\pic[xscale=-3,rotate=90] at (2\myu, 2\myu) {pic1};
\pic[xscale=-3,rotate=90] at (1\myu, 4\myu) {pic2};

\pic[xscale=-2,rotate=90] at (1\myu, 1\myu) {pic3};
\pic[xscale=-2,rotate=90] at (2\myu, 3\myu) {pic1};
\pic[xscale=-2,rotate=90] at (0\myu, 0\myu) {pic2};

\pic[yscale=2] at (1\myu, 2\myu) {pic1};
\pic[yscale=2] at (3\myu, 0\myu) {pic2};
\pic[yscale=2] at (4\myu, 0\myu) {pic3};
\pic[yscale=2] at (4\myu, 3\myu) {pic1};
\pic[yscale=4] at (0\myu, 1\myu) {pic2};

\draw[ultra thick] (0\myu, 0\myu) rectangle (5\myu, 5\myu);
\end{tikzpicture}

\caption{Guess the packing of empty compartments in each bin (\cref{sec:enum-configs}).}
\end{subfigure}
\hfill
\begin{subfigure}[t]{0.3\textwidth}
\centering
\tikzset{compartment/.style={draw,thick},
container/.style={draw,fill={black!25}},
item/.style={}}
\ifcsname myu\endcsname\else\newlength{\myu}\fi
\setlength{\myu}{0.8cm}
\tikzset{pic1/.pic={
    \path[item]
        (0.00\myu, 0\myu) rectangle +(0.09\myu, 0.6\myu)
        (0.09\myu, 0\myu) rectangle +(0.10\myu, 0.6\myu)
        (0.19\myu, 0\myu) rectangle +(0.08\myu, 0.6\myu)
        (0.27\myu, 0\myu) rectangle +(0.08\myu, 0.6\myu);
    \path[item]
        (0\myu, 0.6\myu) rectangle +(0.10\myu, 0.4\myu)
        (0.10\myu, 0.6\myu) rectangle +(0.07\myu, 0.4\myu)
        (0.17\myu, 0.6\myu) rectangle +(0.07\myu, 0.4\myu)
        (0.24\myu, 0.6\myu) rectangle +(0.08\myu, 0.4\myu)
        (0.32\myu, 0.6\myu) rectangle +(0.07\myu, 0.4\myu);
    \path[item]
        (0.40\myu, 0\myu) rectangle +(0.11\myu, 0.5\myu)
        (0.50\myu, 0\myu) rectangle +(0.09\myu, 0.5\myu)
        (0.59\myu, 0\myu) rectangle +(0.09\myu, 0.5\myu);
    \path[item]
        (0.40\myu, 0.5\myu) rectangle +(0.08\myu, 0.5\myu)
        (0.48\myu, 0.5\myu) rectangle +(0.09\myu, 0.5\myu)
        (0.57\myu, 0.5\myu) rectangle +(0.08\myu, 0.5\myu);
    \path[item]
        (0.70\myu, 0\myu) rectangle +(0.1\myu, 0.4\myu)
        (0.80\myu, 0\myu) rectangle +(0.1\myu, 0.4\myu)
        (0.90\myu, 0\myu) rectangle +(0.1\myu, 0.4\myu);
    \path[item]
        (0.70\myu, 0.4\myu) rectangle +(0.07\myu, 0.3\myu)
        (0.77\myu, 0.4\myu) rectangle +(0.07\myu, 0.3\myu)
        (0.84\myu, 0.4\myu) rectangle +(0.07\myu, 0.3\myu)
        (0.91\myu, 0.4\myu) rectangle +(0.07\myu, 0.3\myu);
    \path[item]
        (0.70\myu, 0.7\myu) rectangle +(0.08\myu, 0.3\myu)
        (0.78\myu, 0.7\myu) rectangle +(0.08\myu, 0.3\myu)
        (0.86\myu, 0.7\myu) rectangle +(0.09\myu, 0.3\myu);
    \path[container]
        (0\myu, 0\myu) rectangle +(0.4\myu, 0.6\myu)
        (0\myu, 0.6\myu) rectangle +(0.4\myu, 0.4\myu)
        (0.4\myu, 0\myu) rectangle +(0.3\myu, 0.5\myu)
        (0.4\myu, 0.5\myu) rectangle +(0.3\myu, 0.5\myu)
        (0.7\myu, 0\myu) rectangle +(0.3\myu, 0.4\myu)
        (0.7\myu, 0.4\myu) rectangle +(0.3\myu, 0.3\myu)
        (0.7\myu, 0.7\myu) rectangle +(0.3\myu, 0.3\myu);
    \path[compartment] (0\myu, 0\myu) rectangle (1\myu, 1\myu);
}}
\tikzset{pic2/.pic={
    \path[item]
        (0.0\myu, 0\myu) rectangle +(0.1\myu, 1\myu)
        (0.1\myu, 0\myu) rectangle +(0.1\myu, 1\myu)
        (0.2\myu, 0\myu) rectangle +(0.1\myu, 1\myu)
        (0.3\myu, 0\myu) rectangle +(0.1\myu, 1\myu)
        (0.4\myu, 0\myu) rectangle +(0.1\myu, 1\myu);
    \path[item]
        (0.50\myu, 0\myu) rectangle +(0.08\myu, 0.5\myu)
        (0.58\myu, 0\myu) rectangle +(0.08\myu, 0.5\myu)
        (0.66\myu, 0\myu) rectangle +(0.08\myu, 0.5\myu)
        (0.74\myu, 0\myu) rectangle +(0.08\myu, 0.5\myu)
        (0.82\myu, 0\myu) rectangle +(0.08\myu, 0.5\myu)
        (0.90\myu, 0\myu) rectangle +(0.08\myu, 0.5\myu);
    \path[item]
        (0.50\myu, 0.5\myu) rectangle +(0.09\myu, 0.5\myu)
        (0.59\myu, 0.5\myu) rectangle +(0.09\myu, 0.5\myu)
        (0.68\myu, 0.5\myu) rectangle +(0.09\myu, 0.5\myu)
        (0.77\myu, 0.5\myu) rectangle +(0.09\myu, 0.5\myu)
        (0.86\myu, 0.5\myu) rectangle +(0.09\myu, 0.5\myu);
    \path[container]
        (0\myu, 0\myu) rectangle +(0.5\myu, 1\myu)
        (0.5\myu, 0\myu) rectangle +(0.5\myu, 0.5\myu)
        (0.5\myu, 0.5\myu) rectangle +(0.5\myu, 0.5\myu);
    \path[compartment] (0\myu, 0\myu) rectangle (1\myu, 1\myu);
}}
\tikzset{pic3/.pic={
    \path[item]
        (0.0\myu, 0\myu) rectangle +(0.1\myu, 0.4\myu)
        (0.1\myu, 0\myu) rectangle +(0.1\myu, 0.4\myu)
        (0.2\myu, 0\myu) rectangle +(0.1\myu, 0.4\myu)
        (0.3\myu, 0\myu) rectangle +(0.1\myu, 0.4\myu);
    \path[item]
        (0.00\myu, 0.4\myu) rectangle +(0.09\myu, 0.3\myu)
        (0.09\myu, 0.4\myu) rectangle +(0.09\myu, 0.3\myu)
        (0.18\myu, 0.4\myu) rectangle +(0.09\myu, 0.3\myu)
        (0.27\myu, 0.4\myu) rectangle +(0.09\myu, 0.3\myu);
    \path[item]
        (0.00\myu, 0.7\myu) rectangle +(0.08\myu, 0.3\myu)
        (0.08\myu, 0.7\myu) rectangle +(0.08\myu, 0.3\myu)
        (0.16\myu, 0.7\myu) rectangle +(0.08\myu, 0.3\myu)
        (0.24\myu, 0.7\myu) rectangle +(0.08\myu, 0.3\myu)
        (0.32\myu, 0.7\myu) rectangle +(0.08\myu, 0.3\myu);
    \path[item]
        (0.4\myu, 0\myu) rectangle +(0.1\myu, 0.5\myu)
        (0.5\myu, 0\myu) rectangle +(0.1\myu, 0.5\myu)
        (0.6\myu, 0\myu) rectangle +(0.1\myu, 0.5\myu)
        (0.7\myu, 0\myu) rectangle +(0.1\myu, 0.5\myu)
        (0.8\myu, 0\myu) rectangle +(0.1\myu, 0.5\myu)
        (0.9\myu, 0\myu) rectangle +(0.1\myu, 0.5\myu);
    \path[item]
        (0.40\myu, 0.5\myu) rectangle +(0.07\myu, 0.5\myu)
        (0.47\myu, 0.5\myu) rectangle +(0.07\myu, 0.5\myu)
        (0.54\myu, 0.5\myu) rectangle +(0.09\myu, 0.5\myu)
        (0.63\myu, 0.5\myu) rectangle +(0.09\myu, 0.5\myu)
        (0.72\myu, 0.5\myu) rectangle +(0.08\myu, 0.5\myu)
        (0.80\myu, 0.5\myu) rectangle +(0.08\myu, 0.5\myu)
        (0.88\myu, 0.5\myu) rectangle +(0.08\myu, 0.5\myu);
    \path[container]
        (0\myu, 0\myu) rectangle +(0.4\myu, 0.4\myu)
        (0\myu, 0.4\myu) rectangle +(0.4\myu, 0.3\myu)
        (0\myu, 0.7\myu) rectangle +(0.4\myu, 0.3\myu)
        (0.4\myu, 0\myu) rectangle +(0.6\myu, 0.5\myu)
        (0.4\myu, 0.5\myu) rectangle +(0.6\myu, 0.5\myu);
    \path[compartment] (0\myu, 0\myu) rectangle (1\myu, 1\myu);
}}
\tikzset{pic4/.pic={
    \path[item] foreach \h/\y in {0.15/0.00,0.14/0.15,0.13/0.29,0.12/0.42,0.11/0.54,
            0.10/0.65,0.09/0.75,0.08/0.84,0.07/0.92}{
         foreach \hdiff/\w/\x in {0.000/0.15/0.00,0.001/0.14/0.15,0.002/0.13/0.29,0.003/0.12/0.42,
                0.004/0.11/0.54,0.005/0.10/0.65,0.006/0.09/0.75,0.007/0.08/0.84,0.008/0.07/0.92}{
            (\x\myu, \y\myu) rectangle +(\w\myu, \h\myu-\hdiff\myu)
    }};
}}
\begin{tikzpicture}
\pic at (2\myu, 0\myu) {pic4};

\pic[xscale=-3,rotate=90] at (2\myu, 2\myu) {pic1};
\pic[xscale=-3,rotate=90] at (1\myu, 4\myu) {pic2};

\pic[xscale=-2,rotate=90] at (1\myu, 1\myu) {pic3};
\pic[xscale=-2,rotate=90] at (2\myu, 3\myu) {pic1};
\pic[xscale=-2,rotate=90] at (0\myu, 0\myu) {pic2};

\pic[yscale=2] at (1\myu, 2\myu) {pic1};
\pic[yscale=2] at (3\myu, 0\myu) {pic2};
\pic[yscale=2] at (4\myu, 0\myu) {pic3};
\pic[yscale=2] at (4\myu, 3\myu) {pic1};
\pic[yscale=4] at (0\myu, 1\myu) {pic2};

\draw[ultra thick] (0\myu, 0\myu) rectangle (5\myu, 5\myu);
\end{tikzpicture}

\caption{Fractionally pack wide and tall items into compartments.
This partitions each compartment into \emph{containers} (\cref{sec:feas-lp}).}
\end{subfigure}
\hfill
\begin{subfigure}[t]{0.3\textwidth}
\centering
\tikzset{compartment/.style={draw,ultra thick},
container/.style={draw,thick},
item/.style={draw,very thin,fill={black!25}}}
\ifcsname myu\endcsname\else\newlength{\myu}\fi
\setlength{\myu}{0.8cm}
\tikzset{pic1/.pic={
    \path[item]
        (0.00\myu, 0\myu) rectangle +(0.09\myu, 0.6\myu)
        (0.09\myu, 0\myu) rectangle +(0.10\myu, 0.6\myu)
        (0.19\myu, 0\myu) rectangle +(0.08\myu, 0.6\myu)
        (0.27\myu, 0\myu) rectangle +(0.08\myu, 0.6\myu);
    \path[item]
        (0\myu, 0.6\myu) rectangle +(0.10\myu, 0.4\myu)
        (0.10\myu, 0.6\myu) rectangle +(0.07\myu, 0.4\myu)
        (0.17\myu, 0.6\myu) rectangle +(0.07\myu, 0.4\myu)
        (0.24\myu, 0.6\myu) rectangle +(0.08\myu, 0.4\myu)
        (0.32\myu, 0.6\myu) rectangle +(0.07\myu, 0.4\myu);
    \path[item]
        (0.40\myu, 0\myu) rectangle +(0.11\myu, 0.5\myu)
        (0.50\myu, 0\myu) rectangle +(0.09\myu, 0.5\myu)
        (0.59\myu, 0\myu) rectangle +(0.09\myu, 0.5\myu);
    \path[item]
        (0.40\myu, 0.5\myu) rectangle +(0.08\myu, 0.5\myu)
        (0.48\myu, 0.5\myu) rectangle +(0.09\myu, 0.5\myu)
        (0.57\myu, 0.5\myu) rectangle +(0.08\myu, 0.5\myu);
    \path[item]
        (0.70\myu, 0\myu) rectangle +(0.1\myu, 0.4\myu)
        (0.80\myu, 0\myu) rectangle +(0.1\myu, 0.4\myu)
        (0.90\myu, 0\myu) rectangle +(0.1\myu, 0.4\myu);
    \path[item]
        (0.70\myu, 0.4\myu) rectangle +(0.07\myu, 0.3\myu)
        (0.77\myu, 0.4\myu) rectangle +(0.07\myu, 0.3\myu)
        (0.84\myu, 0.4\myu) rectangle +(0.07\myu, 0.3\myu)
        (0.91\myu, 0.4\myu) rectangle +(0.07\myu, 0.3\myu);
    \path[item]
        (0.70\myu, 0.7\myu) rectangle +(0.08\myu, 0.3\myu)
        (0.78\myu, 0.7\myu) rectangle +(0.08\myu, 0.3\myu)
        (0.86\myu, 0.7\myu) rectangle +(0.09\myu, 0.3\myu);
    \path[container]
        (0\myu, 0\myu) rectangle +(0.4\myu, 0.6\myu)
        (0\myu, 0.6\myu) rectangle +(0.4\myu, 0.4\myu)
        (0.4\myu, 0\myu) rectangle +(0.3\myu, 0.5\myu)
        (0.4\myu, 0.5\myu) rectangle +(0.3\myu, 0.5\myu)
        (0.7\myu, 0\myu) rectangle +(0.3\myu, 0.4\myu)
        (0.7\myu, 0.4\myu) rectangle +(0.3\myu, 0.3\myu)
        (0.7\myu, 0.7\myu) rectangle +(0.3\myu, 0.3\myu);
    \path[compartment] (0\myu, 0\myu) rectangle (1\myu, 1\myu);
}}
\tikzset{pic2/.pic={
    \path[item]
        (0.0\myu, 0\myu) rectangle +(0.1\myu, 1\myu)
        (0.1\myu, 0\myu) rectangle +(0.1\myu, 1\myu)
        (0.2\myu, 0\myu) rectangle +(0.1\myu, 1\myu)
        (0.3\myu, 0\myu) rectangle +(0.1\myu, 1\myu)
        (0.4\myu, 0\myu) rectangle +(0.1\myu, 1\myu);
    \path[item]
        (0.50\myu, 0\myu) rectangle +(0.08\myu, 0.5\myu)
        (0.58\myu, 0\myu) rectangle +(0.08\myu, 0.5\myu)
        (0.66\myu, 0\myu) rectangle +(0.08\myu, 0.5\myu)
        (0.74\myu, 0\myu) rectangle +(0.08\myu, 0.5\myu)
        (0.82\myu, 0\myu) rectangle +(0.08\myu, 0.5\myu)
        (0.90\myu, 0\myu) rectangle +(0.08\myu, 0.5\myu);
    \path[item]
        (0.50\myu, 0.5\myu) rectangle +(0.09\myu, 0.5\myu)
        (0.59\myu, 0.5\myu) rectangle +(0.09\myu, 0.5\myu)
        (0.68\myu, 0.5\myu) rectangle +(0.09\myu, 0.5\myu)
        (0.77\myu, 0.5\myu) rectangle +(0.09\myu, 0.5\myu)
        (0.86\myu, 0.5\myu) rectangle +(0.09\myu, 0.5\myu);
    \path[container]
        (0\myu, 0\myu) rectangle +(0.5\myu, 1\myu)
        (0.5\myu, 0\myu) rectangle +(0.5\myu, 0.5\myu)
        (0.5\myu, 0.5\myu) rectangle +(0.5\myu, 0.5\myu);
    \path[compartment] (0\myu, 0\myu) rectangle (1\myu, 1\myu);
}}
\tikzset{pic3/.pic={
    \path[item]
        (0.0\myu, 0\myu) rectangle +(0.1\myu, 0.4\myu)
        (0.1\myu, 0\myu) rectangle +(0.1\myu, 0.4\myu)
        (0.2\myu, 0\myu) rectangle +(0.1\myu, 0.4\myu)
        (0.3\myu, 0\myu) rectangle +(0.1\myu, 0.4\myu);
    \path[item]
        (0.00\myu, 0.4\myu) rectangle +(0.09\myu, 0.3\myu)
        (0.09\myu, 0.4\myu) rectangle +(0.09\myu, 0.3\myu)
        (0.18\myu, 0.4\myu) rectangle +(0.09\myu, 0.3\myu)
        (0.27\myu, 0.4\myu) rectangle +(0.09\myu, 0.3\myu);
    \path[item]
        (0.00\myu, 0.7\myu) rectangle +(0.08\myu, 0.3\myu)
        (0.08\myu, 0.7\myu) rectangle +(0.08\myu, 0.3\myu)
        (0.16\myu, 0.7\myu) rectangle +(0.08\myu, 0.3\myu)
        (0.24\myu, 0.7\myu) rectangle +(0.08\myu, 0.3\myu)
        (0.32\myu, 0.7\myu) rectangle +(0.08\myu, 0.3\myu);
    \path[item]
        (0.4\myu, 0\myu) rectangle +(0.1\myu, 0.5\myu)
        (0.5\myu, 0\myu) rectangle +(0.1\myu, 0.5\myu)
        (0.6\myu, 0\myu) rectangle +(0.1\myu, 0.5\myu)
        (0.7\myu, 0\myu) rectangle +(0.1\myu, 0.5\myu)
        (0.8\myu, 0\myu) rectangle +(0.1\myu, 0.5\myu)
        (0.9\myu, 0\myu) rectangle +(0.1\myu, 0.5\myu);
    \path[item]
        (0.40\myu, 0.5\myu) rectangle +(0.07\myu, 0.5\myu)
        (0.47\myu, 0.5\myu) rectangle +(0.07\myu, 0.5\myu)
        (0.54\myu, 0.5\myu) rectangle +(0.09\myu, 0.5\myu)
        (0.63\myu, 0.5\myu) rectangle +(0.09\myu, 0.5\myu)
        (0.72\myu, 0.5\myu) rectangle +(0.08\myu, 0.5\myu)
        (0.80\myu, 0.5\myu) rectangle +(0.08\myu, 0.5\myu)
        (0.88\myu, 0.5\myu) rectangle +(0.08\myu, 0.5\myu);
    \path[container]
        (0\myu, 0\myu) rectangle +(0.4\myu, 0.4\myu)
        (0\myu, 0.4\myu) rectangle +(0.4\myu, 0.3\myu)
        (0\myu, 0.7\myu) rectangle +(0.4\myu, 0.3\myu)
        (0.4\myu, 0\myu) rectangle +(0.6\myu, 0.5\myu)
        (0.4\myu, 0.5\myu) rectangle +(0.6\myu, 0.5\myu);
    \path[compartment] (0\myu, 0\myu) rectangle (1\myu, 1\myu);
}}
\tikzset{pic4/.pic={
    \path[item] foreach \h/\y in {0.15/0.00,0.14/0.15,0.13/0.29,0.12/0.42,0.11/0.54,
            0.10/0.65,0.09/0.75,0.08/0.84,0.07/0.92}{
         foreach \hdiff/\w/\x in {0.000/0.15/0.00,0.001/0.14/0.15,0.002/0.13/0.29,0.003/0.12/0.42,
                0.004/0.11/0.54,0.005/0.10/0.65,0.006/0.09/0.75,0.007/0.08/0.84,0.008/0.07/0.92}{
            (\x\myu, \y\myu) rectangle +(\w\myu, \h\myu-\hdiff\myu)
    }};
}}
\begin{tikzpicture}
\pic at (2\myu, 0\myu) {pic4};

\pic[xscale=-3,rotate=90] at (2\myu, 2\myu) {pic1};
\pic[xscale=-3,rotate=90] at (1\myu, 4\myu) {pic2};

\pic[xscale=-2,rotate=90] at (1\myu, 1\myu) {pic3};
\pic[xscale=-2,rotate=90] at (2\myu, 3\myu) {pic1};
\pic[xscale=-2,rotate=90] at (0\myu, 0\myu) {pic2};

\pic[yscale=2] at (1\myu, 2\myu) {pic1};
\pic[yscale=2] at (3\myu, 0\myu) {pic2};
\pic[yscale=2] at (4\myu, 0\myu) {pic3};
\pic[yscale=2] at (4\myu, 3\myu) {pic1};
\pic[yscale=4] at (0\myu, 1\myu) {pic2};

\draw[ultra thick] (0\myu, 0\myu) rectangle (5\myu, 5\myu);
\end{tikzpicture}

\caption{Pack the items non-fractionally (\cref{sec:greedy-cont}).}
\end{subfigure}
\caption{Major steps of $\thinCPack$ after $\round$ing $I$.}
\label{fig:thincpack}
\end{figure}

\appendix
\newpage
\section{Examples of Guillotinable and Non-Guillotinable Packing}
\label{sec:guill-examples}

\begin{figure}[htb]
\centering
\begin{tikzpicture}[
item/.style={fill={black!10},draw},
bin/.style={draw,thick},
]
\path[item]
    (0.0,0.0) rectangle +(1.8,1.2)
    (1.8,0.0) rectangle +(1.2,1.8)
    (0.0,1.2) rectangle +(1.2,1.8)
    (1.2,1.8) rectangle +(1.8,1.2);
\path[bin] (0,0) rectangle +(3,3);
\begin{scope}[shift={(4cm,0cm)},scale=1.25]
\path[item]
    (0.0,0.0) rectangle +(1.2,0.6)
    (1.8,0.0) rectangle +(1.2,0.6)
    (0.0,0.6) rectangle +(0.6,1.8)
    (2.4,0.6) rectangle +(0.6,1.8)
    (1.2,0.0) rectangle +(0.6,1.8)
    (0.6,1.8) rectangle +(1.8,0.6);
\path[bin] (0,0) rectangle +(3,2.4);
\end{scope}
\end{tikzpicture}

\caption{Two bins that are not guillotinable.}
\label{fig:non-guill}
\end{figure}
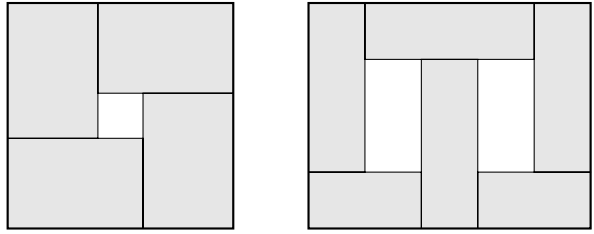

\begin{figure}[htb]
\centering
\tikzset{pics/mypic/.style args={#1,#2,#3}{code={
\begin{scope}
    \begin{scope}
        \begin{scope}
            \path[item] (0, 0) rectangle +(0.8, 2.4);
            \path[bin3] (0, 0) rectangle +(0.8, 2.4);
        \end{scope}
        \begin{scope}[xshift={#3cm}]
            \path[item] (0.8, 0) rectangle +(0.6, 2.4);
            \path[bin3] (0.8, 0) rectangle +(0.6, 2.4);
        \end{scope}
        \begin{scope}[xshift={2*#3cm}]
            \path[item] (1.4, 0) rectangle +(0.8, 2);
            \path[bin3] (1.4, 0) rectangle +(1, 2.4);
        \end{scope}
        \path[bin2] (0, 0) rectangle +(2.4, 2.4);
        \path[cutline2] (0.8, -0.2) -- (0.8, 2.6);
        \path[cutline2] (1.4, -0.2) -- (1.4, 2.6);
    \end{scope}
    \begin{scope}[yshift={#2cm}]
        \begin{scope}
            \path[item] (0, 2.4) rectangle +(1.1, 1.6);
            \path[bin3] (0, 2.4) rectangle +(1.1, 1.6);
        \end{scope}
        \begin{scope}[xshift={#3cm}]
            \path[item] (1.1, 2.4) rectangle +(1.3, 1.4);
            \path[bin3] (1.1, 2.4) rectangle +(1.3, 1.6);
        \end{scope}
        \path[bin2] (0, 2.4) rectangle +(2.4, 1.6);
        \path[cutline2] (1.1, 2.2) -- (1.1, 4.2);
    \end{scope}
    \path[cutline1] (-0.2, 2.4) -- (2.6, 2.4);
    \path[bin1] (0, 0) rectangle +(2.4, 4);
\end{scope}
\begin{scope}[xshift={#1cm}]
    \begin{scope}
        \begin{scope}
            \path[item] (2.4, 0) rectangle +(0.8, 2.8);
            \path[bin3] (2.4, 0) rectangle +(0.8, 3.2);
        \end{scope}
        \begin{scope}[xshift={#3cm}]
            \path[item] (3.2, 0) rectangle +(0.8, 3.2);
            \path[bin3] (3.2, 0) rectangle +(0.8, 3.2);
        \end{scope}
        \path[bin2] (2.4, 0) rectangle +(1.6, 3.2);
        \path[cutline2] (3.2, -0.2) -- (3.2, 3.4);
    \end{scope}
    \begin{scope}[yshift={#2cm}]
        \path[item] (2.4, 3.2) rectangle +(1.6, 0.6);
        \path[bin2,bin3] (2.4, 3.2) rectangle +(1.6, 0.8);
    \end{scope}
    \path[bin1] (2.4, 0) rectangle +(1.6, 4);
    \path[cutline1] (2.2, 3.2) -- (4.2, 3.2);
\end{scope}
\path[bin0] (0, 0) rectangle +(4, 4);
\path[cutline0] (2.4, -0.2) -- (2.4, 4.2);
}}}
\begin{tikzpicture}[
item/.style={fill={black!10},draw},
myarrow/.style={->,>={Stealth},thick},
bin/.style={draw,thick},
cutline/.style={draw={black!50!red},densely dashed,line width=1.1pt},
bin0/.style={},
bin1/.style={},
bin2/.style={},
bin3/.style={},
cutline0/.style={},
cutline1/.style={},
cutline2/.style={},
scale=0.75,
transform shape,
]
\begin{scope}[bin0/.style=bin,cutline0/.style=cutline]
    \pic[yscale=-1]{mypic={0,0,0}};
    \node[rotate=90,transform shape] at (2.42,-4.35) {\large\ding{34}};
\end{scope}
\draw[myarrow] (5,-2) -> (8,-2);
\begin{scope}[xshift=9.5cm,bin1/.style=bin,cutline1/.style=cutline]
    \pic[yscale=-1]{mypic={1,0,0}};
    \node[transform shape] at (-0.4,-2.42) {\large\ding{34}};
    \node[xscale=-1,transform shape] at (5.4,-3.22) {\large\ding{34}};
\end{scope}
\draw[myarrow] (12.1,-4.5) -> (12.1,-6);
\begin{scope}[shift={(9.5cm,-7cm)},bin2/.style=bin,cutline2/.style=cutline]
    \pic[yscale=-1]{mypic={1,0.5,0}};
    \node[yscale=-1,rotate=90,transform shape] at (0.82,0.4) {\large\ding{34}};
    \node[yscale=-1,rotate=90,transform shape] at (1.42,0.4) {\large\ding{34}};
    \node[yscale=-1,rotate=90,transform shape] at (4.22,0.4) {\large\ding{34}};
    \node[rotate=90,transform shape] at (1.12,-4.9) {\large\ding{34}};
\end{scope}
\draw[myarrow] (8.5,-9.5) -> (6.5,-9.5);
\begin{scope}[yshift=-7cm,bin3/.style=bin]
    \pic[yscale=-1]{mypic={1.2,0.5,0.4}};
\end{scope}
\end{tikzpicture}

\caption{Separating items using 3 stages of guillotine cuts.}
\label{fig:guill}
\end{figure}
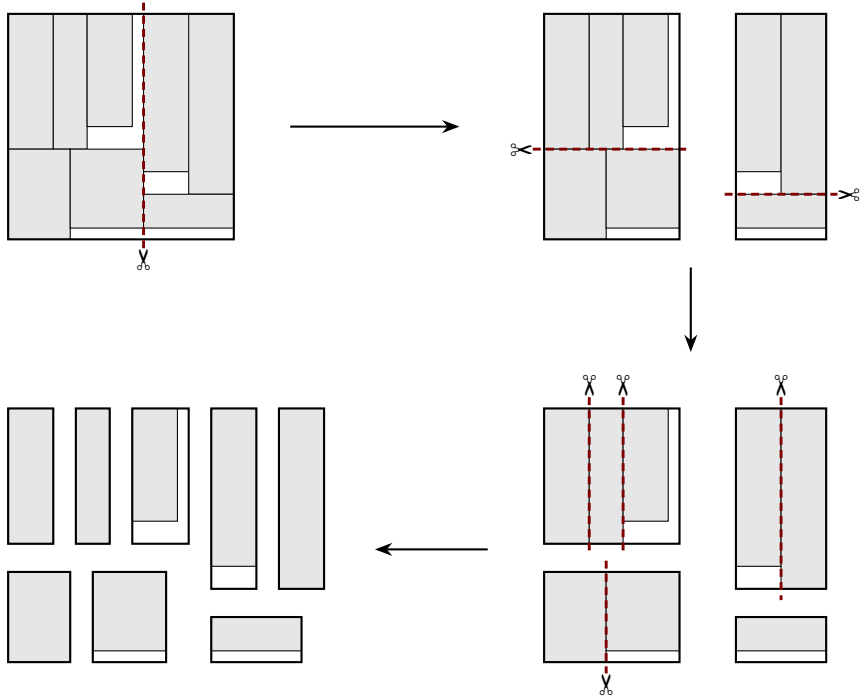

\section{Next-Fit Decreasing Height (NFDH)}
\label{sec:nfdh}

In this section, we give proofs of the results in \cref{sec:preliminaries}
related to NFDH algorithm~\cite{coffman1980performance}.

\begin{lemma}[\cite{johnson-thesis}]
\label{thm:next-fit}
Let $I$ be a classical bin packing instance. Then the number of bins used by the
Next-Fit algorithm~\cite{johnson-thesis} to pack $I$ is less than $2\size(I) + 1$.
\end{lemma}

\begin{lemma}[\cite{johnson-thesis}]
\label{thm:next-fit-small}
Let $I$ be a classical bin packing instance, where each item has size at most $\eps$.
Then the number of bins used by the Next-Fit algorithm~\cite{johnson-thesis} to pack $I$
is less than $\size(I)/(1-\eps) + 1$.
\end{lemma}

\begin{lemma}[\cite{coffman1980performance}]
\label{thm:nfdh-strip}
Let $I$ be a set of rectangles of width at most 1.
Then $I$ can be packed (without rotation) into a rectangular bin of width 1
and height less than $2a(I) + max_{i \in I} h(i)$ using NFDH.
\end{lemma}

\begin{lemma}
\label{thm:nfdh-strip-small}
Let $I$ be a set of rectangles of width at most $\eps$.
Then $I$ can be packed (without rotation) into a rectangular bin of width $W$
and height less than $a(I)/(W-\eps) + max_{i \in I} h(i)$ using NFDH.
\end{lemma}
\begin{proof}
Let there be $p$ shelves output by NFDH.
Let $S_j$ be the items in the $j\Th$ shelf.
Let $h_j$ be the height of the $j\Th$ shelf.
Let $H$ be the sum of heights of all the shelves.

For $j \in [p-1]$, in the $j\Th$ shelf, the total width of items is
more than $(W-\eps)$ and each item has height more than $h_{j+1}$.
Therefore, $a(S_j) > h_{j+1}(W-\eps)$.

Let $H$ be the sum of heights of all the shelves. Then
$a(I) > \sum_{i=1}^{p-1} a(S_j) \ge \sum_{i=1}^{p-1} h_{j+1}(W-\eps) \ge (W-\eps)(H - h_1)$.
This implies,$H < ({a(I)}/{W-\eps}) + h_1$.
\end{proof}

\rthmNfdhSmall*
\begin{proof}
NFDH packs the items into shelves of width $W$.
Let $\Htild$ be the total height of the shelves.
By \cref{thm:nfdh-strip-small}, we get
$\Htild < a(I)/(W-\delta_W) + \delta_H \le H$.
Therefore, NFDH can fit $I$ into the bin.
\end{proof}

\begin{lemma}
Let $I$ be a set of rectangular items where each item has height at most $\delta$.
Then the number of bins required by NFDH to pack $I$ is less than
$(2a(I)+1)/(1-\delta)$.
\end{lemma}
\begin{proof}
The bin packing version of NFDH first packs $I$ into shelves
and then packs the shelves into bins using Next-Fit.
Let $H$ be the sum of heights of all the shelves.
By \cref{thm:nfdh-strip}, $H < 2a(I) + \delta$.
By \cref{thm:next-fit-small}, the number of bins
is less than $1 + H/(1-\delta) < (2a(I)+1)/(1-\delta)$.
\end{proof}

\begin{lemma}
Let $I$ be a set of rectangular items where each item has width at most $\delta$.
Then the number of bins required by NFDH to pack $I$ is less than $2a(I)/(1-\delta) + 3$.
\end{lemma}
\begin{proof}
The bin packing version of NFDH first packs $I$ into shelves
and then packs the shelves into bins using Next-Fit.
Let $H$ be the sum of heights of all the shelves.
By \cref{thm:nfdh-strip-small}, $H < a(I)/(1-\delta) + 1$.
By \cref{thm:next-fit}, the number of bins
is less than $2H + 1 < 2a(I)/(1-\delta) + 3$.
\end{proof}

\begin{lemma}
\label{thm:nfdh-small-2}
Let $I$ be a set of rectangular items where each item has width at most $\delta_W$
and height at most $\delta_H$. Then the number of bins required by NFDH to pack $I$
is at most $a(I)/(1-\delta_W)(1-\delta_H) + 1/(1-\delta_H)$.
\end{lemma}
\begin{proof}
The bin packing version of NFDH first packs $I$ into shelves
and then packs the shelves into bins using Next-Fit.
Let $H$ be the sum of heights of all the shelves.
By \cref{thm:nfdh-strip-small}, $H < a(I)/(1-\delta_W) + \delta_H$.
By \cref{thm:next-fit-small}, the number of bins
is less than $1 + H/(1-\delta_H) < a(I)/(1-\delta_W)(1-\delta_H) + 1/(1-\delta_H)$.
\end{proof}

\section{Linear Grouping}
\label{sec:lingroup}

In this section, we describe the \emph{linear grouping}
technique~\cite{bp-aptas,kenyon1996strip} for wide and tall items.

Let $\eps$ and $\epsLarge$ be constants in $(0, 1)$.
Let $W$ be a set of items where each item has width more than $\epsLarge$.
We will describe an algorithm, called $\lingroupWide$ that takes
$W$, $\eps$ and $\epsLarge$ as input and returns the set $\What$ as output,
where $\What$ is obtained by increasing the width of each item in $W$.

$\lingroupWide(W, \eps, \epsLarge)$ first arranges the items $W$
in decreasing order of width and stacks them one-over-the-other
(i.e., the widest item in $W$ is at the bottom).
Let $h_L$ be the height of the stack.
Let $y(i)$ be the $y$-coordinate of the bottom edge of item $i$.
Split the stack into sections of height $\eps\epsLarge h_L$ each.
For $j \in [1/\eps\epsLarge]$, let $w_j$ be the width of the
widest item intersecting the $j\Th$ section from the bottom, i.e.,
\[ w_j \defeq \max(\{w(i): i \in W \textrm{ and } (y(i), y(i) + h(i))
    \cap ((j-1)\eps\epsLarge h_L, j\eps\epsLarge h_L) \neq \emptyset\}). \]
Round up the width of each item $i$ to the smallest $w_j$ that is at least $w(i)$
(see \cref{fig:lingroup}).
Let $W_j$ be the items whose width got rounded to $w_j$
and let $\What_j$ be the resulting rounded items.
(There may be ties, i.e., there may exist $j_1 < j_2$ such that $w_{j_1} = w_{j_2}$.
In that case, define $W_{j_2} \defeq \What_{j_2} = \emptyset$.
This ensures that all $W_j$ are disjoint.)
Finally, define $\What \defeq \bigcup_j \What_j$.

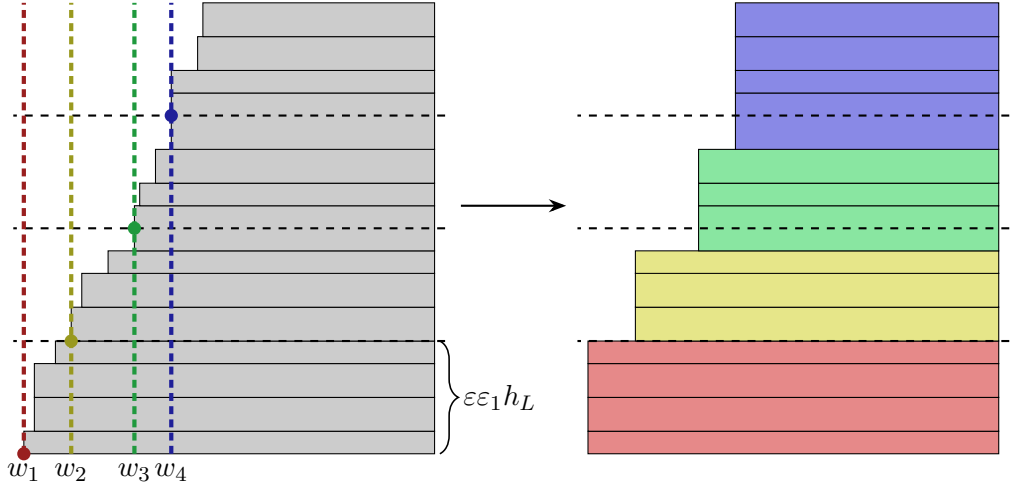
\begin{figure}[htb]
\centering
\ifcsname pGameL\endcsname\else\newlength{\pGameL}\fi
\ifcsname pGameM\endcsname\else\newlength{\pGameM}\fi
\setlength{\pGameL}{0.15cm}
\setlength{\pGameM}{0.07cm}
\tikzset{bin/.style={draw,thick}}
\tikzset{binGrid/.style={draw,xstep=1\pGameL,ystep=1\pGameM,{black!20}}}
\tikzset{item/.style={draw,fill={black!20}}}
\tikzset{sepline/.style={draw,dashed,thick}}
\tikzset{limline/.style={densely dashed,ultra thick}}
\tikzset{myarrow/.style={->,>={Stealth},thick}}
\tikzset{mybrace/.style={decoration={amplitude=7pt,brace,mirror,raise=2pt},semithick,decorate}}
\definecolor{color1}{Hsb}{0,0.4,0.9}
\definecolor{color1d}{Hsb}{0,0.8,0.6}
\definecolor{color2}{Hsb}{60,0.4,0.9}
\definecolor{color2d}{Hsb}{60,0.8,0.6}
\definecolor{color3}{Hsb}{135,0.4,0.9}
\definecolor{color3d}{Hsb}{135,0.8,0.6}
\definecolor{color4}{Hsb}{240,0.4,0.9}
\definecolor{color4d}{Hsb}{240,0.8,0.6}
\begin{tikzpicture}[rotate=90]
\begin{scope}
\path[item] (0\pGameL, 0\pGameM) rectangle +(2\pGameL, 78\pGameM);
\path[item] (2\pGameL, 0\pGameM) rectangle +(3\pGameL, 76\pGameM);
\path[item] (5\pGameL, 0\pGameM) rectangle +(3\pGameL, 76\pGameM);
\path[item] (8\pGameL, 0\pGameM) rectangle +(2\pGameL, 72\pGameM);
\path[item] (10\pGameL, 0\pGameM) rectangle +(3\pGameL, 69\pGameM);
\path[item] (13\pGameL, 0\pGameM) rectangle +(3\pGameL, 67\pGameM);
\path[item] (16\pGameL, 0\pGameM) rectangle +(2\pGameL, 62\pGameM);
\path[item] (18\pGameL, 0\pGameM) rectangle +(4\pGameL, 57\pGameM);
\path[item] (22\pGameL, 0\pGameM) rectangle +(2\pGameL, 56\pGameM);
\path[item] (24\pGameL, 0\pGameM) rectangle +(3\pGameL, 53\pGameM);
\path[item] (27\pGameL, 0\pGameM) rectangle +(5\pGameL, 50\pGameM);
\path[item] (32\pGameL, 0\pGameM) rectangle +(2\pGameL, 50\pGameM);
\path[item] (34\pGameL, 0\pGameM) rectangle +(3\pGameL, 45\pGameM);
\path[item] (37\pGameL, 0\pGameM) rectangle +(3\pGameL, 44\pGameM);
\path[sepline] (10\pGameL, -2\pGameM) -- (10\pGameL, 80\pGameM);
\path[sepline] (20\pGameL, -2\pGameM) -- (20\pGameL, 80\pGameM);
\path[sepline] (30\pGameL, -2\pGameM) -- (30\pGameL, 80\pGameM);
\draw[limline,color1d] (0\pGameL, 78\pGameM) -- (40\pGameL, 78\pGameM);
\draw[limline,color2d] (0\pGameL, 69\pGameM) -- (40\pGameL, 69\pGameM);
\draw[limline,color3d] (0\pGameL, 57\pGameM) -- (40\pGameL, 57\pGameM);
\draw[limline,color4d] (0\pGameL, 50\pGameM) -- (40\pGameL, 50\pGameM);
\node[anchor=north] at (0\pGameL, 78\pGameM) {$w_1$};
\node[anchor=north] at (0\pGameL, 69\pGameM) {$w_2$};
\node[anchor=north] at (0\pGameL, 57\pGameM) {$w_3$};
\node[anchor=north] at (0\pGameL, 50\pGameM) {$w_4$};
\fill[color1d] (0\pGameL, 78\pGameM) circle [radius=2.5pt];
\fill[color2d] (10\pGameL, 69\pGameM) circle [radius=2.5pt];
\fill[color3d] (20\pGameL, 57\pGameM) circle [radius=2.5pt];
\fill[color4d] (30\pGameL, 50\pGameM) circle [radius=2.5pt];
\draw[mybrace] (0,0) -- node[right=7pt] {$\eps\epsLarge h_L$} +(10\pGameL,0);
\end{scope}
\draw[myarrow] (22\pGameL, -5\pGameM) -- (22\pGameL, -25\pGameM);
\begin{scope}[yshift=-7.5cm]
\path[item,fill=color1] (0\pGameL, 0\pGameM) rectangle +(2\pGameL, 78\pGameM);
\path[item,fill=color1] (2\pGameL, 0\pGameM) rectangle +(3\pGameL, 78\pGameM);
\path[item,fill=color1] (5\pGameL, 0\pGameM) rectangle +(3\pGameL, 78\pGameM);
\path[item,fill=color1] (8\pGameL, 0\pGameM) rectangle +(2\pGameL, 78\pGameM);
\path[item,fill=color2] (10\pGameL, 0\pGameM) rectangle +(3\pGameL, 69\pGameM);
\path[item,fill=color2] (13\pGameL, 0\pGameM) rectangle +(3\pGameL, 69\pGameM);
\path[item,fill=color2] (16\pGameL, 0\pGameM) rectangle +(2\pGameL, 69\pGameM);
\path[item,fill=color3] (18\pGameL, 0\pGameM) rectangle +(4\pGameL, 57\pGameM);
\path[item,fill=color3] (22\pGameL, 0\pGameM) rectangle +(2\pGameL, 57\pGameM);
\path[item,fill=color3] (24\pGameL, 0\pGameM) rectangle +(3\pGameL, 57\pGameM);
\path[item,fill=color4] (27\pGameL, 0\pGameM) rectangle +(5\pGameL, 50\pGameM);
\path[item,fill=color4] (32\pGameL, 0\pGameM) rectangle +(2\pGameL, 50\pGameM);
\path[item,fill=color4] (34\pGameL, 0\pGameM) rectangle +(3\pGameL, 50\pGameM);
\path[item,fill=color4] (37\pGameL, 0\pGameM) rectangle +(3\pGameL, 50\pGameM);
\path[sepline] (10\pGameL, -2\pGameM) -- (10\pGameL, 80\pGameM);
\path[sepline] (20\pGameL, -2\pGameM) -- (20\pGameL, 80\pGameM);
\path[sepline] (30\pGameL, -2\pGameM) -- (30\pGameL, 80\pGameM);
\end{scope}
\end{tikzpicture}

\caption{Example invocation of $\lingroupWide$ for $\eps = \epsLarge = 1/2$.}
\label{fig:lingroup}
\end{figure}

We can similarly define the algorithm $\lingroupTall$.
Let $H$ be a set of items where each item has height more than $\epsLarge$.
$\lingroupTall$ that takes $H$, $\eps$ and $\epsLarge$ as input and returns $\Hhat$,
where $\Hhat$ is obtained by increasing the width of each item in $H$.

\begin{claim}
\label{thm:lingroup-n}
Items in $\lingroupWide(W, \eps, \epsLarge)$
have at most $1/(\eps\epsLarge)$ distinct widths.
\\ Items in $\lingroupTall(H, \eps, \epsLarge)$
have at most $1/(\eps\epsLarge)$ distinct heights.
\end{claim}

\begin{lemma}
\label{thm:lingroup-repack}
Let $W$, $H$ and $S$ be sets of items,
where items in $W$ have width more than $\epsLarge$
and items in $H$ have height more than $\epsLarge$.
Let $\What \defeq \lingroupWide(W, \eps, \epsLarge)$
and $\Hhat \defeq \lingroupTall(H, \eps, \epsLarge)$.
If we allow slicing items in $\What$ and $\Hhat$ using horizontal and vertical cuts,
respectively, then we can pack $\What \cup \Hhat \cup S$ into
less than $(1+\eps)\opt(W \cup H \cup S) + 2$ bins.
\end{lemma}
\begin{proof}
$\lingroupWide(W, \eps, \epsLarge)$ arranges the items $W$
in decreasing order of width and stacks them one-over-the-other.
Let $h_L$ be the height of the stack.
Split the stack into sections of height $\eps\epsLarge h_L$ each.
For $j \in [1/\eps\epsLarge]$, let $W_j'$ be the slices of items
that lie between heights $(j-1)\eps\epsLarge h_L$ and $j\eps\epsLarge h_L$.
Let $\What_j'$ be the corresponding rounded items from $W_j'$.
Similarly define $H_j'$ and $\Hhat_j'$.

Consider the optimal packing of $W \cup H \cup S$.
To convert this to a packing of $\What \cup \Hhat \cup S - (\What_1' \cup \Hhat_1')$,
unpack $W_1'$ and $H_1'$, and for each $j \in [1/\eps\epsLarge-1]$,
pack $\What_{j+1}'$ in the place of $W_j'$
and pack $\Hhat_{j+1}'$ in the place of $H_j'$,
possibly after slicing the items.
This gives us a packing of $\What \cup \Hhat \cup S - (\What_1' \cup \Hhat_1')$
into $m_1$ bins, where $m_1 \le \opt(W \cup H \cup S)$.

We can pack $\Hhat_1'$ by stacking the items side-by-side on the base of bins.
We can pack $\What_1'$ into bins by stacking the items one-over-the-other.
Let $h_L$ be the total height of items in $\What$.
Let $w_L$ be the total width of items in $\Hhat$.
This gives us a packing of $\What_1' \cup \Hhat_1'$ into
$m_2 \defeq \smallceil{\eps\epsLarge h_L} + \smallceil{\eps\epsLarge w_L}$ bins.
Since
$\opt(W \cup H \cup S) \ge a(W) + a(H) \ge \epsLarge(h_L + w_L)$,
we get
$m_2 = \smallceil{\eps\epsLarge h_L} + \smallceil{\eps\epsLarge w_L}
< \eps\opt(W \cup H \cup S) + 2$.
Therefore, we get a packing of $\What \cup \Hhat \cup S$ into $m_1 + m_2$ bins and
$m_1 + m_2 < (1+\eps)\opt(W \cup H \cup S) + 2$.
\end{proof}

\section{Details of \texorpdfstring{$\thinGPack$}{skewed4Pack}}
\label{sec:guill-thin-extra}

\subsection{Creating Shelves}
\label{sec:guill-thin-extra:shelves}

Here we will describe how to obtain shelves $\Wtild$ and $\Htild$
from items $\What$ and $\Hhat$, respectively.

Since we allow horizontally slicing items in $\What$,
a packing of $\What$ into $m$ bins gives us a
packing of $\What$ into a strip of height $m$,
and a packing of $\What$ into a strip of height $h'$ gives us a
packing of $\What$ into $\smallceil{h'}$ bins.
Hence, if we denote the optimal strip packing of $\What$ by $\optSP(\What)$,
then $\opt(\What) = \smallceil{\optSP(\What)}$.
We will now try to compute a near-optimal strip packing of $\What$.

Define a horizontal configuration $S$ as a tuple of $1/\eps^2+1$ non-negative integers,
where $S_0 \in \{0, 1\}$ and $\sum_{j=1}^{1/\eps^2} S_jw_j \le 1$.
For any horizontal line at height $y$ in a strip packing of $\What$,
the multiset of items intersecting the line corresponds to a configuration.
$S_0$ indicates whether the line intersects items from $\WSmall$,
and $S_j$ is the number of items from $\WhatLarge_j$ that the line intersects.
Let $\Scal$ be the set of all horizontal configurations. Let $N \defeq |\Scal|$.

To obtain an optimal packing, we need to determine the height of each configuration.
This can be done with the following linear program.
\[ \optprog{\min_{x \in \mathbb{R}^N}}{\sum_{S \in \Scal} x_S}{
\\[1.5em] \textrm{where} & \sum_{S \in \Scal} S_jx_S = h(\WhatLarge_j)
    & \forall j \in [1/\eps^2]
\\ \textrm{and} & \sum_{S: S_0=1} \left(1 - \sum_{j=1}^{1/\eps^2}S_jw_j\right)x_S
    = a(\WSmall)
\\ \textrm{and} & x_S \ge 0 & \forall S \in \Scal
} \]
Let $x^*$ be an optimal extreme-point solution to the above LP.
This gives us a packing where the strip is divided into rectangular regions
called \emph{shelves} that are stacked on top of each other.
Each shelf has a configuration $S$ associated with it
and has height $h(S) \defeq x^*_S$ and contains $S_j$ \emph{containers} of width $w_j$.
Containers of width $w_j$ only contain items from $\WhatLarge_j$,
and we call them \emph{type-$j$} containers.
If $S_0 = 1$, $S$ also contains a container of width $1 - \sum_{j=1}^{1/\eps^2} S_jw_j$
that contains small items. We call this container a \emph{type-$0$} container.
Each container is fully filled with items.
Let $w(S)$ denote the width of shelf $S$, i.e., the sum of widths of all containers in $S$.
Note that if $S_0 = 1$, then $w(S) = 1$. Otherwise, $w(S) = \sum_{j=1}^{1/\eps^2} S_jw_j$.

\begin{lemma}
\label{thm:nonneg-entries}
$x^*$ contains at most $1/\eps^2+1$ positive entries.
\end{lemma}
\begin{proof}[Proof sketch]
Follows by applying Rank Lemma\footnote{\rankLemmaNote} to the linear program.
\end{proof}

\begin{lemma}
\label{thm:width-gt-half}
$x_S^* > 0 \implies w(S) > 1/2$.
\end{lemma}
\begin{proof}
Suppose $w(S) \le 1/2$. Then we could have split $S$ into two parts by making a horizontal cut
in the middle and packed the parts side-by-side, reducing the height of the strip by $x^*_S/2$.
But that would contradict the fact that $x^*$ is optimal.
\end{proof}

Treat each shelf $S$ as an item of width $w(S)$ and height $h(S)$.
Allow each such item to be sliced using horizontal cuts.
This gives us a new set $\Wtild$ of items such that $\What$ can be packed inside $\Wtild$.

By applying an analogous approach to $\Hhat$, we get a new set $\Htild$ of items.
Let $\Itild \defeq \Wtild \cup \Htild$.
We call the shelves of $\Wtild$ \emph{wide shelves}
and the shelves of $\Htild$ \emph{tall shelves}.
The containers in wide shelves are called \emph{wide containers}
and the containers in tall shelves are called \emph{tall containers}.

\rthmCreateShelves*
\begin{proof}
By \cref{thm:nonneg-entries}, $|\Wtild| \le 1 + 1/\eps^2$ and $|\Htild| \le 1 + 1/\eps^2$.
By \cref{thm:width-gt-half}, each item in $\Wtild$ has width more than $1/2$
and each item in $\Htild$ has height more than $1/2$.

$a(\Ihat) = a(\Itild)$ because the shelves are tightly packed.

Since $x^*$ is the optimal solution to the linear program for strip packing $\What$,
$\hsum(\Wtild) = \sum_{S \in \Scal} x^*_S = \optSP(\What)$.
Therefore, $\smallceil{\hsum(\Wtild)} = \fopt(\What) \le \fopt(\Ihat)$.
Similarly, $\smallceil{\wsum(\Htild)} = \fopt(\Hhat) \le \fopt(\Ihat)$.
\end{proof}

\subsection{Packing Items Into Containers}
\label{sec:guill-thin-extra:pack-into-containers}

\rthmDiscardAreaUb*
\begin{proof}
For each $j \in [1/\eps^2]$, number the type-$j$ wide containers arbitrarily,
and number the items in $\WhatLarge_j$ arbitrarily.
Now greedily assign items from $\WhatLarge_j$ to the first container $C$ until the total height
of the items exceeds $h(C)$. Then move to the next container and repeat.
As per the constraints of the linear program, all items in $\WhatLarge_j$
will get assigned to some type-$j$ wide container.
Similarly, number the type-0 wide containers arbitrarily
and number the items in $\WSmall$ arbitrarily.
Greedily assign items from $\WSmall$ to the first container $C$ until the total area
of the items exceeds $a(C)$. Then move to the next container and repeat.
As per the constraints of the linear program, all items in $\WSmall$
will get assigned to some type-0 wide container.
Similarly, assign all items from $\Hhat$ to tall containers.

Let $C$ be a type-$j$ wide container and $\Jhat$ be the items assigned to it.
If we discard the last item from $\Jhat$, then the items can be packed into $C$.
The area of the discarded item is at most $w(C)\delta_H$.
Let $C$ be a type-0 wide container and $\Jhat$ be the items assigned to it.
Arrange the items in $\Jhat$ in decreasing order of height and pack the largest
prefix $\Jhat' \subseteq \Jhat$ into $C$ using NFDW (Next-Fit Decreasing Width),
which is an \analogue{} of NFDH with the coordinate axes swapped.

Discard the items $\Jhat - \Jhat'$. By \cref{thm:nfdh-small},
$a(\Jhat - \Jhat') < \eps h(C) + \delta_H w(C) + \eps\delta_H$.
Therefore, for a wide shelf $S$, the total area of discarded items is less than
$\eps h(S) + \delta_H(1 + \eps)$.

After slicing the shelves in $\Itild$ to get $P$,
we get at most $m + 1/\eps^2$ wide shelves
and at most $m + 1/\eps^2$ tall shelves.
Therefore, the total area of discarded items from $W$ is less than
\[ \eps \hsum(\Wtild) + \delta_H(1 + \eps)(m + 1/\eps^2), \]
and the total area of discarded items from $H$ is less than
\[ \eps \wsum(\Htild) + \delta_W(1 + \eps)(m + 1/\eps^2). \qedhere \]
\end{proof}

\section{Lower Bound on APoG}
\label{sec:apog-lb}

In this section, we prove a lower bound of roughly $4/3$ on the APoG for skewed rectangles.

\begin{lemma}
\label{thm:guill-hardex-area}
Let $m$ and $k$ be positive integers and $\eps$ be a positive real number.
Let $J$ be a set of items packed into a bin, where each item has the longer dimension
equal to $(1+\eps)/2$ and the shorter dimension equal to $(1-\eps)/2k$.
If the bin is guillotine-separable, then $a(J) \le 3/4 + \eps/2 - \eps^2/4$.
\end{lemma}
\begin{proof}
For an item packed in the bin, if the height is $(1-\eps)/2k$, call it a wide item,
and if the width is $(1-\eps)/2k$, call it a tall item.
Let $W$ be the set of wide items in $J$.

The packing of items in the bin can be represented as a tree,
called the \emph{guillotine tree} of the bin,
where each node $u$ represents a rectangular region of the bin
and the child nodes $v_1, v_2, \ldots, v_p$ of node $u$
represent the sub-regions obtained by parallel guillotine cuts.
The ordering of the children has a significance here:
if the guillotine cuts were vertical, children are ordered by increasing $x$-coordinate,
and if the cuts were horizontal, children are ordered by increasing $y$-coordinate.
See \cref{fig:guill-tree} for an example.

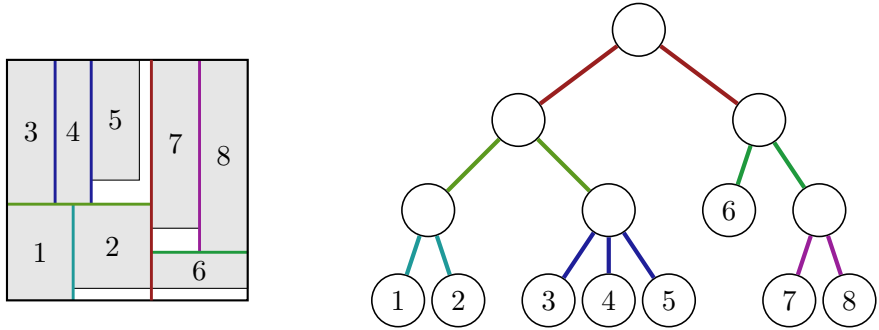
\begin{figure}[htb]
\centering
\definecolor{myred}{Hsb}{0,0.8,0.6}
\definecolor{myorange}{Hsb}{30,0.8,0.6}
\definecolor{myyellow}{Hsb}{60,0.8,0.6}
\definecolor{mylime}{Hsb}{90,0.8,0.6}
\definecolor{mygreen}{Hsb}{135,0.8,0.6}
\definecolor{mycyan}{Hsb}{180,0.8,0.6}
\definecolor{myblue}{Hsb}{240,0.8,0.6}
\definecolor{mypink}{Hsb}{300,0.8,0.6}
\begin{tikzpicture}[
item/.style={fill={black!10},draw},
myarrow/.style={->,>={Stealth},thick},
bin/.style={draw,thick},
vertex/.style={draw,shape=circle,semithick,inner sep=0,minimum size=0.7cm},
myedge/.style={ultra thick},
cutline/.style={draw,line width=1.1pt},
scale=0.8,
]
\begin{scope}[yscale=-1]
\path[item] (0, 0) rectangle +(0.8, 2.4) node[pos=0.5] {3};
\path[item] (0.8, 0) rectangle +(0.6, 2.4) node[pos=0.5] {4};
\path[item] (1.4, 0) rectangle +(0.8, 2) node[pos=0.5] {5};
\path[cutline,draw=myblue] (0.8, 0) -- (0.8, 2.4);
\path[cutline,draw=myblue] (1.4, 0) -- (1.4, 2.4);
\path[item] (0, 2.4) rectangle +(1.1, 1.6) node[pos=0.5] {1};
\path[item] (1.1, 2.4) rectangle +(1.3, 1.4) node[pos=0.5] {2};
\path[cutline,draw=mycyan] (1.1, 2.4) -- (1.1, 4);
\path[cutline,draw=mylime] (0, 2.4) -- (2.4, 2.4);

\path[item] (2.4, 0) rectangle +(0.8, 2.8) node[pos=0.5] {7};
\path[item] (3.2, 0) rectangle +(0.8, 3.2) node[pos=0.5] {8};
\path[cutline,draw=mypink] (3.2, 0) -- (3.2, 3.2);
\path[item] (2.4, 3.2) rectangle +(1.6, 0.6) node[pos=0.5] {6};
\path[cutline,draw=mygreen] (2.4, 3.2) -- (4, 3.2);
\path[bin] (0, 0) rectangle +(4, 4);
\path[cutline,draw=myred] (2.4, 0) -- (2.4, 4);
\end{scope}
\begin{scope}[xshift=6.5cm,yshift=-4cm]
\node[vertex] (v1) at (0,0) {1};
\node[vertex] (v2) at (1,0) {2};
\node[vertex] (v3) at (2.5,0) {3};
\node[vertex] (v4) at (3.5,0) {4};
\node[vertex] (v5) at (4.5,0) {5};
\node[vertex] (v6) at (5.5,1.5) {6};
\node[vertex] (v7) at (6.5,0) {7};
\node[vertex] (v8) at (7.5,0) {8};
\node[vertex] (v12) at (0.5,1.5) {};
\node[vertex] (v345) at (3.5,1.5) {};
\node[vertex] (v78) at (7,1.5) {};
\node[vertex] (v12345) at (2,3) {};
\node[vertex] (v678) at (6,3) {};
\node[vertex] (vroot) at (4,4.5) {};
\draw[myedge,mycyan] (v1) -- (v12);
\draw[myedge,mycyan] (v2) -- (v12);
\draw[myedge,myblue] (v345) -- (v3);
\draw[myedge,myblue] (v345) -- (v4);
\draw[myedge,myblue] (v345) -- (v5);
\draw[myedge,mypink] (v78) -- (v7);
\draw[myedge,mypink] (v78) -- (v8);
\draw[myedge,mylime] (v12345) -- (v12);
\draw[myedge,mylime] (v12345) -- (v345);
\draw[myedge,mygreen] (v678) -- (v6);
\draw[myedge,mygreen] (v678) -- (v78);
\draw[myedge,myred] (vroot) -- (v12345);
\draw[myedge,myred] (vroot) -- (v678);
\end{scope}
\end{tikzpicture}

\caption{A guillotinable packing of items into a bin and the corresponding
guillotine tree.}
\label{fig:guill-tree}
\end{figure}

We will now see how to rearrange the items in the bin so that the packing remains
guillotine-separable but becomes more structured.
We will exploit this structure to show that the packing has a large unpacked area.
See \cref{fig:shift-thin} for an example.

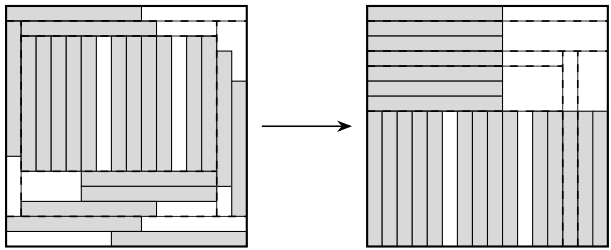
\begin{figure}[htb]
\centering
\ifcsname pGameL\endcsname\else\newlength{\pGameL}\fi
\setlength{\pGameL}{0.2cm}
\tikzset{bin/.style={draw,thick}}
\tikzset{item/.style={draw,fill={black!15}}}
\tikzset{myarrow/.style={->,>={Stealth}}}
\tikzset{gsep/.style={draw,dashed,semithick,fill={black!85}}}
\begin{tikzpicture}
\begin{scope}
\path[item] (0\pGameL, 15\pGameL) rectangle +(9\pGameL, 1\pGameL);
\path[item] (1\pGameL, 14\pGameL) rectangle +(9\pGameL, 1\pGameL);
\path[item] (5\pGameL, 4\pGameL) rectangle +(9\pGameL, 1\pGameL);
\path[item] (5\pGameL, 3\pGameL) rectangle +(9\pGameL, 1\pGameL);
\path[item] (1\pGameL, 2\pGameL) rectangle +(9\pGameL, 1\pGameL);
\path[item] (0\pGameL, 1\pGameL) rectangle +(9\pGameL, 1\pGameL);
\path[item] (7\pGameL, 0\pGameL) rectangle +(9\pGameL, 1\pGameL);
\path[item] (0\pGameL, 6\pGameL) rectangle +(1\pGameL, 9\pGameL);
\path[item] (1\pGameL, 5\pGameL) rectangle +(1\pGameL, 9\pGameL);
\path[item] (2\pGameL, 5\pGameL) rectangle +(1\pGameL, 9\pGameL);
\path[item] (3\pGameL, 5\pGameL) rectangle +(1\pGameL, 9\pGameL);
\path[item] (4\pGameL, 5\pGameL) rectangle +(1\pGameL, 9\pGameL);
\path[item] (5\pGameL, 5\pGameL) rectangle +(1\pGameL, 9\pGameL);
\path[item] (7\pGameL, 5\pGameL) rectangle +(1\pGameL, 9\pGameL);
\path[item] (8\pGameL, 5\pGameL) rectangle +(1\pGameL, 9\pGameL);
\path[item] (9\pGameL, 5\pGameL) rectangle +(1\pGameL, 9\pGameL);
\path[item] (10\pGameL, 5\pGameL) rectangle +(1\pGameL, 9\pGameL);
\path[item] (12\pGameL, 5\pGameL) rectangle +(1\pGameL, 9\pGameL);
\path[item] (13\pGameL, 5\pGameL) rectangle +(1\pGameL, 9\pGameL);
\path[item] (14\pGameL, 4\pGameL) rectangle +(1\pGameL, 9\pGameL);
\path[item] (15\pGameL, 2\pGameL) rectangle +(1\pGameL, 9\pGameL);
\path[bin] (0\pGameL, 0\pGameL) rectangle (16\pGameL, 16\pGameL);
\path[gsep] (0\pGameL, 2\pGameL) -- +(16\pGameL, 0\pGameL);
\path[gsep] (0\pGameL, 15\pGameL) -- +(16\pGameL, 0\pGameL);
\path[gsep] (1\pGameL, 2\pGameL) -- (1\pGameL, 15\pGameL);
\path[gsep] (14\pGameL, 2\pGameL) -- (14\pGameL, 15\pGameL);
\path[gsep] (1\pGameL, 14\pGameL) -- (14\pGameL, 14\pGameL);
\path[gsep] (1\pGameL, 5\pGameL) -- (14\pGameL, 5\pGameL);
\end{scope}
\draw[myarrow,semithick] (17\pGameL, 8\pGameL) -- +(6\pGameL, 0);
\begin{scope}[xshift=24\pGameL]
\path[item] (0\pGameL, 15\pGameL) rectangle +(9\pGameL, 1\pGameL);
\path[item] (0\pGameL, 14\pGameL) rectangle +(9\pGameL, 1\pGameL);
\path[item] (0\pGameL, 13\pGameL) rectangle +(9\pGameL, 1\pGameL);
\path[item] (0\pGameL, 12\pGameL) rectangle +(9\pGameL, 1\pGameL);
\path[item] (0\pGameL, 11\pGameL) rectangle +(9\pGameL, 1\pGameL);
\path[item] (0\pGameL, 10\pGameL) rectangle +(9\pGameL, 1\pGameL);
\path[item] (0\pGameL, 9\pGameL) rectangle +(9\pGameL, 1\pGameL);
\path[item] (0\pGameL, 0\pGameL) rectangle +(1\pGameL, 9\pGameL);
\path[item] (1\pGameL, 0\pGameL) rectangle +(1\pGameL, 9\pGameL);
\path[item] (2\pGameL, 0\pGameL) rectangle +(1\pGameL, 9\pGameL);
\path[item] (3\pGameL, 0\pGameL) rectangle +(1\pGameL, 9\pGameL);
\path[item] (4\pGameL, 0\pGameL) rectangle +(1\pGameL, 9\pGameL);
\path[item] (6\pGameL, 0\pGameL) rectangle +(1\pGameL, 9\pGameL);
\path[item] (7\pGameL, 0\pGameL) rectangle +(1\pGameL, 9\pGameL);
\path[item] (8\pGameL, 0\pGameL) rectangle +(1\pGameL, 9\pGameL);
\path[item] (9\pGameL, 0\pGameL) rectangle +(1\pGameL, 9\pGameL);
\path[item] (11\pGameL, 0\pGameL) rectangle +(1\pGameL, 9\pGameL);
\path[item] (12\pGameL, 0\pGameL) rectangle +(1\pGameL, 9\pGameL);
\path[item] (13\pGameL, 0\pGameL) rectangle +(1\pGameL, 9\pGameL);
\path[item] (14\pGameL, 0\pGameL) rectangle +(1\pGameL, 9\pGameL);
\path[item] (15\pGameL, 0\pGameL) rectangle +(1\pGameL, 9\pGameL);
\path[bin] (0\pGameL, 0\pGameL) rectangle (16\pGameL, 16\pGameL);
\path[gsep] (0\pGameL, 15\pGameL) -- (16\pGameL, 15\pGameL);
\path[gsep] (0\pGameL, 13\pGameL) -- (16\pGameL, 13\pGameL);
\path[gsep] (13\pGameL, 0\pGameL) -- (13\pGameL, 13\pGameL);
\path[gsep] (14\pGameL, 0\pGameL) -- (14\pGameL, 13\pGameL);
\path[gsep] (0\pGameL, 12\pGameL) -- (13\pGameL, 12\pGameL);
\path[gsep] (0\pGameL, 9\pGameL) -- (13\pGameL, 9\pGameL);
\end{scope}
\end{tikzpicture}

\caption{Structuring a guillotine-separable packing.}
\label{fig:shift-thin}
\end{figure}

In the guillotine tree, suppose there is a node $u$
that has children $v_1, v_2, \ldots, v_p$.
\WLoG, assume that the children are obtained by making vertical cuts.
At most one of these children can contain items from $W$.
We can assume \wLoG{} that the other children contain only one item,
because otherwise we can separate them by vertical cuts.
We can reorder the children (which is equivalent to repacking the guillotine partitions)
so that the child containing items from $W$ (if any) is the first child.
Therefore, we can assume \wLoG{} that at any level in the guillotine tree,
only the first node has children.

Based on the argument above, we can see that the first node in each level
touches the bottom-left corner of the bin. All the other nodes either
contain a single wide item and touch the left edge of the bin but not the bottom edge,
or they contain a single tall item and touch the bottom edge of the bin but not the left edge.
In each node containing a wide item, shift the item leftwards,
and in each node containing a tall item, shift the item downwards.
Then each wide item touches the left edge of the bin
and each tall item touches the bottom edge of the bin.

Therefore, the square region of side length $(1-\eps)/2$
at the top-right corner of the bin is empty.
Hence, the area occupied in each bin is at most $3/4 + \eps/2 - \eps^2/4$.
\end{proof}

\begin{theorem}
\label{thm:apog-lb}
Let $m$ and $k$ be positive integers and $\eps \in (0, 1)$.
Let $I$ be a set of $4mk$ rectangular items,
where $2mk$ items have width $(1+\eps)/2$ and height $(1-\eps)/2k$,
and $2mk$ items have height $(1+\eps)/2$ and width $(1-\eps)/2k$.
Let $\opt(I)$ be the number of bins in the
optimal packing of $I$ and $\opt_g(I)$ be the number of bins in the
optimal guillotinable packing of $I$. Then
\[ \frac{\opt_g(I)}{\opt(I)} \ge \frac{4}{3} (1-\eps). \]
This holds true even if items in $I$ are allowed to be rotated.
\end{theorem}
\begin{proof}
For an item $i \in I$, if $h(i) = (1-\eps)/2k$, call it a wide item,
and if $w(i) = (1-\eps)/2k$, call it a tall item.
Let $W$ be the set of wide items and $H$ be the set of tall items.
We will show that $\opt(I)$ and $\opt_g(I)$ have a big difference,
which will give us a lower-bound on $\APoG$.

Partition $W$ into groups of $k$ elements.
In each group, stack items one-over-the-other.
This gives us $2m$ containers of width $(1+\eps)/2$ and height $(1-\eps)/2$.
Similarly, get $2m$ containers of height $(1+\eps)/2$ and height $(1-\eps)/2$
by stacking items from $H$ side-by-side.
We can pack 4 containers in one bin, so $I$ can be packed into $m$ bins.
See \cref{fig:thin-gadget} for an example.
Therefore, $\opt(I) \le m$.

\begin{figure}[htb]
\centering
\ifcsname pGameL\endcsname\else\newlength{\pGameL}\fi
\setlength{\pGameL}{0.2cm}
\tikzset{bin/.style={draw,thick}}
\tikzset{item/.style={draw,fill={black!15}}}
\begin{tikzpicture}
\path[item] (7\pGameL, 15\pGameL) rectangle +(9\pGameL, 1\pGameL);
\path[item] (0\pGameL, 0\pGameL) rectangle +(9\pGameL, 1\pGameL);
\path[item] (7\pGameL, 14\pGameL) rectangle +(9\pGameL, 1\pGameL);
\path[item] (7\pGameL, 13\pGameL) rectangle +(9\pGameL, 1\pGameL);
\path[item] (7\pGameL, 12\pGameL) rectangle +(9\pGameL, 1\pGameL);
\path[item] (7\pGameL, 11\pGameL) rectangle +(9\pGameL, 1\pGameL);
\path[item] (7\pGameL, 10\pGameL) rectangle +(9\pGameL, 1\pGameL);
\path[item] (7\pGameL, 9\pGameL) rectangle +(9\pGameL, 1\pGameL);
\path[item] (0\pGameL, 1\pGameL) rectangle +(9\pGameL, 1\pGameL);
\path[item] (0\pGameL, 2\pGameL) rectangle +(9\pGameL, 1\pGameL);
\path[item] (0\pGameL, 3\pGameL) rectangle +(9\pGameL, 1\pGameL);
\path[item] (0\pGameL, 4\pGameL) rectangle +(9\pGameL, 1\pGameL);
\path[item] (0\pGameL, 5\pGameL) rectangle +(9\pGameL, 1\pGameL);
\path[item] (0\pGameL, 6\pGameL) rectangle +(9\pGameL, 1\pGameL);
\path[item] (0\pGameL, 7\pGameL) rectangle +(1\pGameL, 9\pGameL);
\path[item] (15\pGameL, 0\pGameL) rectangle +(1\pGameL, 9\pGameL);
\path[item] (1\pGameL, 7\pGameL) rectangle +(1\pGameL, 9\pGameL);
\path[item] (2\pGameL, 7\pGameL) rectangle +(1\pGameL, 9\pGameL);
\path[item] (3\pGameL, 7\pGameL) rectangle +(1\pGameL, 9\pGameL);
\path[item] (4\pGameL, 7\pGameL) rectangle +(1\pGameL, 9\pGameL);
\path[item] (5\pGameL, 7\pGameL) rectangle +(1\pGameL, 9\pGameL);
\path[item] (6\pGameL, 7\pGameL) rectangle +(1\pGameL, 9\pGameL);
\path[item] (14\pGameL, 0\pGameL) rectangle +(1\pGameL, 9\pGameL);
\path[item] (13\pGameL, 0\pGameL) rectangle +(1\pGameL, 9\pGameL);
\path[item] (12\pGameL, 0\pGameL) rectangle +(1\pGameL, 9\pGameL);
\path[item] (11\pGameL, 0\pGameL) rectangle +(1\pGameL, 9\pGameL);
\path[item] (10\pGameL, 0\pGameL) rectangle +(1\pGameL, 9\pGameL);
\path[item] (9\pGameL, 0\pGameL) rectangle +(1\pGameL, 9\pGameL);
\path[bin] (0\pGameL, 0\pGameL) rectangle (16\pGameL, 16\pGameL);
\end{tikzpicture}

\caption{Packing $4k$ items in one bin. Here $k = 7$.}
\label{fig:thin-gadget}
\end{figure}
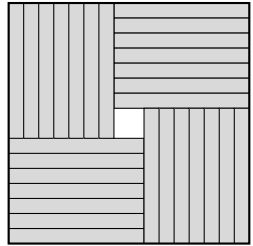

We will now show a lower-bound on $\opt_g(I)$.
In any guillotine-separable packing of $I$,
the area occupied by each bin is at most $3/4 + \eps/2 - \eps^2/4$
(by \cref{thm:guill-hardex-area}).
Note that $a(I) = m(1 - \eps^2)$. Therefore,
\begin{align*}
& \opt_g(I) \ge \frac{m(1-\eps^2)}{3/4 + \eps/2 - \eps^2/4}
\\ &\implies \frac{\opt_g(I)}{\opt(I)}
    \ge \frac{4}{3} \times \frac{1-\eps^2}{1 + 2\eps/3 - \eps^2/3}
    = \frac{4}{3} \times \frac{1-\eps}{1 - \eps/3}
    \ge \frac{4}{3}(1-\eps).
\qedhere
\end{align*}
\end{proof}

\section{APoG for the Rotational Case}
\label{sec:guill-rot}

\begin{theorem}
\label{thm:guill-rot}
Let $\APoG\nonrot$ and $\APoG\rot$ be the APoG for the non-rotational
and rotational versions, respectively, restricted to the $\delta$-\thin{} case.
Then $\APoG\rot \le \APoG\nonrot$.
\end{theorem}
\begin{proof}
For a set $I$ of $\delta$-\thin{} rectangular items,
let $\opt\nonrot(I)$ and $\opt\rot(I)$ be the minimum number of bins
needed to pack $I$ in the non-rotational and rotational versions, respectively.
Let $\opt\nonrot_g(I)$ and $\opt\rot_g(I)$ be the minimum number of guillotinable bins
needed to pack $I$ in the non-rotational and rotational versions, respectively.
Assume \wLoG{} that the bin has width and height at least 1.

Let $I$ be any set of $\delta$-\thin{} items.
Let $K$ be the corresponding rotated items in the optimal rotational packing of $I$,
i.e., $\opt\rot(I) = \opt\nonrot(K)$. Then
\begin{align*}
\opt\rot_g(I) &\le \opt\nonrot_g(K)
\\ &\le \APoG\nonrot\opt\nonrot(K) + c
\tag{$c$ is a constant}
\\ &= \APoG\nonrot\opt\rot(I) + c.
\end{align*}
Hence, we get $\APoG\rot \le \APoG\nonrot$.
\end{proof}

\section{Details of \texorpdfstring{$\thinCPack$}{skewedCPack}}
\label{sec:thin-bp-extra}

\subsection{Removing Medium Items}
\label{sec:thin-bp-extra:remmed}

Let $T \defeq \ceil{2/\eps}$. Let $\mu_0 = \eps$.
For $t \in [T]$, define $\mu_t \defeq f(\mu_{t-1})$ and define
\[ J_t \defeq \{i \in I: w(i) \in (\mu_t, \mu_{t-1}]
    \textrm{ or } h(i) \in (\mu_t, \mu_{t-1}]\}. \]
Define $r \defeq \argmin_{t=1}^T a(J_t)$,
$\Imed \defeq J_r$ and $\epsLarge \defeq \mu_{r-1}$.
Each item belongs to at most 2 sets $J_t$, so
\[ a(\Imed) = \min_{t=1}^T a(J_t)
\le \frac{1}{T} \sum_{t=1}^T a(J_t)
\le \frac{2}{\ceil{2/\eps}} a(I)
\le \eps a(I). \]

\subsection{Creating Compartments}
\label{sec:thin-bp-extra:compartmentalize}

\begin{lemma}
\label{thm:empty-to-rects}
Let there be a set $I$ of rectangles packed inside a bin.
Then there is a polynomial-time algorithm that can decompose the empty space in the bin
into at most $3|I|+1$ rectangles by making horizontal cuts only.
\end{lemma}
\begin{proof}
Extend the top and bottom edge of each rectangle leftwards and rightwards
till they hit another rectangle or an edge of the bin.
This decomposes the empty region into rectangles $R$.
See \cref{fig:thin-bp:empty-space-to-rects}.

For each rectangle $i \in I$, the top edge of $i$ is the bottom edge of a rectangle in $R$,
the bottom edge of $i$ is the bottom edge of two rectangles in $R$.
Apart from possibly the rectangle in $R$ whose bottom edge is at the bottom of the bin,
the bottom edge of every rectangle in $R$ is either the bottom or top edge of a rectangle in $I$.
Therefore, $|R| \le 3|I| + 1$.
\end{proof}

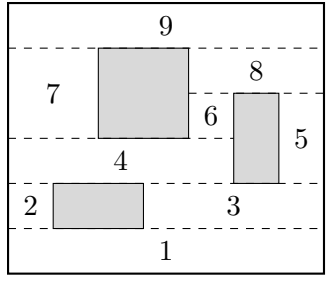
\begin{figure}[htb]
\centering
\ifcsname myu\endcsname\else\newlength{\myu}\fi
\setlength{\myu}{0.6cm}
\begin{tikzpicture}[
item/.style={draw,fill={black!15}},
cutline/.style={draw,dashed},
]
\path[item] (1\myu, 1\myu) rectangle (3\myu, 2\myu);
\path[item] (5\myu, 2\myu) rectangle (6\myu, 4\myu);
\path[item] (2\myu, 3\myu) rectangle (4\myu, 5\myu);
\draw[thick] (0\myu, 0\myu) rectangle (7\myu, 6\myu);

\draw[cutline] (0\myu, 1\myu) -- (7\myu, 1\myu);
\draw[cutline] (0\myu, 2\myu) -- (7\myu, 2\myu);
\draw[cutline] (0\myu, 3\myu) -- (5\myu, 3\myu);
\draw[cutline] (4\myu, 4\myu) -- (7\myu, 4\myu);
\draw[cutline] (0\myu, 5\myu) -- (7\myu, 5\myu);

\path (0\myu, 0\myu) -- node {1} (7\myu, 1\myu);
\path (0\myu, 1\myu) -- node {2} (1\myu, 2\myu);
\path (3\myu, 1\myu) -- node {3} (7\myu, 2\myu);
\path (0\myu, 2\myu) -- node {4} (5\myu, 3\myu);
\path (6\myu, 2\myu) -- node {5} (7\myu, 4\myu);
\path (4\myu, 3\myu) -- node {6} (5\myu, 4\myu);
\path (0\myu, 3\myu) -- node {7} (2\myu, 5\myu);
\path (4\myu, 4\myu) -- node {8} (7\myu, 5\myu);
\path (0\myu, 5\myu) -- node {9} (7\myu, 6\myu);
\end{tikzpicture}

\caption{Using horizontal cuts to partition the empty space
around the 3 items into 9 rectangular regions.}
\label{fig:thin-bp:empty-space-to-rects}
\end{figure}

\rthmCompartmentalize*
\begin{proof}
Draw vertical lines in the bin at the $x$-coordinates in $\Tcal - \{0\}$.
This splits the bin into $|\Tcal|$ columns (see \cref{fig:compartmentalize:begin}).
Each column has 0 or more wide items crossing it.
These wide items divide the column into cells.
A cell is called tall iff it contains a tall item (see \cref{fig:compartmentalize:tall-cells}).
There can be at most $1/\epsLarge-1$ tall cells in a column,
so there can be at most $(1/\epsLarge - 1)|\Tcal|$ tall cells in the bin.

\begin{figure}[!htb]
\begin{subfigure}{0.45\textwidth}
\centering
\ifcsname pGameL\endcsname\else\newlength{\pGameL}\fi
\setlength{\pGameL}{0.15cm}
\tikzset{bin/.style={draw,thick}}
\tikzset{binGrid/.style={draw,step=1\pGameL,{black!20}}}
\tikzset{witem/.style={draw,fill=myGreen}}
\tikzset{hitem/.style={draw,fill=myBlue}}
\tikzset{cutline/.style={draw,dashed,thick}}
\begin{tikzpicture}
\path[hitem] (0\pGameL, 22\pGameL) rectangle +(1\pGameL, 8\pGameL);
\path[hitem] (2\pGameL, 20\pGameL) rectangle +(2\pGameL, 8\pGameL);
\path[hitem] (7\pGameL, 21\pGameL) rectangle +(1\pGameL, 9\pGameL);
\path[hitem] (8\pGameL, 20\pGameL) rectangle +(1\pGameL, 9\pGameL);
\path[hitem] (10\pGameL, 22\pGameL) rectangle +(2\pGameL, 8\pGameL);
\path[hitem] (12\pGameL, 17\pGameL) rectangle +(2\pGameL, 8\pGameL);
\path[hitem] (14\pGameL, 20\pGameL) rectangle +(2\pGameL, 10\pGameL);
\path[hitem] (17\pGameL, 19\pGameL) rectangle +(2\pGameL, 11\pGameL);
\path[hitem] (0\pGameL, 1\pGameL) rectangle +(1\pGameL, 8\pGameL);
\path[hitem] (2\pGameL, 0\pGameL) rectangle +(1\pGameL, 12\pGameL);
\path[hitem] (4\pGameL, 1\pGameL) rectangle +(2\pGameL, 8\pGameL);
\path[hitem] (7\pGameL, 0\pGameL) rectangle +(2\pGameL, 9\pGameL);
\path[hitem] (11\pGameL, 4\pGameL) rectangle +(1\pGameL, 6\pGameL);
\path[hitem] (26\pGameL, 10\pGameL) rectangle +(2\pGameL, 9\pGameL);
\path[hitem] (16\pGameL, 12\pGameL) rectangle +(1\pGameL, 12\pGameL);
\path[hitem] (17\pGameL, 7\pGameL) rectangle +(1\pGameL, 11\pGameL);
\path[hitem] (18\pGameL, 7\pGameL) rectangle +(2\pGameL, 7\pGameL);
\path[hitem] (21\pGameL, 9\pGameL) rectangle +(2\pGameL, 13\pGameL);
\path[hitem] (25\pGameL, 8\pGameL) rectangle +(1\pGameL, 16\pGameL);
\path[hitem] (29\pGameL, 2\pGameL) rectangle +(1\pGameL, 8\pGameL);
\path[hitem] (27\pGameL, 2\pGameL) rectangle +(2\pGameL, 8\pGameL);
\path[hitem] (28\pGameL, 10\pGameL) rectangle +(1\pGameL, 9\pGameL);

\path[witem] (0\pGameL, 12\pGameL) rectangle +(10\pGameL, 2\pGameL);
\path[witem] (0\pGameL, 17\pGameL) rectangle +(10\pGameL, 2\pGameL);
\path[witem] (5\pGameL, 16\pGameL) rectangle +(10\pGameL, 1\pGameL);
\path[witem] (0\pGameL, 15\pGameL) rectangle +(10\pGameL, 1\pGameL);
\path[witem] (0\pGameL, 14\pGameL) rectangle +(15\pGameL, 1\pGameL);
\path[witem] (15\pGameL, 5\pGameL) rectangle +(10\pGameL, 2\pGameL);
\path[witem] (5\pGameL, 10\pGameL) rectangle +(10\pGameL, 2\pGameL);
\path[witem] (15\pGameL, 0\pGameL) rectangle +(10\pGameL, 1\pGameL);
\path[witem] (15\pGameL, 1\pGameL) rectangle +(15\pGameL, 1\pGameL);
\path[witem] (10\pGameL, 2\pGameL) rectangle +(10\pGameL, 2\pGameL);
\path[witem] (20\pGameL, 25\pGameL) rectangle +(10\pGameL, 2\pGameL);

\path[bin] (0\pGameL, 0\pGameL) rectangle (30\pGameL, 30\pGameL);
\path[cutline] foreach \x in {5,10,15,20,25}{
    (\x\pGameL, -1\pGameL) -- (\x\pGameL, 32\pGameL)
};
\foreach \index/\x in {1/2.5,2/7.5,3/12.5,4/17.5,5/22.5,6/27.5}{
    \node[anchor=south] at (\x\pGameL, 30\pGameL) {\index};
}
\end{tikzpicture}

\caption{A packing of items in a bin. Wide items are green and tall items are blue.
Draw vertical lines at $x$-coordinates from $\Tcal - \{0\}$.
They divide the bin into columns. In this figure, we have 6 columns.}
\label{fig:compartmentalize:begin}
\end{subfigure}
\hfill
\begin{subfigure}{0.45\textwidth}
\centering
\ifcsname pGameL\endcsname\else\newlength{\pGameL}\fi
\setlength{\pGameL}{0.15cm}
\tikzset{bin/.style={draw,thick}}
\tikzset{binGrid/.style={draw,step=1\pGameL,{black!20}}}
\tikzset{witem/.style={draw,fill=myGreen}}
\tikzset{hitem/.style={draw={black!40},fill=shadedLightBlue}}
\tikzset{tcell/.style={fill={black!10}}}
\tikzset{tcell2/.style={draw,semithick}}
\tikzset{cutline/.style={draw,dashed,thick}}
\begin{tikzpicture}
\path[tcell] (0\pGameL, 19\pGameL) rectangle +(5\pGameL, 11\pGameL);
\path[tcell] (5\pGameL, 19\pGameL) rectangle +(5\pGameL, 11\pGameL);
\path[tcell] (10\pGameL, 17\pGameL) rectangle +(5\pGameL, 13\pGameL);
\path[tcell] (0\pGameL, 0\pGameL) rectangle +(5\pGameL, 12\pGameL);
\path[tcell] (5\pGameL, 0\pGameL) rectangle +(5\pGameL, 10\pGameL);
\path[tcell] (10\pGameL, 4\pGameL) rectangle +(5\pGameL, 6\pGameL);
\path[tcell] (15\pGameL, 7\pGameL) rectangle +(5\pGameL, 23\pGameL);
\path[tcell] (20\pGameL, 7\pGameL) rectangle +(5\pGameL, 18\pGameL);
\path[tcell] (25\pGameL, 2\pGameL) rectangle +(5\pGameL, 23\pGameL);
\path[hitem] (0\pGameL, 22\pGameL) rectangle +(1\pGameL, 8\pGameL);
\path[hitem] (2\pGameL, 20\pGameL) rectangle +(2\pGameL, 8\pGameL);
\path[hitem] (7\pGameL, 21\pGameL) rectangle +(1\pGameL, 9\pGameL);
\path[hitem] (8\pGameL, 20\pGameL) rectangle +(1\pGameL, 9\pGameL);
\path[hitem] (10\pGameL, 22\pGameL) rectangle +(2\pGameL, 8\pGameL);
\path[hitem] (12\pGameL, 17\pGameL) rectangle +(2\pGameL, 8\pGameL);
\path[hitem] (14\pGameL, 20\pGameL) rectangle +(2\pGameL, 10\pGameL);
\path[hitem] (17\pGameL, 19\pGameL) rectangle +(2\pGameL, 11\pGameL);
\path[hitem] (0\pGameL, 1\pGameL) rectangle +(1\pGameL, 8\pGameL);
\path[hitem] (2\pGameL, 0\pGameL) rectangle +(1\pGameL, 12\pGameL);
\path[hitem] (4\pGameL, 1\pGameL) rectangle +(2\pGameL, 8\pGameL);
\path[hitem] (7\pGameL, 0\pGameL) rectangle +(2\pGameL, 9\pGameL);
\path[hitem] (11\pGameL, 4\pGameL) rectangle +(1\pGameL, 6\pGameL);
\path[hitem] (26\pGameL, 10\pGameL) rectangle +(2\pGameL, 9\pGameL);
\path[hitem] (16\pGameL, 12\pGameL) rectangle +(1\pGameL, 12\pGameL);
\path[hitem] (17\pGameL, 7\pGameL) rectangle +(1\pGameL, 11\pGameL);
\path[hitem] (18\pGameL, 7\pGameL) rectangle +(2\pGameL, 7\pGameL);
\path[hitem] (21\pGameL, 9\pGameL) rectangle +(2\pGameL, 13\pGameL);
\path[hitem] (25\pGameL, 8\pGameL) rectangle +(1\pGameL, 16\pGameL);
\path[hitem] (29\pGameL, 2\pGameL) rectangle +(1\pGameL, 8\pGameL);
\path[hitem] (27\pGameL, 2\pGameL) rectangle +(2\pGameL, 8\pGameL);
\path[hitem] (28\pGameL, 10\pGameL) rectangle +(1\pGameL, 9\pGameL);

\path[witem] (0\pGameL, 12\pGameL) rectangle +(10\pGameL, 2\pGameL);
\path[witem] (0\pGameL, 17\pGameL) rectangle +(10\pGameL, 2\pGameL);
\path[witem] (5\pGameL, 16\pGameL) rectangle +(10\pGameL, 1\pGameL);
\path[witem] (0\pGameL, 15\pGameL) rectangle +(10\pGameL, 1\pGameL);
\path[witem] (0\pGameL, 14\pGameL) rectangle +(15\pGameL, 1\pGameL);
\path[witem] (15\pGameL, 5\pGameL) rectangle +(10\pGameL, 2\pGameL);
\path[witem] (5\pGameL, 10\pGameL) rectangle +(10\pGameL, 2\pGameL);
\path[witem] (15\pGameL, 0\pGameL) rectangle +(10\pGameL, 1\pGameL);
\path[witem] (15\pGameL, 1\pGameL) rectangle +(15\pGameL, 1\pGameL);
\path[witem] (10\pGameL, 2\pGameL) rectangle +(10\pGameL, 2\pGameL);
\path[witem] (20\pGameL, 25\pGameL) rectangle +(10\pGameL, 2\pGameL);

\path[tcell2] (0\pGameL, 19\pGameL) rectangle +(5\pGameL, 11\pGameL) node[pos=0.5] {1};
\path[tcell2] (5\pGameL, 19\pGameL) rectangle +(5\pGameL, 11\pGameL) node[pos=0.5] {2};
\path[tcell2] (10\pGameL, 17\pGameL) rectangle +(5\pGameL, 13\pGameL) node[pos=0.5] {3};
\path[tcell2] (0\pGameL, 0\pGameL) rectangle +(5\pGameL, 12\pGameL) node[pos=0.5] {4};
\path[tcell2] (5\pGameL, 0\pGameL) rectangle +(5\pGameL, 10\pGameL) node[pos=0.5] {5};
\path[tcell2] (10\pGameL, 4\pGameL) rectangle +(5\pGameL, 6\pGameL) node[pos=0.5] {6};
\path[tcell2] (15\pGameL, 7\pGameL) rectangle +(5\pGameL, 23\pGameL) node[pos=0.5] {7};
\path[tcell2] (20\pGameL, 7\pGameL) rectangle +(5\pGameL, 18\pGameL) node[pos=0.5] {8};
\path[tcell2] (25\pGameL, 2\pGameL) rectangle +(5\pGameL, 23\pGameL) node[pos=0.5] {9};
\path[bin] (0\pGameL, 0\pGameL) rectangle (30\pGameL, 30\pGameL);
\end{tikzpicture}

\caption{Wide items divide each column into \emph{cells}.
Each cell containing a tall item is called a \emph{tall cell}.
There are 9 tall cells in this figure, which are shaded gray.}
\label{fig:compartmentalize:tall-cells}
\end{subfigure}
\caption{Creating tall cells in a bin}
\label{fig:compartmentalize-1}
\end{figure}
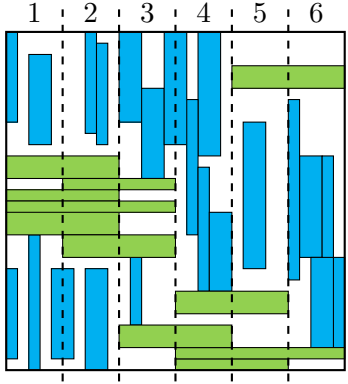
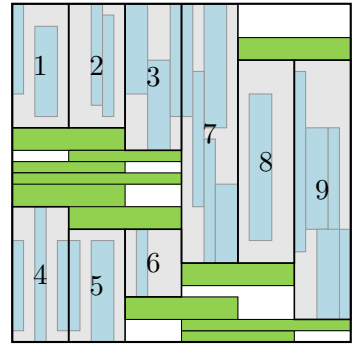

By \cref{thm:empty-to-rects}, we can use horizontal cuts to partition the space outside
tall cells into at most $3(1/\epsLarge-1)|\Tcal| + 1$ rectangular regions
(this may slice some wide items). See \cref{fig:compartmentalize:boxes}.
If a region contains a wide item, call it a box.

\begin{figure}[!htb]
\begin{subfigure}[t]{0.35\textwidth}
\centering
\ifcsname pGameL\endcsname\else\newlength{\pGameL}\fi
\setlength{\pGameL}{0.15cm}
\tikzset{bin/.style={draw,thick}}
\tikzset{binGrid/.style={draw,step=1\pGameL,{black!20}}}
\tikzset{witem/.style={draw={black!40},fill=shadedLightGreen}}
\tikzset{hitem/.style={draw={black!30},fill=myLightBlue}}
\tikzset{tcell/.style={draw}}
\tikzset{box/.style={fill={black!10}}}
\tikzset{box2/.style={draw,thick}}
\begin{tikzpicture}

\path[box] (0\pGameL, 17\pGameL) rectangle +(10\pGameL, 2\pGameL);
\path[box] (0\pGameL, 12\pGameL) rectangle +(15\pGameL, 5\pGameL);
\path[box] (5\pGameL, 10\pGameL) rectangle +(10\pGameL, 2\pGameL);
\path[box] (20\pGameL, 25\pGameL) rectangle +(10\pGameL, 5\pGameL);
\path[box] (15\pGameL, 4\pGameL) rectangle +(10\pGameL, 3\pGameL);
\path[box] (10\pGameL, 2\pGameL) rectangle +(15\pGameL, 2\pGameL);
\path[box] (10\pGameL, 0\pGameL) rectangle +(20\pGameL, 2\pGameL);

\path[witem] (0\pGameL, 12\pGameL) rectangle +(10\pGameL, 2\pGameL);
\path[witem] (0\pGameL, 17\pGameL) rectangle +(10\pGameL, 2\pGameL);
\path[witem] (5\pGameL, 16\pGameL) rectangle +(10\pGameL, 1\pGameL);
\path[witem] (0\pGameL, 15\pGameL) rectangle +(10\pGameL, 1\pGameL);
\path[witem] (0\pGameL, 14\pGameL) rectangle +(15\pGameL, 1\pGameL);
\path[witem] (15\pGameL, 5\pGameL) rectangle +(10\pGameL, 2\pGameL);
\path[witem] (5\pGameL, 10\pGameL) rectangle +(10\pGameL, 2\pGameL);
\path[witem] (15\pGameL, 0\pGameL) rectangle +(10\pGameL, 1\pGameL);
\path[witem] (15\pGameL, 1\pGameL) rectangle +(15\pGameL, 1\pGameL);
\path[witem] (10\pGameL, 2\pGameL) rectangle +(10\pGameL, 2\pGameL);
\path[witem] (20\pGameL, 25\pGameL) rectangle +(10\pGameL, 2\pGameL);

\path[hitem] (0\pGameL, 22\pGameL) rectangle +(1\pGameL, 8\pGameL);
\path[hitem] (2\pGameL, 20\pGameL) rectangle +(2\pGameL, 8\pGameL);
\path[hitem] (7\pGameL, 21\pGameL) rectangle +(1\pGameL, 9\pGameL);
\path[hitem] (8\pGameL, 20\pGameL) rectangle +(1\pGameL, 9\pGameL);
\path[hitem] (10\pGameL, 22\pGameL) rectangle +(2\pGameL, 8\pGameL);
\path[hitem] (12\pGameL, 17\pGameL) rectangle +(2\pGameL, 8\pGameL);
\path[hitem] (14\pGameL, 20\pGameL) rectangle +(2\pGameL, 10\pGameL);
\path[hitem] (17\pGameL, 19\pGameL) rectangle +(2\pGameL, 11\pGameL);
\path[hitem] (0\pGameL, 1\pGameL) rectangle +(1\pGameL, 8\pGameL);
\path[hitem] (2\pGameL, 0\pGameL) rectangle +(1\pGameL, 12\pGameL);
\path[hitem] (4\pGameL, 1\pGameL) rectangle +(2\pGameL, 8\pGameL);
\path[hitem] (7\pGameL, 0\pGameL) rectangle +(2\pGameL, 9\pGameL);
\path[hitem] (11\pGameL, 4\pGameL) rectangle +(1\pGameL, 6\pGameL);
\path[hitem] (26\pGameL, 10\pGameL) rectangle +(2\pGameL, 9\pGameL);
\path[hitem] (16\pGameL, 12\pGameL) rectangle +(1\pGameL, 12\pGameL);
\path[hitem] (17\pGameL, 7\pGameL) rectangle +(1\pGameL, 11\pGameL);
\path[hitem] (18\pGameL, 7\pGameL) rectangle +(2\pGameL, 7\pGameL);
\path[hitem] (21\pGameL, 9\pGameL) rectangle +(2\pGameL, 13\pGameL);
\path[hitem] (25\pGameL, 8\pGameL) rectangle +(1\pGameL, 16\pGameL);
\path[hitem] (29\pGameL, 2\pGameL) rectangle +(1\pGameL, 8\pGameL);
\path[hitem] (27\pGameL, 2\pGameL) rectangle +(2\pGameL, 8\pGameL);
\path[hitem] (28\pGameL, 10\pGameL) rectangle +(1\pGameL, 9\pGameL);

\path[tcell] (0\pGameL, 19\pGameL) rectangle +(5\pGameL, 11\pGameL);
\path[tcell] (5\pGameL, 19\pGameL) rectangle +(5\pGameL, 11\pGameL);
\path[tcell] (10\pGameL, 17\pGameL) rectangle +(5\pGameL, 13\pGameL);
\path[tcell] (0\pGameL, 0\pGameL) rectangle +(5\pGameL, 12\pGameL);
\path[tcell] (5\pGameL, 0\pGameL) rectangle +(5\pGameL, 10\pGameL);
\path[tcell] (10\pGameL, 4\pGameL) rectangle +(5\pGameL, 6\pGameL);
\path[tcell] (15\pGameL, 7\pGameL) rectangle +(5\pGameL, 23\pGameL);
\path[tcell] (20\pGameL, 7\pGameL) rectangle +(5\pGameL, 18\pGameL);
\path[tcell] (25\pGameL, 2\pGameL) rectangle +(5\pGameL, 23\pGameL);

\path[box2] (0\pGameL, 17\pGameL) rectangle +(10\pGameL, 2\pGameL) node[pos=0.5] {1};
\path[box2] (0\pGameL, 12\pGameL) rectangle +(15\pGameL, 5\pGameL) node[pos=0.5] {2};
\path[box2] (5\pGameL, 10\pGameL) rectangle +(10\pGameL, 2\pGameL) node[pos=0.5] {3};
\path[box2] (20\pGameL, 25\pGameL) rectangle +(10\pGameL, 5\pGameL) node[pos=0.5] {4};
\path[box2] (15\pGameL, 4\pGameL) rectangle +(10\pGameL, 3\pGameL) node[pos=0.5] {5};
\path[box2] (10\pGameL, 2\pGameL) rectangle +(15\pGameL, 2\pGameL) node[pos=0.5] {6};
\path[box2] (10\pGameL, 0\pGameL) rectangle +(20\pGameL, 2\pGameL) node[pos=0.5] {7};

\path[bin] (0\pGameL, 0\pGameL) rectangle (30\pGameL, 30\pGameL);
\end{tikzpicture}

\caption{Partition the space outside tall cells into rectangular regions by
extending the horizontal edges of tall cells (see \cref{thm:empty-to-rects}).
Each rectangular region containing a wide item is called a \emph{box}.
There are 7 boxes in this figure, which are shaded gray.}
\label{fig:compartmentalize:boxes}
\end{subfigure}
\hfill
\begin{subfigure}[t]{0.3\textwidth}
\centering
\ifcsname pGameL\endcsname\else\newlength{\pGameL}\fi
\setlength{\pGameL}{0.15cm}
\tikzset{bin/.style={draw,thick}}
\tikzset{binGrid/.style={draw,step=1\pGameL,{black!20}}}
\tikzset{witem/.style={draw={black!40},fill=shadedLightGreen}}
\tikzset{hitem/.style={draw={black!30},fill=myLightBlue}}
\tikzset{tcell/.style={draw}}
\tikzset{box/.style={fill={black!10}}}
\tikzset{box2/.style={draw,semithick}}
\begin{tikzpicture}

\path[box] (0\pGameL, 12\pGameL) rectangle +(15\pGameL, 4\pGameL);
\path[box] (5\pGameL, 10\pGameL) rectangle +(10\pGameL, 2\pGameL);
\path[box] (20\pGameL, 26\pGameL) rectangle +(10\pGameL, 4\pGameL);
\path[box] (15\pGameL, 4\pGameL) rectangle +(10\pGameL, 2\pGameL);
\path[box] (10\pGameL, 2\pGameL) rectangle +(15\pGameL, 2\pGameL);
\path[box] (10\pGameL, 0\pGameL) rectangle +(20\pGameL, 2\pGameL);

\path[witem] (0\pGameL, 12\pGameL) rectangle +(10\pGameL, 2\pGameL);
\path[witem] (0\pGameL, 15\pGameL) rectangle +(10\pGameL, 1\pGameL);
\path[witem] (0\pGameL, 14\pGameL) rectangle +(15\pGameL, 1\pGameL);
\path[witem] (15\pGameL, 5\pGameL) rectangle +(10\pGameL, 1\pGameL);
\path[witem] (5\pGameL, 10\pGameL) rectangle +(10\pGameL, 2\pGameL);
\path[witem] (15\pGameL, 0\pGameL) rectangle +(10\pGameL, 1\pGameL);
\path[witem] (15\pGameL, 1\pGameL) rectangle +(15\pGameL, 1\pGameL);
\path[witem] (10\pGameL, 2\pGameL) rectangle +(10\pGameL, 2\pGameL);
\path[witem] (20\pGameL, 26\pGameL) rectangle +(10\pGameL, 1\pGameL);

\path[hitem] (0\pGameL, 22\pGameL) rectangle +(1\pGameL, 8\pGameL);
\path[hitem] (2\pGameL, 20\pGameL) rectangle +(2\pGameL, 8\pGameL);
\path[hitem] (7\pGameL, 21\pGameL) rectangle +(1\pGameL, 9\pGameL);
\path[hitem] (8\pGameL, 20\pGameL) rectangle +(1\pGameL, 9\pGameL);
\path[hitem] (10\pGameL, 22\pGameL) rectangle +(2\pGameL, 8\pGameL);
\path[hitem] (12\pGameL, 17\pGameL) rectangle +(2\pGameL, 8\pGameL);
\path[hitem] (14\pGameL, 20\pGameL) rectangle +(2\pGameL, 10\pGameL);
\path[hitem] (17\pGameL, 19\pGameL) rectangle +(2\pGameL, 11\pGameL);
\path[hitem] (0\pGameL, 1\pGameL) rectangle +(1\pGameL, 8\pGameL);
\path[hitem] (2\pGameL, 0\pGameL) rectangle +(1\pGameL, 12\pGameL);
\path[hitem] (4\pGameL, 1\pGameL) rectangle +(2\pGameL, 8\pGameL);
\path[hitem] (7\pGameL, 0\pGameL) rectangle +(2\pGameL, 9\pGameL);
\path[hitem] (11\pGameL, 4\pGameL) rectangle +(1\pGameL, 6\pGameL);
\path[hitem] (26\pGameL, 10\pGameL) rectangle +(2\pGameL, 9\pGameL);
\path[hitem] (16\pGameL, 12\pGameL) rectangle +(1\pGameL, 12\pGameL);
\path[hitem] (17\pGameL, 7\pGameL) rectangle +(1\pGameL, 11\pGameL);
\path[hitem] (18\pGameL, 7\pGameL) rectangle +(2\pGameL, 7\pGameL);
\path[hitem] (21\pGameL, 9\pGameL) rectangle +(2\pGameL, 13\pGameL);
\path[hitem] (25\pGameL, 8\pGameL) rectangle +(1\pGameL, 16\pGameL);
\path[hitem] (29\pGameL, 2\pGameL) rectangle +(1\pGameL, 8\pGameL);
\path[hitem] (27\pGameL, 2\pGameL) rectangle +(2\pGameL, 8\pGameL);
\path[hitem] (28\pGameL, 10\pGameL) rectangle +(1\pGameL, 9\pGameL);

\path[tcell] (0\pGameL, 19\pGameL) rectangle +(5\pGameL, 11\pGameL);
\path[tcell] (5\pGameL, 19\pGameL) rectangle +(5\pGameL, 11\pGameL);
\path[tcell] (10\pGameL, 17\pGameL) rectangle +(5\pGameL, 13\pGameL);
\path[tcell] (0\pGameL, 0\pGameL) rectangle +(5\pGameL, 12\pGameL);
\path[tcell] (5\pGameL, 0\pGameL) rectangle +(5\pGameL, 10\pGameL);
\path[tcell] (10\pGameL, 4\pGameL) rectangle +(5\pGameL, 6\pGameL);
\path[tcell] (15\pGameL, 7\pGameL) rectangle +(5\pGameL, 23\pGameL);
\path[tcell] (20\pGameL, 7\pGameL) rectangle +(5\pGameL, 18\pGameL);
\path[tcell] (25\pGameL, 2\pGameL) rectangle +(5\pGameL, 23\pGameL);

\path[box2] (0\pGameL, 12\pGameL) rectangle +(15\pGameL, 4\pGameL) node[pos=0.5] {2};
\path[box2] (5\pGameL, 10\pGameL) rectangle +(10\pGameL, 2\pGameL) node[pos=0.5] {3};
\path[box2] (20\pGameL, 26\pGameL) rectangle +(10\pGameL, 4\pGameL) node[pos=0.5] {4};
\path[box2] (15\pGameL, 4\pGameL) rectangle +(10\pGameL, 2\pGameL) node[pos=0.5] {5};
\path[box2] (10\pGameL, 2\pGameL) rectangle +(15\pGameL, 2\pGameL) node[pos=0.5] {6};
\path[box2] (10\pGameL, 0\pGameL) rectangle +(20\pGameL, 2\pGameL) node[pos=0.5] {7};

\path[bin] (0\pGameL, 0\pGameL) rectangle (30\pGameL, 30\pGameL);
\end{tikzpicture}

\caption{For each box, discard some items and shift horizontal edges
to make their $y$-coordinates multiples of $\epsCont$.
Boxes that continue to contain a wide item are now wide compartments.}
\label{fig:compartmentalize:boxes2}
\end{subfigure}
\hfill
\begin{subfigure}[t]{0.3\textwidth}
\centering
\ifcsname pGameL\endcsname\else\newlength{\pGameL}\fi
\setlength{\pGameL}{0.15cm}
\tikzset{bin/.style={draw,thick}}
\tikzset{binGrid/.style={draw,step=1\pGameL,{black!20}}}
\tikzset{witem/.style={draw={black!30},fill=myLightGreen}}
\tikzset{hitem/.style={draw={black!30},fill=myLightBlue}}
\tikzset{tcell/.style={draw,semithick}}
\tikzset{box2/.style={draw}}
\begin{tikzpicture}

\path[witem] (0\pGameL, 12\pGameL) rectangle +(10\pGameL, 2\pGameL);
\path[witem] (0\pGameL, 15\pGameL) rectangle +(10\pGameL, 1\pGameL);
\path[witem] (0\pGameL, 14\pGameL) rectangle +(15\pGameL, 1\pGameL);
\path[witem] (15\pGameL, 5\pGameL) rectangle +(10\pGameL, 1\pGameL);
\path[witem] (5\pGameL, 10\pGameL) rectangle +(10\pGameL, 2\pGameL);
\path[witem] (15\pGameL, 0\pGameL) rectangle +(10\pGameL, 1\pGameL);
\path[witem] (15\pGameL, 1\pGameL) rectangle +(15\pGameL, 1\pGameL);
\path[witem] (10\pGameL, 2\pGameL) rectangle +(10\pGameL, 2\pGameL);
\path[witem] (20\pGameL, 26\pGameL) rectangle +(10\pGameL, 1\pGameL);

\path[hitem] (0\pGameL, 22\pGameL) rectangle +(1\pGameL, 8\pGameL);
\path[hitem] (2\pGameL, 20\pGameL) rectangle +(2\pGameL, 8\pGameL);
\path[hitem] (7\pGameL, 21\pGameL) rectangle +(1\pGameL, 9\pGameL);
\path[hitem] (8\pGameL, 20\pGameL) rectangle +(1\pGameL, 9\pGameL);
\path[hitem] (10\pGameL, 22\pGameL) rectangle +(2\pGameL, 8\pGameL);
\path[hitem] (12\pGameL, 17\pGameL) rectangle +(2\pGameL, 8\pGameL);
\path[hitem] (14\pGameL, 20\pGameL) rectangle +(2\pGameL, 10\pGameL);
\path[hitem] (17\pGameL, 19\pGameL) rectangle +(2\pGameL, 11\pGameL);
\path[hitem] (0\pGameL, 1\pGameL) rectangle +(1\pGameL, 8\pGameL);
\path[hitem] (2\pGameL, 0\pGameL) rectangle +(1\pGameL, 12\pGameL);
\path[hitem] (4\pGameL, 1\pGameL) rectangle +(2\pGameL, 8\pGameL);
\path[hitem] (7\pGameL, 0\pGameL) rectangle +(2\pGameL, 9\pGameL);
\path[hitem] (11\pGameL, 4\pGameL) rectangle +(1\pGameL, 6\pGameL);
\path[hitem] (26\pGameL, 10\pGameL) rectangle +(2\pGameL, 9\pGameL);
\path[hitem] (16\pGameL, 12\pGameL) rectangle +(1\pGameL, 12\pGameL);
\path[hitem] (17\pGameL, 7\pGameL) rectangle +(1\pGameL, 11\pGameL);
\path[hitem] (18\pGameL, 7\pGameL) rectangle +(2\pGameL, 7\pGameL);
\path[hitem] (21\pGameL, 9\pGameL) rectangle +(2\pGameL, 13\pGameL);
\path[hitem] (25\pGameL, 8\pGameL) rectangle +(1\pGameL, 16\pGameL);
\path[hitem] (29\pGameL, 2\pGameL) rectangle +(1\pGameL, 8\pGameL);
\path[hitem] (27\pGameL, 2\pGameL) rectangle +(2\pGameL, 8\pGameL);
\path[hitem] (28\pGameL, 10\pGameL) rectangle +(1\pGameL, 9\pGameL);

\path[tcell] (0\pGameL, 16\pGameL) rectangle +(5\pGameL, 14\pGameL) node[pos=0.5] {1};
\path[tcell] (5\pGameL, 16\pGameL) rectangle +(5\pGameL, 14\pGameL) node[pos=0.5] {2};
\path[tcell] (10\pGameL, 16\pGameL) rectangle +(5\pGameL, 14\pGameL) node[pos=0.5] {3};
\path[tcell] (0\pGameL, 0\pGameL) rectangle +(5\pGameL, 12\pGameL) node[pos=0.5] {4};
\path[tcell] (5\pGameL, 0\pGameL) rectangle +(5\pGameL, 10\pGameL) node[pos=0.5] {5};
\path[tcell] (10\pGameL, 4\pGameL) rectangle +(5\pGameL, 6\pGameL) node[pos=0.5] {6};
\path[tcell] (15\pGameL, 6\pGameL) rectangle +(5\pGameL, 24\pGameL) node[pos=0.5] {7};
\path[tcell] (20\pGameL, 6\pGameL) rectangle +(5\pGameL, 20\pGameL) node[pos=0.5] {8};
\path[tcell] (25\pGameL, 2\pGameL) rectangle +(5\pGameL, 24\pGameL) node[pos=0.5] {9};

\path[box2] (0\pGameL, 12\pGameL) rectangle +(15\pGameL, 4\pGameL);
\path[box2] (5\pGameL, 10\pGameL) rectangle +(10\pGameL, 2\pGameL);
\path[box2] (20\pGameL, 26\pGameL) rectangle +(10\pGameL, 4\pGameL);
\path[box2] (15\pGameL, 4\pGameL) rectangle +(10\pGameL, 2\pGameL);
\path[box2] (10\pGameL, 2\pGameL) rectangle +(15\pGameL, 2\pGameL);
\path[box2] (10\pGameL, 0\pGameL) rectangle +(20\pGameL, 2\pGameL);

\path[bin] (0\pGameL, 0\pGameL) rectangle (30\pGameL, 30\pGameL);
\end{tikzpicture}

\caption[Obtaining tall compartments]%
{Wide compartments divide each column into rectangular regions.
Each such region containing a tall item is a tall compartment.
There are 9 tall compartments in this figure.}
\label{fig:compartmentalize:final}
\end{subfigure}
\caption{Obtaining compartments}
\label{fig:compartmentalize-2}
\end{figure}

For each box $i$, slice and discard some items from the bottom of the box
and increase $y_1(i)$ so that it becomes a multiple of $\epsCont$.
Then slice and discard some items from the top of the box
and reduce $y_2(i)$ so that it becomes a multiple of $\epsCont$.
The total area of items discarded is less than $2\epsCont$.
If $i$ continues to contain a wide item, it becomes a wide compartment.
Now all wide items belong to some wide compartment
(see \cref{fig:compartmentalize:boxes2}).

Each column has 0 or more wide compartments crossing it.
These wide compartments divide the column into rectangular regions.
Each region that contains a tall item is a tall compartment
(see \cref{fig:compartmentalize:final}).

Therefore, by removing wide and small items of area less than
$6|\Tcal|\epsCont/\epsLarge \le \eps$, we get a compartmental packing of items
where there are at most $(1/\epsLarge-1)|\Tcal|$ tall compartments
and at most $3(1/\epsLarge-1)|\Tcal| + 1$ wide compartments.
\end{proof}

\rthmStruct*
\begin{proof}
Consider a fractional packing of $\Itild$ into $m \defeq \fopt(\Itild)$ bins.
By \cref{thm:disc-hor-pos,thm:thin-bp:compartmentalize}, in each bin,
we can discard items of area at most $2\eps$ from the bin
and get a compartmental packing of the remaining items.

Let $X$ be the set of wide and small discarded items
and let $Y$ be the set of tall discarded items.
For each item $i \in X$, if $w(i) \le 1/2$, slice it using a horizontal cut in the middle
and place the pieces horizontally next to each other to get a new item
of width $2w(i)$ and height $h(i)/2$. Repeat until $w(i) > 1/2$.
Now pack the items in bins by stacking them one-over-the-other
so that for each item $i \in X$, $x_1(i) = 0$.
This will require less than $2a(X) + 1$ bins,
and the packing will be compartmental.

Similarly, we can get a compartmental packing of $Y$ into $2a(Y) + 1$ bins.
Since $a(X \cup Y) < 2\eps m$, we will require less than $4\eps m + 2$ bins.
Therefore, the total number of compartmental bins used to pack $\Itild$
is less than $(1 + 4\eps)m + 2$.
\end{proof}

\subsection{Enumerating Packing of Compartments}
\label{sec:enum-configs}

We will compute the optimal fractional compartmental packing of $\Itild$ in two steps.
First, for each bin, we will guess the compartments in the bin.
Each such packing of compartments into bins is called a \emph{configuration}.
Then we will fractionally pack the items into the compartments.

There can be at most $\nW \defeq 3(1/\epsLarge-1)|\Tcal| + 1$ wide compartments in a bin.
Each wide compartment can have $(1/\epsCont)^2$ $y$-coordinates
of the top and bottom edges and at most $|\Tcal|^2/2$ $x$-coordinates
of the left and right edges, where $\epsCont \defeq \eps\epsLarge/6|\Tcal|$.
The rest of the space is for tall compartments.
Therefore, the number of configurations is at most
\[ \nC \defeq \left((1/\epsCont)^2|\Tcal|^2/2\right)^{\nW}
    \le \left(\frac{3|\Tcal|^2}{\eps\epsLarge}\right)^{6|\Tcal|/\epsLarge}
    \le \left(1+\frac{1}{\eps\epsLarge}\right)^{
        \left(1+\frac{1}{\eps\epsLarge}\right)^{2/\epsLarge + 1}}. \]
Since each configuration can have at most $n$ bins, the number of combinations
of configurations is at most $(n+1)^{\nC}$.

Therefore, we can iterate over all possible bin packings of empty compartments in $O(n^{\nC})$ time.
Let $\iterPackings(\Itild)$ be an algorithm for this, i.e.,
$\iterPackings(\Itild)$ outputs the set of all possible bin packings of empty compartments into
at least $\smallceil{a(\Itild)}$ and at most $n$ bins,
where $n$ is the number of items in $\Itild$.

\subsection{Packing Items Into Compartments}
\label{sec:feas-lp}

For each bin packing of empty compartments, we will try to
fractionally pack the items into the bins.
Formally, let $P$ be a packing of empty compartments into bins.
We will create a feasibility linear program, called $\FP(\Itild, P)$,
that is feasible iff wide and tall items in $\Itild$ can be packed
into the compartments in $P$.
If $\FP(\Itild, P)$ is feasible, then small items can also be
fractionally packed since $P$ contains at least $a(\Itild)$ bins.

Let $w_1', w_2', \ldots, w_p'$ be the distinct widths of wide compartments in $P$.
Let $U_j$ be the set of wide compartments in $P$ having width $w_j'$.
Let $h(U_j)$ be the sum of heights of the compartments in $U_j$.
By \cref{defn:thin-bp:compartmental}, we know that $p \le |\Tcal|^2/2$.
Let $w_1, w_2, \ldots, w_r$ be the distinct widths of items in $\Wtild$
(recall that $\Wtild$ is the set of wide items in $\Itild$).
Let $\Wtild_j$ be the items in $\Wtild$ having width $w_j$.
Let $h(\Wtild_j)$ be the sum of heights of all items in $\Wtild_j$.
By \cref{thm:lingroup-n}, we get $r \le 1/\eps\epsLarge$.

Let $C \defeq [C_0, C_1, \ldots, C_r]$ be a vector, where
$C_0 \in [p]$ and $C_j \in \mathbb{Z}_{\ge 0}$ for $j \in [r]$.
$C$ is called a \emph{wide configuration} iff
$w(C) \defeq \sum_{j=1}^r C_jw_j \le w_{C_0}'$.
Intuitively, a wide configuration $C$ represents a set of wide items that can be placed
side-by-side into a compartment of width $w_{C_0}'$.
Let $\Ccal$ be the set of all wide configurations.
Then $|\Ccal| \le p/\epsLarge^r$, which is a constant.
Let $\Ccal_j \defeq \{C \in \Ccal: C_0 = j\}$.

To pack $\Wtild$ into wide compartments,
we must determine the height of each configuration.
Let $x \in \mathbb{R}_{\ge 0}^{|\Ccal|}$ be a vector where
$x_C$ denotes the height of configuration $C$.
Then $\Wtild$ can be packed into wide compartments according to $x$ iff
$x$ is a feasible solution the following feasibility linear program,
named $\FP_W(\Itild, P)$:
\[ \begin{array}{*4{>{\displaystyle}l}}
\sum_{C \in \Ccal} C_jx_C &\ge h(\Wtild_j)
    & \forall j \in [r]
    & \qquad\textrm{($\Wtild_j$ should be covered)}
\\[1.75em] \sum_{C \in \Ccal \textrm{ and } C_0 = j} x_C &\le h(U_j)
    & \forall j \in [p]
    & \qquad\textrm{($\Ccal_j$ should fit in $U_j$)}
\\[1.75em] x_C &\ge 0 & \forall C \in \Ccal
\end{array} \]

Let $x^*$ be an extreme point solution to $\FP_W(\Itild, P)$
(if $\FP_W(\Itild, P)$ is feasible).
By Rank Lemma\footnote{\rankLemmaNote}, at most $p+r$ entries of $x^*$ are non-zero.
Since the number of variables and constraints is constant,
$x^*$ can be computed in constant time.

Let $\Htild$ be the set of tall items in $\Itild$.
Items in $\Htild$ have at most $1/\eps\epsLarge$ distinct heights.
Let there be $q$ distinct heights of tall compartments in $P$.
By \cref{defn:thin-bp:compartmental}, we know that $q \le 1/\epsCont = 6|\Tcal|/\eps\epsLarge$.
We can similarly define \emph{tall configurations} and we can similarly define
a feasibility linear program for tall items, named $\FP_H(\Itild, P)$.
$\Htild$ can be packed into tall compartments in $P$ iff $\FP_H(\Itild, P)$ is feasible.
Let $y^*$ be an extreme point solution to $\FP_H(\Itild, P)$.
Then $y^*$ can be computed in constant time and
$y^*$ has at most $q + 1/\eps\epsLarge$ positive entries.

Therefore, $\Itild$ can be packed into $P$ iff the feasibility linear program
$\FP(\Itild, P) \defeq \FP_W(\Itild, P) \wedge \FP_H(\Itild, P)$ is feasible.

The solution $(x^*, y^*)$ shows us how to split each compartment into \emph{shelves},
where each shelf corresponds to a configuration $C$
and the shelf can be split into $C_j$ \emph{containers} of width $w_j$
and one container of width $w_{C_0}' - w(C)$.
Let there be $m$ bins in $P$. After splitting the configurations across compartments,
we get at most $p + q + 2/\eps\epsLarge + m(\nW + \nH)$ shelves.

\subsection{Converting a Fractional Packing to a Non-Fractional Packing}
\label{sec:greedy-cont}

Let there be $m$ bins in a packing $P$ of empty compartments into bins.
Suppose it is possible to pack $\Itild$ into $P$.
Let $x^*$ and $y^*$ be extreme-point solutions to
$\FP_W(\Itild, P)$ and $\FP_H(\Itild, P)$, respectively.
This gives us a fractional compartmental packing of $\Itild$ into $m$ bins.
We will now show how to convert this to a non-fractional compartmental packing
by removing some items of small total area.
Formally, we give an algorithm called $\greedyCPack(\Itild, P, x^*, y^*)$.
It returns a pair $(Q, D)$, where $Q$ is a (non-fractional) compartmental
bin packing of items $\Itild - D$, where the compartments in the bin are as per $P$.
$D$ is called the set of discarded items, and we will prove that $a(D)$ is small.

For a configuration $C$ in a wide compartment, there is a container
of width $w_{C_0}' - w(C)$ available for packing small items.
Hence, there are $p + q + 2/\eps\epsLarge + m(\nW + \nH)$ containers available
inside compartments for packing small items.
By \cref{thm:empty-to-rects}, we can partition the space outside compartments into
at most $m(3(\nW + \nH) + 1)$ containers.
Therefore, the total number of containers available for packing small items is at most
\[ m_S \defeq (p + q + 2/\eps\epsLarge) + m(4(\nW + \nH) + 1)
\le \left(\frac{|\Tcal|^2}{2} + \frac{6|\Tcal|}{\eps\epsLarge}
    + \frac{2}{\eps\epsLarge}\right) + \frac{16|\Tcal|}{\epsLarge} m. \]

Greedily assign small items to small containers, i.e., keep assigning small items
to a container till the area of items assigned to it is at least
the area of the container, and then resume from the next container.
Each small item will get assigned to some container.
For each container $C$, pack the largest possible prefix of the assigned items using
the Next-Fit Decreasing Height (NFDH) algorithm.
By \cref{thm:nfdh-small}, the area of unpacked items would be
less than $\epsSmall + \delta + \epsSmall\delta$. Summing over all containers,
we get that the unpacked area is less than
$(\epsSmall + \delta + \epsSmall\delta)m_S \le 3\epsSmall m_S$.

For each $j$, greedily assign wide items from $\Wtild_j$ to containers of width $w_j$,
i.e., keep assigning items till the height of items exceeds the height of the container.
Each wide item will get assigned to some container.
Then discard the last item from each container.
For each shelf in a wide compartment having configuration $C$,
the total area of items we discard is at most $\delta w(C)$.
Similarly, we can discard tall items of area at most $\delta h(C)$
from each shelf in a tall compartment having configuration $C$.

Hence, across all configurations, we discard wide and tall items of area at most
\[ \delta((p + q + 2/\eps\epsLarge) + m(\nW + \nH))
\le \delta\left(\frac{|\Tcal|^2}{2} + \frac{6|\Tcal|}{\eps\epsLarge}
    + \frac{2}{\eps\epsLarge}\right) + \frac{4\delta|\Tcal|}{\epsLarge}m. \]

Therefore, for $(Q, D) \defeq \greedyCPack(\Itild, P, x^*, y^*)$, we get
\begin{equation}
\label{eqn:greedy-discard-area}
a(D) < \frac{52|\Tcal|\epsSmall}{\epsLarge} m
    + 4\epsSmall\left(\frac{|\Tcal|^2}{2} + \frac{6|\Tcal|}{\eps\epsLarge}
        + \frac{2}{\eps\epsLarge}\right),
\end{equation}
where $m$ is the number of bins used by $P$.

\subsection{The \texorpdfstring{$\thinCPack$}{skewedCPack} Algorithm}
\label{sec:thinCPack}

We now summarize the $\thinCPack$ algorithm for bin packing $\delta$-\thin{} items $I$
(see \cref{algo:thinCPack} for a more precise description).
First, use $\round$ on $I$, i.e., let $(\Itild, \Imed) \defeq \round(I)$.
Then enumerate all packings $P$ of compartments into bins as per \cref{sec:enum-configs}.
For each packing $P$, check if $\Itild$ can be fractionally packed into $P$
by solving the feasibility linear program (see \cref{sec:feas-lp}).
If yes, then use a solution to the feasibility linear program to
compute a (non-fractional) compartmental packing of $\Itild - D$
using $\greedyCPack$ (see \cref{sec:greedy-cont}),
where $D$ is the set of items discarded by $\greedyCPack$.
Then pack $\Imed \cup D$ into bins using the Next-Fit Decreasing Height (NFDH) algorithm.
Output the best bin packing of $I$ across all choices of $P$.

\begin{algorithm}[htb]
\caption{$\thinCPack_{\eps}(I)$: Packs a set $I$ of $\delta$-\thin{} rectangular items
into bins without rotating the items.}
\label{algo:thinCPack}
\begin{algorithmic}[1]
\State $(\Itild, \Imed) = \round_{\eps}(I)$.
\State Initialize $\Qbest$ to \Null.
\For{$P \in \iterPackings(\Itild)$}
    \Comment{$\iterPackings$ is defined in \cref{sec:enum-configs}.}
    \State $x^* = \opt(\FP_W(\Itild, P))$.
    \Comment{$\FP_W$ and $\FP_H$ are defined in \cref{sec:feas-lp}.}
    \LineComment{If $\FP_W(\Itild, P)$ is feasible,
        $x^*$ is an extreme-point solution to $\FP_W(\Itild, P)$.}
    \LineComment{If $\FP_W(\Itild, P)$ is infeasible, $x^*$ is \Null.}
    \State $y^* = \opt(\FP_H(\Itild, P))$.
    \If{$x^* \neq \Null$ and $y^* \neq \Null$}
        \Comment{if $\Itild$ can be packed into $P$}
        \State $(Q, D) = \greedyCPack(\Itild, P, x^*, y^*)$.
        \Comment{$\greedyCPack$ is defined in \cref{sec:greedy-cont}.}
        \State $Q_D = \operatorname{NFDH}(D \cup \Imed)$.
        \If{$Q \cup Q_D$ uses less bins than $\Qbest$}
            \State $\Qbest = Q \cup Q_D$.
        \EndIf
    \EndIf
\EndFor
\State \Return $\Qbest$
\end{algorithmic}
\end{algorithm}

Recall the function $f$ from \cref{sec:thin-bp:remmed}.
Since $\epsSmall \defeq f(\epsLarge)$, we get
\begin{equation}
\label{eqn:epsSmall}
\epsSmall = f(\epsLarge)
= \frac{\eps\epsLarge}{104(1+1/\eps\epsLarge)^{2/\epsLarge-2}}
\le \frac{\eps\epsLarge}{104|\Tcal|}.
\end{equation}
The last inequality follows from the fact that
$|\Tcal| \le (1+1/\eps\epsLarge)^{2/\epsLarge-2}$.

\rthmThinCPack*
\begin{proof}
In an optimal fractional compartmental bin packing of $\Itild$,
let $P^*$ be the corresponding packing of empty compartments into bins.
Hence, $P^*$ contains $m \defeq \fcopt(\Itild)$ bins.
Since $\iterPackings(\Itild)$ iterates over all packings of compartments into bins,
$P^* \in \iterPackings(\Itild)$.
Since wide and tall items in $\Itild$ can be packed into the compartments of $P^*$,
we get that $x^*$ and $y^*$ are not \Null.
By \cref{thm:nfdh-wide-tall}, the number of bins used by NFDH to pack
$\Imed \cup D$ is less than $2a(\Imed \cup D)/(1-\delta) + 3 + 1/(1-\delta)$.
Therefore, the number of bins used by $\thinCPack(I)$ is less than
\begin{align*}
& m + \frac{2a(\Imed \cup D)}{1-\delta} + 3 + \frac{1}{1-\delta}
\\ &< m + \frac{2\eps}{1-\delta}a(I)
    + \frac{2\epsSmall}{1-\delta}\left(\frac{52|\Tcal|}{\epsLarge} m
        + 4\left(\frac{|\Tcal|^2}{2} + \frac{6|\Tcal|+2}{\eps\epsLarge}\right)\right)
    + 3 + \frac{1}{1-\delta}
    \tag{by \cref{eqn:greedy-discard-area} and $a(\Imed) \le \eps a(I)$}
\\ &= \left(1 + \frac{104\epsSmall|\Tcal|}{\epsLarge(1-\delta)}\right)m
    + \frac{2\eps}{1-\delta}a(I)
    + 3 + \frac{1}{1-\delta} + \frac{8\epsSmall}{1-\delta}\left(
        \frac{|\Tcal|^2}{2} + \frac{6|\Tcal|+2}{\eps\epsLarge}\right)
\\ &= \left(1 + \frac{\eps}{1-\delta}\right)m + \frac{2\eps}{1-\delta}a(I)
    + 3 + \frac{1}{13(1-\delta)}\left(\frac{\eps\epsLarge|\Tcal|}{2}
        + 19 + \frac{2}{|\Tcal|}\right).
    \tag{by \cref{eqn:epsSmall}}
\end{align*}
By \cref{thm:struct,thm:thin-bp:lingroup-opt-compare}, we get
\[ m = \fcopt(\Itild) < (1+4\eps)\fopt(\Itild) + 2
< (1+4\eps)(1+\eps)\opt(I) + 4 + 8\eps. \]
Therefore, the number of bins used by $\thinCPack(I)$ is less than
\begin{align*}
& \left((1+4\eps)(1+\eps)\left(1 + \frac{\eps}{1-\delta}\right)
    + \frac{2\eps}{1-\delta}\right)\opt(I)
\\ &\qquad + (4 + 8\eps)\left(1 + \frac{\eps}{1-\delta}\right) + 3
    + \frac{1}{13(1-\delta)}\left(\frac{\eps\epsLarge|\Tcal|}{2} + 19 + \frac{2}{|\Tcal|}\right)
\\ &\le (1+20\eps)\opt(I)
    + \frac{1}{13}\left(1 + \frac{1}{\eps\epsLarge}\right)^{2/\epsLarge - 2} + 23.
    \tag{since $\delta \le \epsLarge \le \eps \le 1/2$}
\end{align*}
\end{proof}

\end{document}